\documentclass[a4paper,12pt]{article}

\usepackage[english]{babel}
\usepackage[utf8]{inputenc}
\usepackage[T1]{fontenc}
\usepackage{etex}
\usepackage{pdfpages}
\usepackage{amsfonts}
\usepackage{amsmath}
\usepackage{authblk}
\usepackage{amssymb}
\usepackage{amsthm}
\usepackage{appendix}
\usepackage[pagebackref,colorlinks=true]{hyperref}
\usepackage{cases}
\usepackage{enumerate}
\usepackage{caption}
\usepackage[labelformat = simple]{subcaption}

\usepackage{setspace}
\usepackage{ulem}
\usepackage{soul}
\usepackage{enumitem}
\usepackage{scalerel}

\usepackage{cancel}
\usepackage{upgreek} 
\usepackage{fancyhdr}
\usepackage{float}
\usepackage{manfnt} 
\usepackage{mathrsfs}
\usepackage{minitoc}
\usepackage{accents}

\usepackage{rotating}

\usepackage{booktabs}
\usepackage{graphics}
\usepackage{graphicx}
\usepackage{hyperref}

\usepackage{lastpage}
\usepackage{latexsym}

\usepackage{color}
\usepackage{comment}

\usepackage{enumerate}

\usepackage{fancyhdr}
\usepackage[top=3cm, bottom=3cm, left=1.5cm , right=2cm]{geometry}
\newcommand{\smallerthansmall}{\fontsize{10.75}{11}\selectfont}

\newcommand{\setfolder}[1]{\def\currentfolder{#1/}}

\newtheorem{Theorem}{Theorem}[section]
\newtheorem{Corollary}[Theorem]{Corollary}

\newtheorem{Definition}[Theorem]{Definition}
\newtheorem{Proposition}[Theorem]{Proposition}

\numberwithin{equation}{section}

\def\A{{\mathcal A}}

\def\E{{\mathbb E}}

\def\P{{\mathbb P}}

\def\wL^*{{\widehat L^{(\mu^*, \sigma^*)}}}
\def\wL{{\widehat L}}

\newcommand{\rmi}{{\rm (i) $\>\>$}}

\newcommand{\rmii}{{\rm (ii) $\hspace{1.5mm}$}}

\def\bit{\begin{itemize}}
\def\eit{\end{itemize}}

\def\beqn{\begin{eqnarray}}
\def\eeqn{\end{eqnarray}}

\def\beq*{\begin{eqnarray*}}
\def\eeq*{\end{eqnarray*}}

\def\E{{\mathbb E}}

\def\P{{\mathbb{P}}}

\def\exp{{\text{exp}}}
\def\bit{\begin{itemize}}
\def\eit{\end{itemize}}

\def\bc{\begin{center}}
\def\ec{\end{center}}

\def\super { \end{document}}
\def\bcom{}
\def\edoc{\end{document}}

\usepackage[framemethod=tikz]{mdframed}
\usetikzlibrary{decorations.pathmorphing,calc}
\newcommand\wavydecor{%
    \draw[decoration={coil,aspect=0.1,segment length=5pt,amplitude=1.0pt},decorate,line width=1.5pt,black]
      (O|-P) -- (O);
}
\newmdenv[
hidealllines=true,
innerleftmargin=10pt,
innerrightmargin=0pt,
innertopmargin=0pt,
innerbottommargin=0pt,
leftmargin=-10pt,
skipabove=.5\baselineskip,
skipbelow=.5\baselineskip,
singleextra={\wavydecor},
firstextra={\wavydecor},
secondextra={\wavydecor},
middleextra={\wavydecor}
]{done}

\usepackage[linecolor=black,textwidth=25mm,textsize= footnotesize]{todonotes} 

{ \begin{list}%
	{$\bullet$}%
	{\setlength{\labelwidth}{10pt}%
	 \setlength{\leftmargin}{35pt}%
	 \setlength{\itemsep}{\parsep}{1pt}}}%
{ \end{list} }

\newcommand\blfootnote[1]{%
  \begingroup
  \renewcommand\thefootnote{}\footnote{#1}%
  \addtocounter{footnote}{-1}%
  \endgroup
}

\title{Optimal investment and Pension policy in Pay-As-You-Go systems under forward utility and ageing population.\blfootnote{This work was supported by the Europlace Institute of Finance and the  ANR project DREAMeS (ANR-21-CE46-0002). C. Hillairet, S. Kaakai and M. Mrad  acknowledge the support received
from the Research Chair ACTIONS under the aegis of the Risk Foundation, an
initiative by BNP Paribas Cardif and the Institute of Actuaries of France.}
}

\author[1]{Jennifer Alonso-Garcia}
\author[2]{Caroline Hillairet}
\author[3]{Sarah Kaakai}
\author[3]{Mohamed Mrad}
\affil[1]{\footnotesize{Université Libre de Bruxelles, CEPAR, Netspar}}
\affil[2]{\footnotesize{CREST, UMR CNRS 9194, ENSAE, Institut Polytechnique de Paris}}
\affil[3]{\footnotesize{LAGA, UMR CNRS 7539,  Université Sorbonne Paris Nord} }

\date{}

\begin{document}
 \maketitle

\abstract{  
This paper investigates optimal investment and pension policies in a Pay-As-You-Go (PAYG) system supplemented by a buffer fund used as an intergenerational risk‑sharing mechanism. The social planner’s preference criterion is represented by non-zero volatility forward Constant Relative Risk Aversion (CRRA) utilities, and explicitly accounts for both sustainability and adequacy constraints.   The optimal policies are characterized in closed form, and an in-depth analysis  of the impact of preference sensitivities on the pension scheme is conducted. A detailed numerical analysis is performed  to evaluate the sustainability and benefit adequacy of this hybrid PAYG–buffer‑fund arrangement under a range of demographic, financial, and macroeconomic scenarios.  
}

\bigskip


\date{}

\section{Introduction}
Many countries are facing declining birth rates combined with increasing life expectancy. According to the 2025 edition of Pensions at a Glance \cite{OECD25}, the old-age dependency ratio across the OECD is projected to rise from 33 people aged 65+ per 100 people aged 20-64 in 2025 to 52 in 2050, up from just 22 in 2000. This steady increase in the old-age dependency ratio (DR) challenges the fiscal sustainability of pension systems and urges ageing countries to undertake reforms aimed at reducing pension-related expenditures. In pure Pay-As-You-Go (PAYG) pension schemes, in which current pensions are financed by redistributing contributions paid by working participants, the obvious remedies: increasing contribution rates, reducing pension levels, or a combination of both, raise adequacy and acceptability concerns, as discussed in \cite{morsomme2025intergenerational} and \cite{alonso2018adequacy}.

One widely discussed reform  is the transition from Defined Benefit to Defined Contribution schemes, in which the pension indexation rate is linked to demographic and/or economic indicators. This  approach is already implemented in Sweden and Italy, among others (see \cite{godinez2016optimal}, \cite{alonso2019continuous}). Another widespread solution is to establish a Public Pension Reserve Fund (PPRF). More than two-thirds of OECD countries hold ring-fenced reserves in such funds, which collectively amounted to nearly USD~7 trillion across the OECD area at end-2024 \cite{OECD25}. The largest reserves are held by the United States and Japan, together accounting for over 60\% of the OECD total, while Korea, Canada, France, and Sweden also maintain substantial positions.\footnote{In relative terms, PPRF assets represent over 11\% of the combined GDP of OECD countries with reserves, with the highest ratios observed in Korea, Japan, Finland, Luxembourg, and Sweden, where reserves exceed 30\% of GDP in each case.} Naturally, alongside sustainability, pension systems should also provide adequate benefits for retirees as well as an acceptable level of fairness between generations, as underlined in \cite{alonso2018adequacy}.\footnote{We refer to \cite{OECD25} for an overview of key features of pension provision and detailed reforms undertaken by OECD countries.}\\

This paper addresses the sustainability, adequacy and fairness objectives for Defined Contribution public Pay-As-You-Go (PAYG) pension schemes. The social planner has access to a buffer fund, which allows economic, demographic, and financial risk sharing across generations \cite{alonso2025assessing}. Surplus funds generated during periods of prosperity are invested by the social planner in this buffer fund, which can then be drawn upon to cover future cash shortfalls. The paper determines and analyzes the optimal policy for such a mixed pension scheme, in which, the social planner additionally guarantees a minimum pension level corresponding to the ``pure PAYG'' pension benefits, thereby satisfying an adequacy constraint, while a sustainability constraint is imposed on the fund (for example, a no-borrowing constraint). The optimal policy aims to achieve the best balance between pensions paid to current retirees and the buffer fund, which represents reserves for future generations.

The determination of optimal policy is related to the extensive  literature on consumption/portfolio optimization with random endowment, such as \cite{he1991consumption}, \cite{duffie1993optimal}, \cite{el1998optimization}, \cite{karatzas2003optimal}, \cite{mostovyi2017optimal}, \cite{mostovyi2020optimal}, and the references therein, among many others.
This literature usually states the optimization problem in a
backward formulation, which has several drawbacks when considering
the framework of pensions schemes  in a stochastic environment, as outlined in \cite{ng2024optimal} and \cite{hillairet2024time}. Indeed the backward approach lacks  flexibility  to  incorporate  changes in the agents’ preferences, which may occur due to
the uncertain evolution of the environment variables. Moreover, in the pension context, defining a finite time horizon is particularly challenging, as optimal decisions tend to be highly sensitive to this choice,  a dependence that may prove problematic when the social planner seeks to ensure consistency both over time and across generations.

We therefore adopt a forward utility framework, which was  pioneered by Musiela and Zariphopoulou \cite{zarD,zar2011}. 
The framework of forward utilities leads to   adaptive and coherent optimal policies in the long run, as the preference criteria is  dynamically adjusted to the information flow. For a literature  review on forward utilities, we refer to the survey in \cite{ng2024optimal}. Forward utilities are used in many  application contexts, such as   in finance for the valuation of american options in  \cite{leung2012forward} or long term yield curve modeling \cite{el2022ramsey}, for  risk measures in \cite{zariphopoulou2010maturity} and \cite{chong2019ergodic},
 for equity-linked life insurance contracts in \cite{chong2019pricing}, and also
for robo-advising in \cite{liang2023predictable},  \cite{capponi2022personalized} and \cite{mradlearning}.
Forward utilities are all the more appropriate for pension fund management, 
due  to the long-term characteristics of pension schemes, 
together with a complex and uncertain environment, see 
\cite{bernard2016dynamic}, \cite{ng2024optimal}, and \cite{anthropelos2025time}. 
 In practice, the sensitivity of agents to changes in the environment is modeled through forward utilities with non-zero volatility. One of the key challenges  for practical applications is to understand how this sensitivity affects the optimal policies. In this paper, we provide an in-depth study of this issue which, to the best of our knowledge, has not been addressed in the existing  literature.

This paper is based on a stochastic partial differential equation (SPDE) representation of forward utilities, developed  in \cite{elkarouimrad2013} and further in \cite{el2018consistent} in an investment/consumption context.  We introduce a PAYG pension framework with buffer fund in an incomplete financial market with  correlated economic and financial risks. The  investment and pension policy is optimal with respect to the central planner’s forward criterion that incorporates a pure PAYG minimum guarantee and sustainability constraints. Building on the theoretical framework and results of Hillairet, Kaakai and Mrad (2024) \cite{hillairet2024time}, we obtain closed-form expressions for the  optimal policy in the case of Constant Relative Risk Aversion (CRRA) forward utilities criterion and a characterization of the viability of this pension system. In particular, we prove that the buffer depletion time does not depend on the initial value of the buffer, which implies that the buffer fund intervention remains feasible regardless of the initial level of pre-existing Public Pension Reserve Funds or, for new funds, public debt the government is willing to incur. Moreover, we show that when the social planner aggregates the individual preferences of retirees by assigning a constant weight  to each of them, the pension surplus does not depend on the dependency ratio. 
We provide a comprehensive numerical analysis of the pension scheme, calibrated to financial and economic data, and under different demographic scenarios (steady state, baby boom). In particular, we quantify the risk–benefit sharing mechanism across generations, depending on different preferences sensitivities (volatility of the utility) to the hedgeable and non-hedgeable risk. We also investigate the impact of the initial social planner preferences. \\

Following this introduction, Section \ref{sec1} describes the model of the  mixed PAYG pension scheme with a buffer fund, and the dynamics of the financial and economic processes. Section \ref{sec3} introduces the forward utilities of the buffer fund and the retirees. We  derive the optimal  policy and buffer depletion time under adequacy and sustainability constraints, and study the impact of preferences. Section \ref{sec:numerics} provides an extensive numerical and sensitivity analysis  based on realistic financial, economic dynamics and long-term
European ageing trends. Additional numerical results  are postponed in the Appendix.

\section{Pension Scheme and Financial framework}\label{sec1}
All stochastic processes are defined on a standard filtered probability space $(\Omega, \mathcal{F}, \mathbb{F}, \mathbb{P})$, where the filtration $\mathbb{F} = (\mathcal{F}_t)_{t \geq 0}$ is the natural filtration generated by a $4$-dimensional correlated Brownian motion $B$, assumed to be right-continuous and complete. 
In the following, the symbol $\cdot$ denotes the scalar product and $^{t}\!(.)$ denotes the transpose operator. \\
 This section introduces the model for the pension system characterized by the contributions of the workers and the pensions paid to retirees, and the buffer fund managed by the social planner to absorb demographic and economic shocks. We first 
describe the pension system, followed by the financial and wages dynamics and then the buffer fund dynamics.

\subsection{Pension system}

The social planner aims to determine the pension amount paid to retirees. 
 For tractability, we consider that earnings and pensions do no depend on the age of individuals (see \cite{hillairet2024time} for a more general setting). Population dynamics and average wages are treated as exogenous processes, which allows the model to accommodate a broad range of demographic scenarios including changes in fertility, longevity, and migration patterns. 

\paragraph{}We denote by $p_t$ the average individual pension amount and $N_t^r$ the number of retirees at a given time $t$. The total pension expenditure is therefore given by the process $(P_t)_{t\geq 0}$:
\begin{equation}\label{eq:Pension}
P_t = p_t \; N^r_t.
\end{equation}
These pensions are financed through contributions from workers, who contribute a fraction $\alpha$ (the contribution rate) of their wages to the system, yielding the contribution process $(C_t)_{t\geq 0}$:
\begin{equation}\label{eq:Contribution}
C_t = \alpha \; \mathfrak{e}_t \; N_t^w,
\end{equation}
where $\mathfrak{e}_t$ is the average wage and $N^w_t$ denotes the number of workers   at time $t$.\\
Adequacy is modeled  by imposing a lower bound $p^{\min}_t$ on the individual pension $p_t$:
\begin{equation}\label{eq:Pensionmin}
\forall t \geq 0, \quad p_t \geq p^{\min}_t \quad \mathbb{P}\text{-a.s.}
\end{equation}
A natural  choice is to set the minimum pension equal to the pure pay-as-you-go (PAYG) amount,
\begin{equation}\label{eq:sustainablepmin}
p^{\min}_t = \frac{C_t}{N_t^r},
\end{equation}
which ensures that a minimum pension level can always be financed through worker contributions alone, without incurring additional debt. This specification will serve as the benchmark case throughout the paper.

\subsection{Wages and financial market dynamics}\label{eq:wages_fin_market}
We consider  a financial market in which the social planner can invest. An important point is that the market is incomplete, as demographic and income risk in particular cannot be fully hedged through financial markets.  
\paragraph{Brownian motion} The dynamics of the financial and wages processes are driven by the  
 4-dimensional correlated Brownian motion 
 $B: ={}^{t} \!(B^S, B^\nu, B^r, B^{\mathfrak{e}})$.  
The correlation matrix of $B$ is denoted $\Gamma$:
\begin{align}\label{eq:Gamma_matrix}
	\Gamma=\begin{pmatrix}
		1 &\rho_{ S,\nu}&\rho_{ S,r} & \rho_{ S,\mathfrak{e}} \\
		\rho_{S, \mathfrak{\nu}} & 1 & \rho_{\nu,r} & \rho_{\nu,\mathfrak{e}}\\
		\rho_{S,r} & \rho_{\nu,r} & 1 & \rho_{r,\mathfrak{e}}\\
		\rho_{S,\mathfrak{e}}& \rho_{\nu,\mathfrak{e}} & \rho_{r,\mathfrak{e}} & 1
	\end{pmatrix} =  L  \;  ^{t}\!L, 
\end{align} 
where $L$ is  the lower triangular matrix given by Cholesky decomposition of $\Gamma$.

\paragraph{Financial market} The social planner can invest in a risk-free bank account and a  risky asset $S$ representing a market index. \\
The risky asset is assumed to follow a Heston Model: 
 \begin{equation}\label{eq:asset}
		d S_t=  S_t(\mu_t dt + \sqrt{\nu_t}  dB^{S}_t ),
\end{equation}
with  drift $(\mu_t)_{t\geq 0}$ being an $\mathbb F$-adapted process,  and with stochastic volatility dynamics
\begin{equation}\label{eq:heston}
d \nu_t =  \kappa (\bar \nu -\nu_t) dt + \sigma^{\nu} \sqrt{\nu_t}    dB^{\nu}_t, 
\end{equation}
where  $\kappa$, $\bar\nu$,  $\sigma^{\nu}$  are  non-negative constants satisfying  Feller's condition  $\sigma^{\nu} < \sqrt{2 \kappa \bar\nu}$. \\
The instantaneous risk-free interest rate $r$ is given by the Vasicek model
\begin{equation}\label{eq:Vasicek} 
d r_t =  a (b-r_t) dt +   \sigma^{r}_t dB^{r}_t, \quad \text{with } (a,b) \in \mathbb R^+_* \times \mathbb R
 \end{equation}   
 and $(\sigma^{r}_t)_{t\geq 0}$ an $\mathbb F$-adapted process.

\paragraph{Wages} Finally, the wages process $\mathfrak{e}$ is assumed to satisfy the following dynamics
 \begin{equation}\label{eq:salary} 
  d \mathfrak{e}_t=  \mathfrak{e}_t(\lambda_t dt  +  \sigma^{\mathfrak{e}}_t dB^{\mathfrak{e}}_t), 
 \end{equation}
with $(\lambda_t)_{t\geq 0}$,  $(\sigma^{\mathfrak{e}}_t)_{t\geq 0}$ $\mathbb F$-adapted processes.

\subsection{Buffer fund dynamics}
The social planner manages a buffer fund designed to absorb demographic and economic shocks. The fund receives contributions from workers, $C_t = \alpha \mathfrak{e}_t N_t^w$, and finances the pensions paid to retirees, $P_t = p_t N^r_t$. In addition, the social planner invests the fund's wealth in the financial market, allocating an amount $\phi_t$ to the risky asset at time $t$, with the remainder invested in the risk-free asset. The buffer fund $(F_t)_{t \geq 0}$ therefore follows self-financing dynamics with endowment  $C$ and consumption $P$:
\begin{align*}
d F_t  = (C_t-P_t) dt  +  (F_t-\phi_t) r_t dt +   \frac{\phi_t}{S_t} dS_t.\end{align*}
Let $\eta$ the risk premium process
\begin{equation}
\label{eq:Riskpremium}
\eta_t := \frac{\mu_t - r_t}{\sqrt{\nu_t}}, 
\end{equation}
and $\pi$ the rescaled strategy  defined by 
\begin{equation}
\label{eq:Pi}
\pi_t:= \phi_t \sqrt{\nu_t}. 
\end{equation}
The value $F^{\pi,  p}$ of the buffer fund induced by the investment-pension policy  $(\pi,p)$ satisfies
\begin{equation}\label{eq:SFfundbis}
d F^{\pi,  p}_t=  F^{\pi, p}_t  r_t dt + (C_t- p_t N_t^r ) dt +  {\pi_t}   ( dB_t^S + \eta_t dt).
\end{equation}
The social planner  chooses his investment strategy $\pi$ as well as the pension $p$ paid to each retiree under constraints. 
First, the adequacy constraint \eqref{eq:Pensionmin}  guarantees a minimum pension. Second, the sustainability of the pension scheme is ensured by imposing the buffer fund to not exceed a  maximum amount of debt $( - \mathfrak{K}_t)$, that is:  
\begin{equation}\label{eq:sustain} 
\forall t \geq 0, \quad  F_t^{\pi,  p} \geq \mathfrak{K}_t  \quad \mathbb \P\text{-a.s.}, 
\end{equation}
where $ (\mathfrak{K}_t)_{t \geq 0}$   is an $\mathbb{F}$-adapted process.

\begin{Definition}[Admissible policy]\mbox{}\\
Let $F_0> \mathfrak{K}_0$ be the initial value of the buffer fund.  An $\mathbb{F}$-adapted policy $(\pi,p)$ is said to be admissible if and only if
 \begin{itemize}
\item  $\int_0^t \left( | C_s - N_s^r p_s | + \| \pi_s\|^2 \right)ds <\infty,   \; \; \forall \; t\geq 0 \;  \; \P\text{-a.s.} $
\item $p_t \geq p^{\min}_t$ 
(adequacy) and   $ F_t^{\pi,p} \geq  \mathfrak{K}_t$ (sustainability) $\; \; \forall \; t\geq 0 \;  \; \P\text{-a.s.} $
 \end{itemize}
 The set of all  admissible policies   $(\pi,p)$ is  denoted $\A(F_0)$.
\end{Definition}
This investment-pension policy   $(\pi,p)$ optimizes a forward utility criterion that combines the buffer utility (that is the utility of future generations) and the aggregate retirees' utility.

\section{Optimal investment and pension policy}\label{sec3}

This section introduces the buffer fund and retirees' utility, as well as the consistency criteria required for the forward utility framework to be well defined. 

\subsection{Buffer fund and retirees' utility}

The social planner manages a trade-off between the pension paid to current retirees and the fund held in reserve for future generations. Accordingly, we introduce a preference process of the form
$$U(t,F^{\pi, p}_t)  + \int_0^t  V(s, p_s)  ds,$$
where $U$ and $V$ are \textit{random} utility functions $U(t, x)$  (strictly concave increasing  functions) whose   temporal evolution is  ``updated''  in accordance with the  information flow $(\mathcal F_t)_{t\geq 0}$, with initial deterministic utility function $u(x)=U(0,x)$.\\  
In the standard setting of expected utility optimization, the problem is tackled in a backward framework, that consists in an optimization problem specified  for  a fixed time horizon $T$
\vspace{-2mm}
\begin{equation*}
\sup_{(\pi, p) \in \mathcal A(F_0) } \mathbb E [ {U}(T, F^{\pi, p}_{T}) + \int_0^ T   {V}(s,  p_s) ds ]
\end{equation*}
with preferences at terminal time $T$ that are pre-determined. Nevertheless, this backward approach leads to several difficulties. The choice of the time horizon $T$ and the associated terminal preferences is somewhat arbitrary, especially over long horizons, and can significantly impact the optimal strategy, for instance by depleting the fund as $t$ approaches $T$. In contrast to standard individualized investment-consumption models, where $T$ is often interpreted as the death time or statutory retirement age, a parameter subject to political risk but unlikely to vary drastically, our framework is inherently multi-generational and involves multiple agents. Pre-specifying a finite time horizon therefore becomes problematic, as it privileges the utility of some agents over that of others.
In addition,  the backward framework can give rise to sub-optimality and time-inconsistency issues, particularly over long-term horizons,  as highlighted in \cite{el2018consistent} and \cite{ng2024optimal}, among others. \\
Therefore, we use the framework of forward utilities. In this setting,  there is no intrinsic time horizon, while  initial preferences are calibrated,  and future preferences are dynamically adapted to a learning set of admissible strategies that reflects the economic, financial and demographic changes. As a result, the derived optimal strategy is time-consistent over the long run, which is an essential feature when it comes to managing intergenerational risk sharing as it is the case here.\\
In order to remain close to the standards of economic literature, we use in this study  a dynamic version of  utilities having constant relative risk aversion (CRRA), also called forward power utilities.

\paragraph{Buffer fund utility} The buffer fund utility ${U(t, \cdot) }$ is a power  forward utility with constant relative risk aversion  $\theta \in ]0,1[ \cup  ]1,+\infty[ $. Both situations $\theta \in ]0,1[$ and $\theta > 1$ can be found in the literature: $\theta > 1$ is mainly used for portfolio and consumption optimization.%

The dynamic CRRA buffer fund utility takes into account the sustainability constraint $\mathfrak{K}$ as follows:
\begin{align}\label{eq:defUfund}
U(t,x) &:= Z^{u}_t \,  \frac{(x-\mathfrak{K}_t)^{1-\theta}}{ 1-\theta},
\end{align}\vspace{-3mm}
with \vspace{-3mm}
\begin{align} 
\label{eq:SDEZut}
d Z_t^u  &= Z_t^u( b_t dt + {\delta_t} \cdot d B_t),
\end{align}
 where $Z^{u}_t$ is a positive  stochastic process reflecting the random evolution of the time preferences. The vector process $\delta= {}^{t} \!(\delta^S, \delta^\nu, \delta^r, \delta^{\mathfrak{e}})$ represents the sensitivity of $Z^u$ with respect to the correlated Brownian motions driving the financial and economic processes, and plays an important role. More precisely,  the elements $\delta^S$, $\delta^\nu$, $\delta^r$ and $ \delta^{\mathfrak{e}}$ represents the sensitivity of $Z^u$ to the volatilities of the risky asset, its volatility, the  interest rate and wages, respectively, where the sign determines the correlation of $Z^u$ with the corresponding processes.

The sensitivities to the hedgeable and non-hedgeable risks also play an important role. Let $L$ be the square root matrix arising from the Cholesky decomposition of the correlation matrix   $\Gamma$, as introduced in \eqref{eq:Gamma_matrix}. Then,  the first component  ${({}^{t}\!L\delta_t)^S \in \mathbb{R}}$
of ${}^{t} L\delta_t \in \mathbb{R}^4$ represents the sensitivity of the buffer fund with respect to the hedgeable risk, while the vector of its  last three components $({}^{t}\!L\delta_t)^{\perp} \in \mathbb{R}^3$ represent the non-hedgeable risk, and only depends on ${}^{t} \!(\delta^\nu, \delta^r, \delta^\mathfrak{e})$.

\paragraph{Aggregate retirees' utility}
The social planner aggregates the individual preferences of retirees by assigning a weight $\omega_t$ to each retiree at time $t$. Formally, the preference process $v(t, \cdot)$ of one retiree is also specified as a power forward utility, with the same relative risk aversion $\theta$ for time-consistency reasons (see \cite{el2018consistent}):
     \begin{equation*}
v(t,p):= Z_t \frac{(p-p^{\min}_{t})^{1-\theta}}{1-\theta},
\end{equation*}
 where $Z$ can be interpreted as  a time preference process. 
The aggregate retirees' utility then becomes 
  \begin{equation}\label{eq:Vpension}
 V(t, p ) := N_t^r  \omega_t v(t, p)= N_t^r \omega_t  Z_t \frac{(p-p^{\min}_{t})^{1-\theta}}{1-\theta}.
\end{equation}

\subsection{Consistency and optimal processes}
The satisfaction provided by an admissible  policy $(\pi,p) \in \mathcal A(F_0)$  is measured  by the dynamic  criterion $U(t,F^{\pi,p}_t) + \int_0^t V(s,p_s) ds$,
that is assumed to satisfy a dynamic programming principle, also called
consistency given the constraints set $\mathcal A(F_0)$.

\begin{Definition}[Forward utility]\label{def:consistent system}
Let $( U,  V)$ be a dynamic utility  system with admissible policies set  $\A(F_0)$. The utility  system $( U,  V)$ is said to be {consistent}, if \\
\rmi For any admissible strategies   $(\pi,p) \in \mathcal A(F_0)$,  the preference process 
$(U(t,F^{\pi,p}_t) + \int_0^t V(s,p_s) ds)$    is a non-negative supermartingale.\\
\rmii There exists an {\rm optimal } policy  $(\pi^*,p^*) \in \mathcal A(F_0)$, for  which 
the optimal  preference process   $(U(t,F^{\pi^*,p^*}_t) + \int_0^t V(s,p^*_s) ds)$    is a martingale.
 \end{Definition}

General results for the existence and characterization of forward  utilities in this particular setting are given\footnote{In \cite{hillairet2024time}, $\theta$ is taken in $]0,1[$ but all the results hold for $\theta>1$.}  in \cite{hillairet2024time}. It relies on an Hamilton-Jacobi-Bellmann (HJB) constraint on the drift of the utility $U$ as well as a constraint on the sustainability bound $\mathfrak{K}$ representing a  buffer fund receiving the contribution $C$ and paying the minimal pension amount $ P^{\min}= p^{\min}N^r $ (see Theorem 4.3 \cite{hillairet2024time}):
 \begin{equation}\label{BFdynamic}
 d\mathfrak{K}_t =  (\mathfrak{K}_t r_t + C_t - p^{\min}_t N^r_t ) dt + \pi_t^\mathfrak{K}   (d B_t^S + \eta_t dt). 
 \end{equation}
More precisely, using It\^o-Ventzel’s formula,  the supermartingale property implied by the consistency condition translates into an HJB-constraint on the drift $(b_t)_{t\geq 0}$ of the  process $Z^u$ driving the buffer fund utility dynamics $U$ as defined in \eqref{eq:SDEZut} where the process $Z^u$ 
is a solution to the following non-linear SDE (see Proposition 5.4 in \cite{hillairet2024time})
\begin{align}
\label{eq:SDEZut_part}
d Z_t^u  &= Z_t^u\left( -\left[(1-\theta) r_t + \frac{1-\theta}{2 \theta} ( ({}^{t} \!L\delta_t)^S+\eta_t)^2 + \theta N_t^r \left( \frac{Z_t \omega_t}{Z_t^u}\right)^{1/\theta}\right] dt + {\delta_t} \cdot d B_t\right), \,  t \in [0, \tau^Z[
\end{align}
\vspace{-1cm}
\begin{align}
\label{eq:tauZ}
\text{where  \quad } \tau^Z:=\inf\{t \geq 0, Z^u_t =0\}.
\end{align}
In the previous equation,   $L$ is the square root matrix arising from 
the Cholesky decomposition of the correlation matrix   $\Gamma$, as introduced in \eqref{eq:Gamma_matrix}. For all $t \geq 0$, ${({}^{t}\!L\delta_t)^S \in \mathbb{R}}$ denotes the first component 
of ${}^{t} L\delta_t \in \mathbb{R}^4$, and $({}^{t}\!L\delta_t)^{\perp} \in \mathbb{R}^3$ 
its last three components. 
The derivation of the optimal investment-pension policy is  then computed using  Theorem 5.5 in \cite{hillairet2024time}.

 \begin{Theorem}[Optimal Policy]\label{ThOpt}
Let $U$ be the buffer fund utility  and $V$ be the aggregated retirees' utility as defined in \eqref{eq:defUfund}-\eqref{eq:SDEZut} and  \eqref{eq:Vpension}, respectively. \\ Assume $F_0 > \mathfrak{K}_0$ and  the HJB-dynamics \eqref{eq:SDEZut_part} for $Z^u$.
Then the optimal policy $(\pi^*,p^*)$ is 
\begin{numcases}{} \label{eq:pstar}
\pi_t^* =  \pi_t^\mathfrak{K} + \frac{(F_t^*-\mathfrak{K}_t)}{\theta} (  ({}^{t} \!L\delta_t)^S+\eta_t) \mathbf{1}_{\{t <\tau^Z\} } \\\label{eq:p*}
p^{*}_t =p^{\min}_t+ {(F_t^*-\mathfrak{K}_t)}  \left(\frac{{Z_t} { \omega_t}}{{Z_t^u} }\right)^{\frac{1}{\theta}}\mathbf{1}_{\{t <\tau^Z\} }=p^{\min}_t+  {(F_t^*-\mathfrak{K}_t) (Z_t^u)^{\frac{-1}{\theta}}}\left(Z_t \omega_t\right)^{\frac{1}{\theta}} \mathbf{1}_{\{t <\tau^Z\} }. 
\end{numcases}
 In particular, 
 the buffer fund $F^{*}$ induced by the optimal policy   $(\pi^{*}, p^{*})$ 
 satisfies the dynamics
\begin{equation}\label{DynPuiF^*}
dF_{t}^*= d \mathfrak{K}_t + (F_t^* -\mathfrak{K}_t)  \left[ \left(r_{t} {-} N^r_t \left(\frac{Z_t \omega_t}{Z_t^u }\right)^{\frac{1}{\theta}} \right)dt
+\frac{1}{\theta} (({}^{t} \!L\delta_t)^S +\eta_t)  \big(dB^S_{t}+\eta_{t}dt\big)\right].
\end{equation}
\end{Theorem}
\begin{proof} 
This is a straightforward adaptation of results in \cite{hillairet2024time}  to this framework of correlated Brownian motion $B$. Writing $B=L W,$ where 
$W$ is an uncorrelated Brownian motion, the volatility $\bar \delta$ of $Z^u$ with respect to $W$ becomes $\bar \delta= {}^{t} \!L \, \delta.$
With this transformation, the volatility component of $Z^u$  on the hedgeable part of $W$  (the first component of $W$ which is $B^S$) is  $\bar \delta_1 = ({}^{t} \!L\delta)^S$. Then the optimal  policy  $(\pi^*,p^*)$ follows from 
Theorem 5.5, the drift $b$ of  $Z^u$ follows from Proposition 5.3 and the fund dynamics follows from Proposition 5.2 in \cite{hillairet2024time}. 
\end{proof}
The optimal investment policy \eqref{eq:pstar} consists of two components: the baseline investment strategy at the sustainability bound $\pi_t^\mathfrak{K}$, plus an additional term proportional to the buffer fund surplus $(F_t^*-\mathfrak{K}_t)$. The first term corresponds to the classical myopic Merton strategy $\eta_t/\theta$. The second term represents the planner's sensitivity to hedgeable risk $({}^{t} \!L\delta_t)^S/\theta$. As expected, higher risk aversion $\theta$ yields more conservative portfolio allocations; however, unlike the myopic strategy, the social planner can adjust market exposure by modifying the utility preferences $({}^{t} \!L\delta_t)^S$, providing an additional degree of freedom for portfolio control. 

The optimal pension payment \eqref{eq:p*} consists of the minimum pension plus a surplus also proportional to the cushion $(F_t^* - \mathfrak{K}_t)$. The risk aversion adjusted weight $(\omega_t)^{1/\theta}$ captures the relative weight assigned by the social planner to a retiree at time $t$, for instance $\omega_t=1$ under equal treatment. On the other hand,  the preference component $\left(\frac{Z_t}{Z_t^u}\right)^{1/\theta}$ governs the trade-off between pension payouts and fund preservation.\\

Note that the consistency constraint determines the drift $b$ of $Z^u$. This means that the weight of the buffer fund utility updates in response to the evolution of both financial markets and the economic environment, which in turn drives the optimal pension paid to retirees. In particular, $(F^* - \mathfrak{K})$ and  $Z^u$ are strongly dependent. The following Proposition \ref{PropSurplus} further elaborates on the evolution of the total nominal pension surplus $(p^* - p^{\min})$. 

\begin{Proposition}[Pension surplus] \label{PropSurplus} Under the assumptions of Theorem \ref{ThOpt}, the  optimal surplus  pension 
is $(p^*_t - p^{\min}_t)  = \left(Z_t \omega_t\right)^{\frac{1}{\theta}} Y_t \mathbf{1}_{\{t <\tau^Z\} }$, with $Y_t := (F^*_t -\mathfrak{K}_t)(Z^u_t)^{-\frac{1}{\theta}}$  verifying 
\begin{eqnarray*}
 Y_t &=& (F^*_0 -\mathfrak{K}_0)(Z^u_0)^{-\frac{1}{\theta}} \operatorname{exp}{\Big[\frac{1}{\theta}  \int_0^t \left(\big(r_s +  \frac{\eta_s^2 +||({}^{t} \!L\delta_s)^\perp||^2 }{2}\big) ds +( ({}^{t} \!L\delta_s)^S+\eta_s )dB_s^S-\delta_s \cdot  dB_s \right) \Big]},
\end{eqnarray*}
and whose dynamic is given by
\begin{equation}\label{eq:UzF_part}
d Y_t= \frac{Y_t}{\theta}\Big[\Big(r_t +\frac{1+\theta}{2\theta} \left(\eta_t^2 +||({}^{t} \!L\delta_t)^\perp||^2 \right)\Big)dt + ( ({}^{t} \!L\delta_t)^S +\eta_t )dB_t^S-\delta_t \cdot  dB_t\Big].
\end{equation}
Furthermore, $\tau^Z = \inf \{ t \geq 0 ; Z^u_t=0\}$ coincides with  the buffer depletion time, i.e., the first time that the buffer fund  reaches the sustainability bound $\mathfrak{K}$.
\end{Proposition}
Proposition \ref{PropSurplus} shows that the pension surplus only depends on the population evolution through $\left(Z_t \omega_t\right)^{\frac{1}{\theta}}$ where $\omega_t$ is the weight attributed to a retiree at time $t$ and $Z_t$ is the individual preference weight, since $Y$ is independent of demographic factors. 

The initial value $Z_0^u$ has a strong impact on the surplus level, while higher risk aversion $\theta$ naturally yields lower pension surpluses. The financial drivers impact the pension surplus through the short rate $r$ and market risk premium $\eta$: on average, strong market performance increases both buffer fund growth and pension surplus.  

The sensitivity to the non-hedgeable risks $({}^{t} \!L\delta)^\perp$ determines the extent to which payouts deviate from what market conditions alone would deliver.  Notably, the sign of $({}^{t} \!L\delta_t)^\perp$ does not affect the average surplus, as it only depends on the  squared norm $\|({}^{t} \!L\delta_t)^\perp\|^2$. 
Positive values of $({}^{t} \!L\delta_t)^S$ of the sensitivity to the hedgeable risk increase both portfolio allocation  and pension surplus (Theorem \ref{ThOpt}) in favorable financial scenarios.
Negative values of $({}^{t} \!L\delta_t)^\perp$
mitigate the downward shocks on non-hedgeable risks.

A detailed analysis of the impact of the volatility $\delta $ on the optimal policy requires to consider the joint distribution of $p^{\min}$, the surplus $(p^*-p^{\min})$ (that is $Y$) together with the depletion time: while the dependence on the wages $\mathfrak e$ of  $p^{\min}$ and the depletion time is obvious, wages are also correlated to  $Y$ via the utility volatility term $\delta^{\mathfrak e}$. This subtle, nonetheless important, feature will be analyzed in detail in the numerical section.

\begin{proof}
First note that since $\Gamma= L \, ^{t}\!L$,
$  {}^{t} \!\delta  \Gamma \delta -  (({}^{t} \!L\delta)^S)^2  
= {}^{t} \!({}^{t} \!L\delta)({}^{t} \!L\delta) - (({}^{t} \!L\delta)^S)^2 = ||({}^{t} \!L\delta)^\perp||^2.$ Then, 
using \eqref{DynPuiF^*} and \eqref{eq:SDEZut}, 
the  exponential form  of $Y_t =(F_t^* -\mathfrak{K}_t)(Z_t^u)^{\frac{-1}{\theta}}$ is 
\begin{eqnarray*}
Y_t&=& Y_0 \, \exp{\Big[\frac{1}{\theta}  \int_0^t \left((r_s +  \frac{\eta_s^2 - (({}^{t} \!L \delta_s)^S)^2 + {}^{t} \!\delta_s \Gamma \delta_s }{2} ) ds +( ({}^{t} \!L \delta_s)^S+\eta_s )dB_s^S-\delta_s \cdot  dB_s \right) \Big]}\\
&=& Y_0 \, \exp{\Big[\frac{1}{\theta}  \int_0^t \left(\big(r_s +  \frac{\eta_s^2 +||({}^{t} \!L\delta_s)^\perp||^2 }{2})\big) ds +( ({}^{t} \!L \delta_s)^S+\eta_s )dB_s^S-\delta_s \cdot  dB_s \right) \Big]}
\end{eqnarray*}
and using Itô's formula, its  dynamics  is given by
$$
d Y_t= \frac{Y_t}{\theta}\Big[\Big(r_t +\frac{1+\theta}{2\theta} \left(\eta_t^2 +||({}^{t} \!L\delta_t)^\perp||^2 \right)\Big)dt + ( ({}^{t} \!L\delta_t)^S +\eta_t )dB_t^S-\delta_t \cdot  dB_t\Big].$$
As a result, $(F_t^* -\mathfrak{K}_t) =(Z_t^u)^{\frac{1}{\theta}} Y_t$ where $Y$ is a positive process, therefore 
  $ \inf \{ t \geq 0 ; Z^u_t=0\}$ coincides with  the buffer depletion time $\inf \{ t \geq 0 ; \; F^*_t  = \mathfrak{K}_t\}$.
\end{proof}

\subsection{Sustainability and adequacy criteria}\label{subsec:sustainability_adequacy}

We  now investigate a sustainability and several adequacy criteria to study the relevance of the pension scheme.

\paragraph{Buffer depletion time} An important sustainability criterion is the time  $\tau$ defined as 
the first time the buffer fund reaches the maximum allowed debt amount, 
$$\tau := \inf \{ t \geq 0 ; \; F^*_t  = \mathfrak{K}_t\}=\tau^Z,$$
that is how long the pension system remains sustainable while still ensuring the adequacy constraints \eqref{eq:Pensionmin}.  

\begin{Proposition}
\label{PropTau}
   The buffer depletion time $\tau= \inf \{ t \geq 0 ; \; F^*_t  = \mathfrak{K}_t\},$
is characterized as 
   \begin{equation}\label{deftau}
\tau = \inf \left\{ t \; ; \;  \int_0^t N_s^r \left(\frac{Z_s \omega_s \xi_s}{ Z_0^u } \right)^{\frac{1}{\theta}}ds  \geq1 \right\},
\end{equation}
where \begin{equation}\label{eq:xi}
\xi_t   := \operatorname{exp}\left(\int_0^t \left((1-\theta)r_s + \frac{(1-\theta)}{2\theta}   (({}^{t} \!L\delta_s)^S+ \eta_s)^{2} +\frac{{}^{t} \!\delta_s \Gamma  \delta_s }{2}\right)ds - \int_0^t \delta_s\cdot  dB_s\right)
\end{equation}
can be interpreted as a  stochastic utility-dependent discount factor.
\end{Proposition}

\begin{proof}
By Proposition \ref{PropSurplus} the buffer depletion time $\tau$ is also  the first time  $\tau^Z$ that the process $Z^u$  reaches 0. Thanks to the  (HJB)-consistency condition \eqref{eq:SDEZut_part}, $Z_t^u$ can be written  on the random time interval $[0,\tau[$ as  (see Proposition 5.4 in \cite{hillairet2024time})
\begin{equation}\label{exprZu}
Z_t^u = \xi_t^{-1}\left((Z^u_0)^{\frac{1}{\theta}} -  \int_0^t N_s^r(Z_s \omega_s\xi_s)^{\frac{1}{\theta}} ds \right)^{\theta}, \quad \forall t \in [0,\tau[, 
\end{equation}
where $\xi$ is given by \eqref{eq:xi}, and \eqref{deftau} follows immediately.
\end{proof}
The buffer depletion time $\tau$ depends on the stochastic utility-dependent discount factor $\xi$ and on the demographic evolution $N^r$. It also depends on the social planner and retirees' preferences through the processes  $\omega$, $Z$, and the parameters $Z^u_0$ and $\theta$. A larger $Z^u_0$ delays the buffer depletion time $\tau$ but reduces pension levels (see Equation \eqref{eq:p*}). Notably, the buffer depletion time $\tau$ does not depend on the initial buffer fund level $F_0$, while the initial wealth level naturally affects the nominal pension surplus.

\paragraph{Adequacy criteria} For the adequacy criterion, we will assess three metrics: the benefit ratio  $\text{BR}_t$, the relative pension increase $\rho_t$ with respect to $p^{\min}_t$ and the equivalent annual indexation rate (EAIR).
First, the benefit ratio  $\text{BR}_t$ is the ratio between pensions $p^*_t$ and wages $\mathfrak{e}_t$ in the same year:
\begin{equation}\label{eq:BR}
    \text{BR}_t:=\frac{p^*_t}{\mathfrak{e}_t}.
\end{equation}
Note that, in a pure PAYG context with a steady-state population the benefit ratio is constant: $\text{BR}_t=\text{BR}$ for all $t$. Introducing a buffer fund, the benefit ratio will  no longer be constant, even under the steady state.\\
Secondly, we study the relative increase $\rho_t$ with respect to $p^{\min}_t$, that is
\begin{equation}
    \rho_t:=\frac{p^*_t-p^{\min}_t}{p^{\min}_t}.
\end{equation}
For any annual random cash flow stream $c$, we assess the constant Equivalent Annual Indexation Rate (EAIR) over $[s_1,s_n]$, defined as the random rate $y$ solving
\begin{align*}
\sum_{i=1}^{n} c_{s_i} = c_{s_1} \sum_{i=1}^{n} (1+y)^{s_i-s_1}.
\end{align*}
Since $y$ is random and scenario-dependent, our analysis focuses on the mean and selected percentiles of the EAIR for the buffer fund scheme ($y^{\text{BF}}_t$) when $c=p^*$ and the pure PAYG scheme ($y^{\min}_t$) when $c=p^{\min}$.

Section \ref{sec:numerics} further investigates the sensitivity of $\tau$, $\text{BR}$, $\rho$ and EAIR to the initial condition and preference parameters under different demographic configurations, including steady-state population and baby boom scenarios. 

\section{Numerical application}\label{sec:numerics}
We model a stylized pension system over a 40-year horizon.
The numerical study is conducted under two demographic scenarios: a steady-state scenario (SS) in which the dependency ratio $DR := \frac{N^r }{N^w}$ is constant at $0.3$, and a baby boom scenario (BB) in which $DR$ increases linearly from $0.3$ to $0.5$ over a 40-year period, mimicking Eurostat projections for ageing European populations \cite{eurostat2023pop}. For a constant\footnote{Our framework scales to any value of $N^w_t$ without affecting the results.} working population $N^w_t = 100$ for all $t$, the steady-state scenario yields 30 retirees throughout, while the baby boom scenario implies a retiree population growing from 30 to 50 individuals.\\
The numerical analysis is implemented in Python, using Euler schemes to discretize the stochastic dynamics of the financial and economic processes. 

\subsection{Setup and Parametrization}
 The  financial and economic parameters of the numerical analysis are calibrated from the literature and the demographic assumptions represent stylized demographic trends representative of European pension systems. 
We conduct an in-depth analysis of the impact  on the pension scheme  of the  initial fund size $F_0$ and  the buffer utility parameters (its sensitivity $\delta$,  its weight $Z^u_0$ and its relative risk aversion parameter $\theta$).
Pathwise results presented in this section are based on two financial scenarios\footnote{The financial and economic processes associated to the presented pathwise scenarios are shown in Appendix \ref{app:finmarketscenarios}.}, an optimistic one and a pessimistic one,   as specified in Table \ref{tab:basecasepar}. Monte Carlo analysis uses $N_{\text{sim}}=10{,}000$ simulations to assess distributional properties of buffer depletion time and optimal policy metrics.  Table \ref{tab:basecasepar} summarizes the baseline parametrization detailed below. 

\paragraph{Adequacy and sustainability constraints} The minimum pension $p^{\min}$  is defined as the pure PAYG pension amount 
\begin{equation}\label{eq:pmin_part}
p^{\min}_t  =\frac{C_t}{N_t^r} = \frac{\alpha \; N^w_t \;\mathfrak{e}_t}{N^r_t}=\frac{\alpha \; \mathfrak{e}_t}{DR_t}, 
\end{equation}
with $DR_t = \frac{N^r_t }{N^w_t}$ the dependency ratio. In particular, the total pension expenditures associated to the minimum pension is $P_t^{\min}= p^{\min}_t N^r_t =C_t$.  Since $C_t - P_t \leq 0$ for all $t\geq 0$, a deficit  that will be financed by the buffer fund $(F_t)_{t\geq 0}$. The contribution rate is set at $\alpha=0.15$ (15\% of wages), a typical level for European pension systems. \\
In the numerical application, we assume that the baseline investment strategy is $\pi^\mathfrak{K}\equiv 0$. This implies that the present value of maximum amount of debt is constant, i.e. 
\begin{equation*}
\mathfrak{K}_t = \mathfrak{K}_0 e^{\int_0^t r_s ds},\quad \forall t\geq 0. 
\end{equation*}
Then the optimal policy $(\pi^*,p^*)$ is 
\begin{numcases}{   }
\pi_t^* =  \frac{(F_t^*-\mathfrak{K}_0 e^{\int_0^t r_s ds})}{\theta} (({}^{t} \!L\delta_t)^S+\eta_t) \mathbf{1}_{\{t <\tau^Z\}}  \\ \label{eq:pstarex}
p^{*}_t =\frac{C_t}{N_t^r}+ {(F_t^*-\mathfrak{K}_0 e^{\int_0^t r_s ds})}  \left(\frac{{Z_t} { \omega_t}}{{Z_t^u}}\right)^{\frac{1}{\theta}}  \mathbf{1}_{\{t <\tau^Z\}} \label{p*ex}
\end{numcases}
and the dynamics of the  buffer fund surplus $(F_t^*-\mathfrak{K}_0 e^{\int_0^t r_s ds})$ satisfies
\begin{equation*}
d(F_t^*-\mathfrak{K}_0 e^{\int_0^t r_s ds})= (F_t^*-\mathfrak{K}_0 e^{\int_0^t r_s ds})  \left( \left(r_{t} {-} N^r_t (\frac{Z_t \omega_t}{Z_t^u})^{\frac{1}{\theta}} \right)dt
+\frac{1}{\theta} ( ({}^{t} \!L\delta_t)^S +\eta_t)  \big(dB^S_{t}+\eta_{t}dt\big)\right).
\end{equation*}
Since the optimal policy only depends on $\mathfrak{K}$ through the buffer fund surplus,  we  assume without loss of generality that $\mathfrak{K}_0=0$.

\paragraph{Aggregated retirees' utility} Following standard subjective preference representation, we specify the time preference process for one retiree's utility  as $Z_t=Z_0 e^{-\beta t}$, yielding the following representation for an individual retiree:
$$v(t,p)= Z_0 e^{-\beta t}\frac{(p-p^{\min}_{t})^{1-\theta}}{1-\theta}.$$
We adopt $\beta=0.03$ and the risk-aversion parameter $\theta=4$ from \cite{munk2010dynamic}. Appendix \ref{app:theta} discusses the sensitivity to the risk aversion  parameter $\theta$ in detail. The constant $Z_0$ is a normalizing constant with no impact, taken as $Z_0=(\frac{1}{N_0^w})^\theta$.
The aggregate retiree's utility is  given by 
 $$V(t, p ) = N^r_t \omega_t {Z_0} e^{-\beta t} \frac{(p-p^{\min}_{t})^{1-\theta}}{1-\theta}.$$
 We consider two aggregation choices for retirees. In the first case, we take
\begin{equation}
\label{eq:omega_1}
\omega_t = 1, \quad \forall t \geq 0,
\end{equation}
meaning that all retirees receive equal weight regardless of whether they belong to a baby-boom generation. In this case, Proposition~\ref{PropSurplus} shows that the nominal pension surplus $(p^* - p^{\min})$ does not depend on the demographic scenarios, but depends on the wages through $\delta^{\mathfrak{e}}$ and hence on $p^{\min}$ (see discussion below).  However, under a baby boom scenario (with increasing $DR$), the sustainable minimum pension $p^{\min}$ will naturally decrease. Hence, the relative pension increase $\rho$ will be substantially higher following the introduction of the buffer fund compared to the steady-state case.

In an alternative specification, the individual retiree weight is given by
\begin{equation}
\label{eq:omega_DR}
\omega_t = \frac{DR_t}{DR_0} = \frac{N^r_t}{N^r_0},
\end{equation}
since the working population remains constant. Under this choice, retirees belonging to larger cohorts receive a higher relative weight compared to the equal-treatment case $\omega_t = 1$.

\paragraph{Financial market and wages} 
The risky asset follows the Heston stochastic volatility dynamics \eqref{eq:heston}. We set the equity premium at 4\% above the initial risk-free rate ($\mu=r_0+0.04$), following \cite{munk2010dynamic}. The stochastic variance process has initial value $\nu_0=0.04$ (corresponding to 20\% volatility), long-run mean $\bar\nu=0.04$, mean reversion speed $\kappa=3$, and volatility of volatility $\sigma_\nu=0.2$. These parameters are calibrated from the implied stochastic volatility literature \cite{ait2021implied}. We specify a leverage effect correlation $\rho_{S,\nu}=-0.7$ between stock returns and variance innovations \cite{ait2021implied}, and set all other correlations to zero following \cite{munk2010dynamic} and \cite{grzelak2011heston}.
The risk-free rate follows a mean-reverting Vasicek process \eqref{eq:Vasicek} with long-run mean $b=0.02$, mean reversion speed $a=0.50$, and volatility $\sigma_r=0.02$, consistent with \cite{munk2010dynamic}. We set the initial rate $r_0=0.03$ to generate a declining interest rate trajectory over the simulation horizon.\\
The average wage process follows Equation \eqref{eq:salary},  with constant  drift $\lambda=0.02$ (unless stated otherwise) and volatility $\sigma_\mathfrak{e}=0.02$, calibrated from \cite{munk2010dynamic}. The initial wage level $\mathfrak{e}_0=39$ thousand euros corresponds to the Eurostat EU-25 average annual wage \cite{eurostat2023wages}.

\paragraph{Buffer fund utility} 
The volatility process $\delta$ of $Z^u$, associated to the correlated Brownian motion $B$ as defined in \eqref{eq:SDEZut}, is set to the constant
$$\delta = {}^{t} \!(0, -0.2, -0.2, -0.2),$$
where the components correspond to the sensitivity of the buffer fund utility to the volatilities of the financial and economic processes ${}^{t} \!(S, \nu, r, \mathfrak{e})$. Since the risky asset is traded, the social planner is neutral with respect to its volatility. The remaining parameters $\delta^\nu$, $\delta^r$, and $\delta^{\mathfrak{e}}$ are chosen negative in order to mitigate non-hedgeable uncertainty.

Due to the correlation structure $\Gamma = L {}^{t} \!L$, decomposing $\delta$ into its hedgeable and non-hedgeable parts yields ${}^{t} \!L\delta = (0.14, -0.1428, -0.2, -0.2)$, where
\begin{equation}
L=\begin{pmatrix}
1 & 0 & 0 & 0 \\
\rho_{S,\nu} & \sqrt{1-\rho_{S,\nu}^2} & 0 & 0 \\
0 & 0 & 1 & 0 \\
0 & 0 & 0 & 1
\end{pmatrix}=\begin{pmatrix}
1 & 0 & 0 & 0 \\
-0.7 & \sqrt{0.51} & 0 & 0 \\
0 & 0 & 1 & 0 \\
0 & 0 & 0 & 1
\end{pmatrix}.
\end{equation}
The first component $({}^{t} \!L\delta)^S$ is the hedgeable part, which increases portfolio allocation as in \eqref{eq:pstar}, while $({}^{t} \!L\delta)^\perp$ governs the pension surplus \eqref{eq:p*} as shown in Proposition~\ref{PropSurplus}. The role of $\delta$ is analyzed further in the remainder of this Section.

\paragraph{Initial conditions} In the baseline scenario, the initial buffer fund is set to $F_0=C_0$, equivalent to one year of total contributions. As discussed in Proposition \ref{PropTau}, the buffer depletion time $\tau$ depends heavily on the choice $Z^u_0$ but not on $F_0$. We choose to calibrate $Z^u_0$ to yield an initial pension 5\% above the sustainable level: $p_0=1.05 \cdot p^{\min}_0$. From the control \eqref{eq:p*}, this implies
$$\left(\frac{Z^u_0}{Z_0}\right)^{\frac{1}{\theta}}=\frac{F_0 - \mathfrak{K}_0}{p_0 - p^{\min}_0}=\frac{F_0}{0.05 \cdot p^{\min}_0}=20 \cdot N^r_0=600,$$
yielding $Z^u_0=1296$ for $\theta=4$. This calibration reflects a policy objective of moderate initial generosity while maintaining long-run sustainability.

\newpage

\begin{table}[h!]\smallerthansmall
\centering 
\caption{Base Case Parametrization}\label{tab:basecasepar}
\begin{tabular}{@{}llll@{}}
\toprule
\textbf{Category} & \textbf{Parameter} & \textbf{Value} & \textbf{Source} \\
\midrule
\multicolumn{4}{l}{\textit{Simulation Design}} \\\hline
& Time horizon $T$ & 40 years & \\
& Time step & 10 per month & \\
& Illustrative seeds & 3, 5 & Optimistic, pessimistic \\
& Monte Carlo simulations $N_{\text{sim}}$ & 10,000 & \\ 
\midrule
\multicolumn{4}{l}{\textit{Demographics}} \\\hline
& Workers $N^w_t$ & 100 & Constant \\
& Steady-state $DR_t$ & 0.3 & Fixed \\
& Baby boom $DR_t$ & $0.3 \to 0.5$ & Linear over 40 years\\
\midrule
\multicolumn{4}{l}{\textit{Stock Process (Heston)}} \\\hline
& Initial price $S_0$ & 1 & Normalization \\
& Equity premium & 0.04 & $\mu = r_0 + 0.04$ \cite{munk2010dynamic} \\
& Initial variance $\nu_0$ & 0.04 & $\sqrt{\nu_0}=0.2$ \\
& Long-run variance $\bar \nu$ & 0.04 & $\sqrt{\bar \nu}=0.2$ \cite{ait2021implied} \\
& Mean reversion $\kappa$ & 3 & \cite{ait2021implied} \\
& Vol of vol $\sigma_\nu$ & 0.2 & \cite{ait2021implied} \\
\midrule
\multicolumn{4}{l}{\textit{Interest Rate Process (Vasicek)}} \\\hline
& Initial rate $r_0$ & 0.03 & \\
& Long-run mean $b$ & 0.02 & \cite{munk2010dynamic} \\
& Mean reversion $a$ & 0.50 & \cite{munk2010dynamic} \\
& Volatility $\sigma_r$ & 0.02 & \cite{munk2010dynamic} \\
\midrule
\multicolumn{4}{l}{\textit{Wage Process}} \\\hline
& Initial wage $e_0$ & 39k & \cite{eurostat2023wages} \\
& Drift $\lambda$ & 0.02 & \cite{munk2010dynamic} \\
& Volatility $\sigma_e$ & 0.02 & \cite{munk2010dynamic} \\
\midrule
\multicolumn{4}{l}{\textit{Correlations}} \\\hline
& $\rho_{S,\nu}$ (leverage effect) & $-0.7$ & \cite{ait2021implied} \\
& $\rho_{S,r}$, $\rho_{S,e}$ & 0 & \cite{munk2010dynamic} \\
& $\rho_{\nu,r}$, $\rho_{\nu,e}$ & 0 & \cite{grzelak2011heston} \\
& $\rho_{r,e}$ & 0 & \cite{munk2010dynamic} \\
\midrule
\multicolumn{4}{l}{\textit{Pension Policy}} \\\hline
& Contribution rate $\alpha$ & 0.15 & European benchmark \\
& Initial fund $F_0$ & 585k & $C_0=\alpha \cdot N^w_0 \cdot e_0$ \\
& Sustainability bound $\mathfrak{K}_0$ & 0 & No borrowing constraint\\
& Minimum pension $p^{\min}_t$ & $\alpha \cdot \mathfrak{e}_t / DR_t$ & Equation \eqref{eq:sustainablepmin}\\
& Initial surplus target & 5\% & $p_0 = 1.05 \cdot p^{\min}_0$ \\
\midrule
\multicolumn{4}{l}{\textit{Preferences}} \\\hline
& Retiree  utility weight $Z_0$ & $10^{-8}$ &  Normalizing constant $(\frac{1}{N_0^w})^\theta$\\ 
& Time preference $\beta$ & 0.03 & \cite{munk2010dynamic} \\
& Risk aversion $\theta$ & 4 & \cite{munk2010dynamic}\\
& Sensitivity of $Z^u$ $\delta$ & ${}^{t} \!(0,-0.2,-0.2,-0.2)$ & Baseline case scenario \\
& Buffer utility weight $Z^u_0$ & 1296 & Calibrated from surplus target \\
& Retirees' weight $\omega_t$ & $1, \; \frac{DR_t}{DR_0}$ &  Equal weight, DR dependent \\
\bottomrule
\end{tabular}
\end{table}

\newpage 

\subsection{Demographic impact}

We start by analyzing the impact of population ageing by comparing results  for the steady state scenario (SS, in black) and the  baby-boom scenario (BB, in red). In this first example, the retirees' weight $\omega_t$ are all equal to $1$. Figure \ref{fig:SS_v_BB_seed_3} presents the optimal processes corresponding to the optimistic financial scenario (seed 3), while Figure \ref{fig:SS_v_BB_MC} illustrates the empirical distributions  obtained via Monte Carlo simulations. The pessimistic scenario (seed 5, Figure \ref{fig:SS_v_BB_seed_5}) as well as some additional distributional results (Figure \ref{fig:SS_v_BB_MC_appendix}) are relegated to Appendix \ref{app:single_path_seed_5} and \ref{app:MC_extra}. 

\paragraph{Pathwise Analysis} As stated above, the proportion  $\frac{\pi^*_t}{\sqrt{\nu_t} F^*_t} =  \frac{({}^{t} \!L\delta)^S+\eta_t}{\theta\sqrt{\nu_t}}$
of the buffer fund  allocated to the risky asset and the nominal pension surplus $(p^*_t-p^{\min}_t)$ are not affected by demographics, as shown in Figures~\ref{subfig:SS_v_BB_seed_3_proportion}-\ref{subfig:SS_v_BB_seed_3_surplus}.\\
Demographics do, however, affect other key variables. First, as shown in Figure~\ref{subfig:SS_v_BB_seed_3_Zu}, the buffer utility weight $Z^u$ decreases more rapidly in the baby boom scenario due to higher pension expenditures, diminishing the weight placed on fund preservation relative to pension payments and leading to earlier depletion (see the definition of $\tau$ in Equation \eqref{deftau}). The volatility observed in $Z^u$ stems from the non-zero buffer fund preference volatility parameter $\delta$. Second, while the nominal surplus per retiree remains constant, the buffer fund depletes faster when the same surplus must be distributed to a growing number of retirees, see Figure~\ref{subfig:SS_v_BB_seed_3_buffer_fund}. This is particularly salient in the pessimistic scenario (Figure \ref{fig:SS_v_BB_seed_5}).

Figures~\ref{subfig:SS_v_BB_seed_3_p_star}-\ref{subfig:SS_v_BB_seed_3_BR} illustrate that the minimum PAYG pension, chosen to be equal to the sustainable pure PAYG pension $p_t^{\min} = \alpha \mathfrak{e}_t / DR_t$ is also affected by demographics. Note that in the baby boom scenario, the rise in $DR$ is partially offset by wage growth over time. The buffer fund enables a pension surplus to be paid, substantially increasing the benefit ratio \eqref{eq:BR}. In the fund depletes, the  system reverts to paying only the pure PAYG pension $p^{\min}_t$ (Figure~\ref{subfig:SS_v_BB_seed_5_pstar}-\ref{subfig:SS_v_BB_seed_5_benefit_ratio}). Since the nominal surplus is identical across demographic scenarios, the relative improvement from the buffer fund is more pronounced in the baby boom case, given the lower baseline $p^{\min}_t$. 

\paragraph{Distribution of optimal strategy} 
 To go further,  Figure \ref{fig:SS_v_BB_MC} presents a more  comprehensive analysis of the optimal buffer fund policy over time, based on $N = 10{,}000$ Monte Carlo simulations 
under both the SS and BB demographic scenarios. The evolution of the buffer fund,  $Z^u$  and the benefit ratio are presented in Figure \ref{fig:SS_v_BB_MC_appendix} in the Appendix. \\
In all figures, solid lines represent the mean, dashed lines 
the median, and the shaded area indicates the interquartile range IQR  (25th-75th percentile). In addition, all panels display the empirical survival function of the buffer depletion time 
$\tau$ (dotted lines) under the SS scenario (Scenario 1, black) and the BB 
scenario (Scenario 2, red), as defined in  Proposition~\ref{PropTau}.\\
The formula makes it clear that the larger the retiree population,  the earlier the fund is depleted. 
Under the SS scenario, the average buffer depletion time is approximately 28  years, compared to just over 23 years under the BB scenario,  
reflecting the heavier pension burden imposed by the baby boom cohort. To illustrate, after 20 years, more than 80\% of funds remain solvent under the SS scenario, compared to 64\% under the BB scenario. The influence of the buffer fund utility sensitivity $\delta$ is further discussed in Section~\ref{subsec:sustainability}.\footnote{The sensitivity analysis with respect to the risk aversion parameter $\theta$ is outlined in Appendix \ref{app:theta}.}\\
Figure~\ref{subfig:SS_v_BB_MC_proportion} displays the evolution of the proportion 
of the buffer fund invested in the risky asset,  conditional on the buffer fund 
remaining solvent, which makes observations noisier near the end of the horizon. As established in the previous paragraph, the optimal investment strategy is identical across both demographic scenarios and remains stable with the IQR ranging between 35\% and 65\% with a median allocation close to 50\%. 

The yearly relative pension increase attributable to the buffer fund is presented in Figure~\ref{subfig:SS_v_BB_MC_relative_surplus}. Consistent with the single simulation analysis, the relative pension increase is more pronounced in the BB scenario than in the SS scenario, driven by the lower PAYG pension baseline. 
Conditional on fund survival, this increase rises on average from the initially calibrated 5\% (used to determine $Z^u_0$) to approximately 7\% at the median depletion time in the BB scenario, while it remains slightly above 5\% in the SS scenario. The evolution of the benefit ratio (BR) is presented in Figure~\ref{subfig:SS_v_BB_MC_benefit_ratio} in the Appendix, showing an average increase of approximately 2.5-3\% in both scenarios. Table \ref{tab:eair_combined} (Panel A) further confirms the greater effect and added value of the buffer fund in the baby-boom scenario since the relative increase in the equivalent annual indexation rate is greater. {\color{black}Notably, the EAIR $\overline{y^{\min}_t}$ grows over time because $p_t^{\min}$ depends linearly on $1/DR_t$. Therefore, the decrease in $p_t^{\min}$ induced by demographic ageing will be greatest during the first 10-20 years.}

Finally, Figures~\ref{subfig:SS_v_BB_MC_p_tot_withbuffer} 
and~\ref{subfig:SS_v_BB_MC_p_tot_withoutbuffer} display the yearly pension amount $p^*_t$ under two complementary conditioning sets: scenarios in which the buffer fund remains solvent at time $t$ 
(Figure~\ref{subfig:SS_v_BB_MC_p_tot_withbuffer}), and scenarios in which it has already been depleted by time $t$ (Figure~\ref{subfig:SS_v_BB_MC_p_tot_withoutbuffer}) (the lines do not start at time 0 since all buffer funds stay solvent for a number of years). The average minimum PAYG pension $p^{\min}_t$ is displayed as 
a blue dashed line in both figures  and for both demographic scenarios (SS and BB). 
Figure~\ref{subfig:SS_v_BB_MC_p_tot_withbuffer} confirms the expected result: conditional on the buffer fund remaining solvent, retirees receive on average  a higher pension than under the pure PAYG scheme, reflecting the additional  income provided by the buffer fund. 
Figure~\ref{subfig:SS_v_BB_MC_p_tot_withoutbuffer} reveals a more striking and important property of the optimal policy. Even after the buffer fund is depleted, that is, when retirees no longer benefit from any buffer fund supplement, the average pension $\mathbb{E}[p^*_t \mid F_t = 0]$ remains higher than the average pension $\mathbb{E}[p^{\min}_t]$ paid by the pure PAYG scheme. This  is a  consequence of the choice $\delta^{\mathfrak{e}} = -0.2 \neq 0$ in Equation~\eqref{eq:SDEZut}, which creates a dependency between the nominal pension surplus, the depletion time, and the wage process. This choice of $\delta^{\mathfrak{e}}$ mitigates economic risk by concentrating surplus payments in unfavorable wage, and hence low PAYG pension, scenarios. The impact of $\delta$ on the distribution of pension 
amounts is further discussed in the next section and in 
Figures~\ref{fig:bb_delta_0_v_delta_base_dist_pension} 
and~\ref{fig:bb_delta_4_-_v_delta_4_+_dist_pension}.

\begin{figure}
\centering
	\setfolder{fig_24feb_delta_base_S1_SS_v_S2_BB_seed_3}
\begin{subfigure}[t]{0.49\textwidth}
\centering
\includegraphics[width=\textwidth,height=0.65\textheight,keepaspectratio]{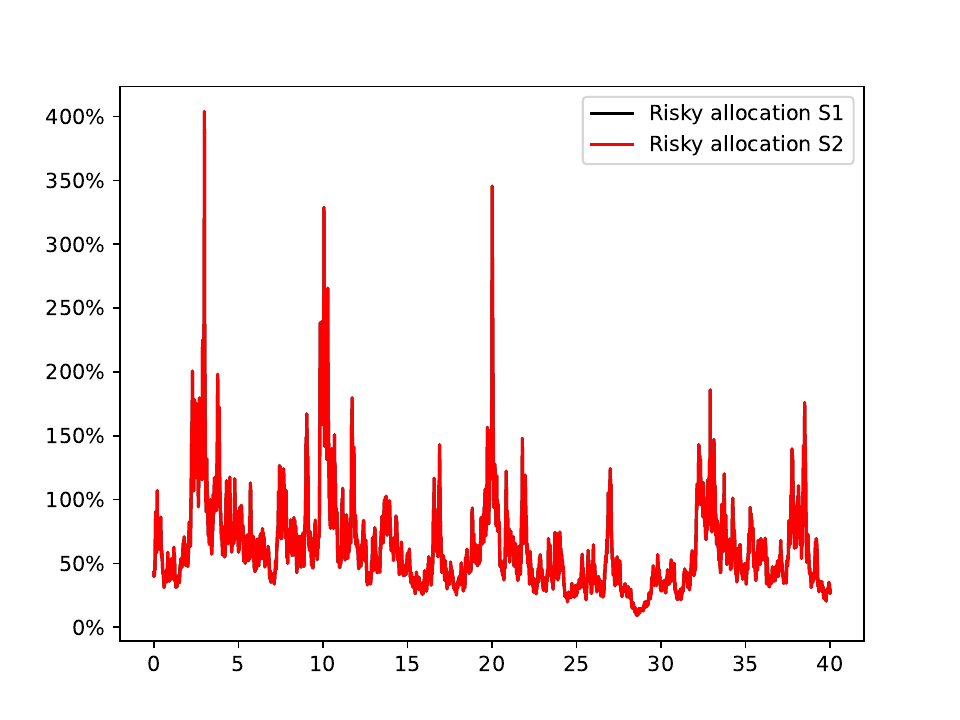}
\vspace{-0.2cm}\caption{Proportion risky investment $\left(\frac{\pi^*_t}{\sqrt{\nu_t}F^*_t}\right)$}\label{subfig:SS_v_BB_seed_3_proportion}
\end{subfigure}
\begin{subfigure}[t]{0.49\textwidth}
\centering
\includegraphics[width=\textwidth,height=0.65\textheight,keepaspectratio]{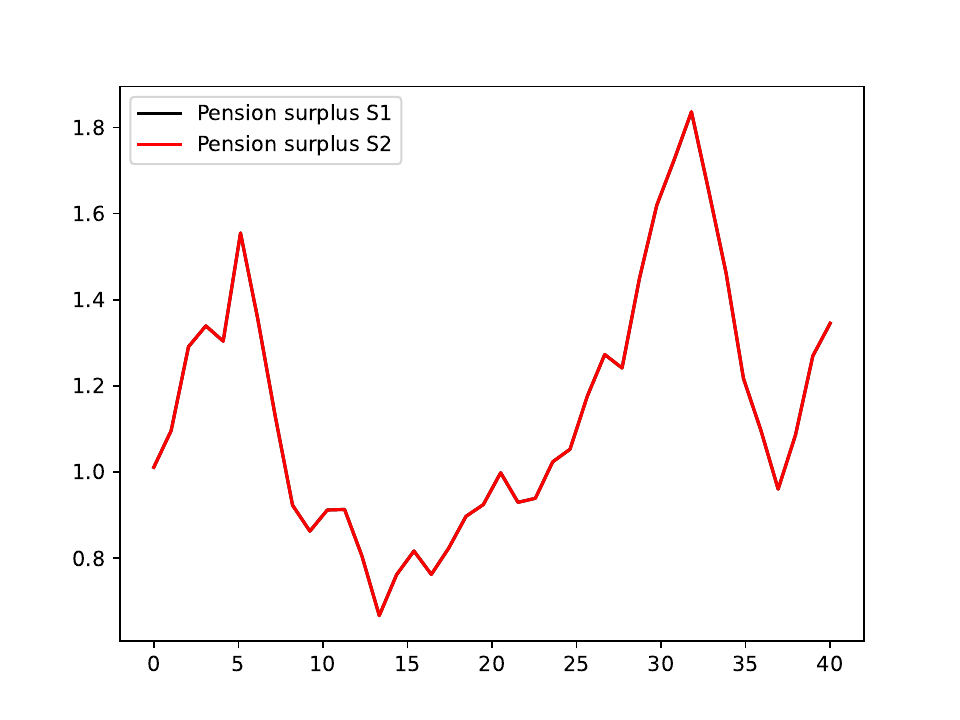}
\vspace{-0.2cm}\caption{Pension surplus $p^*_t-p^{\min}_t$}\label{subfig:SS_v_BB_seed_3_surplus}
\end{subfigure}
\begin{subfigure}[t]{0.49\textwidth}
\centering
\includegraphics[width=\textwidth,height=0.65\textheight,keepaspectratio]{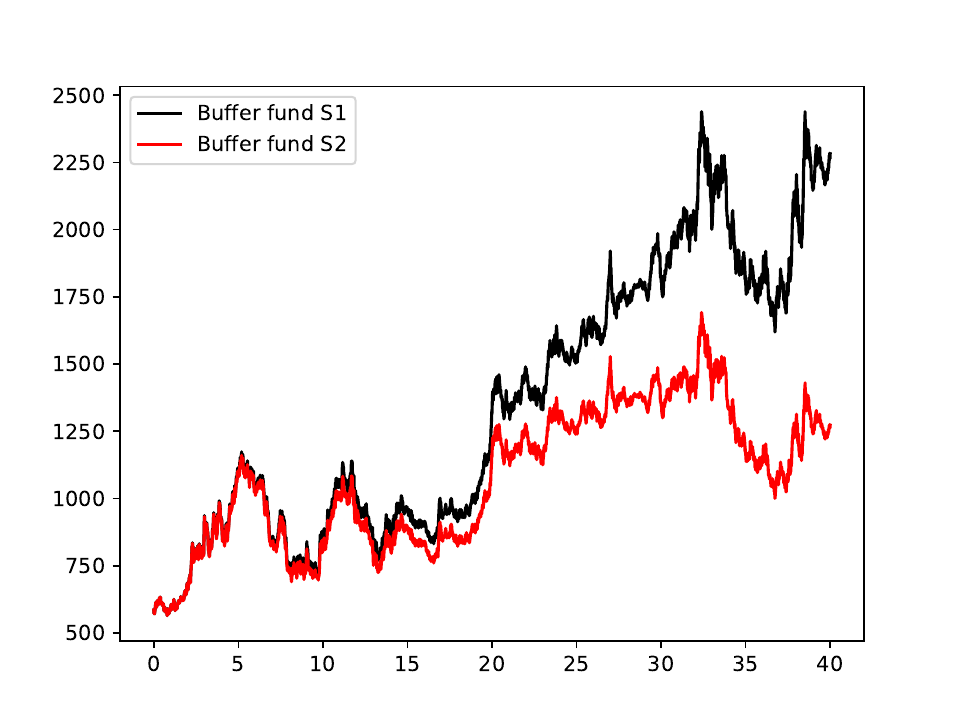}
\vspace{-0.2cm}\caption{Buffer fund $F^*_t$}\label{subfig:SS_v_BB_seed_3_buffer_fund}
\end{subfigure}\hfill
\begin{subfigure}[t]{0.49\textwidth}
\centering
\includegraphics[width=\textwidth,height=0.65\textheight,keepaspectratio]{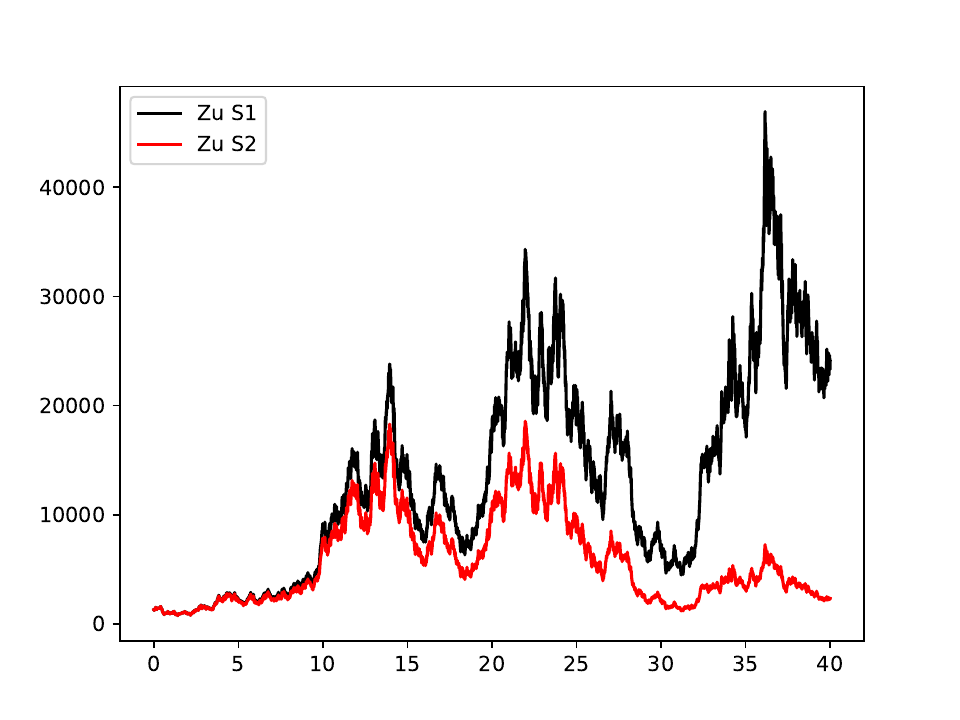}
\caption{Buffer fund utility weight $Z^u_t$}\label{subfig:SS_v_BB_seed_3_Zu}
\end{subfigure}\hfill
\begin{subfigure}[t]{0.49\textwidth}
\centering
\includegraphics[width=\textwidth,height=0.65\textheight,keepaspectratio]{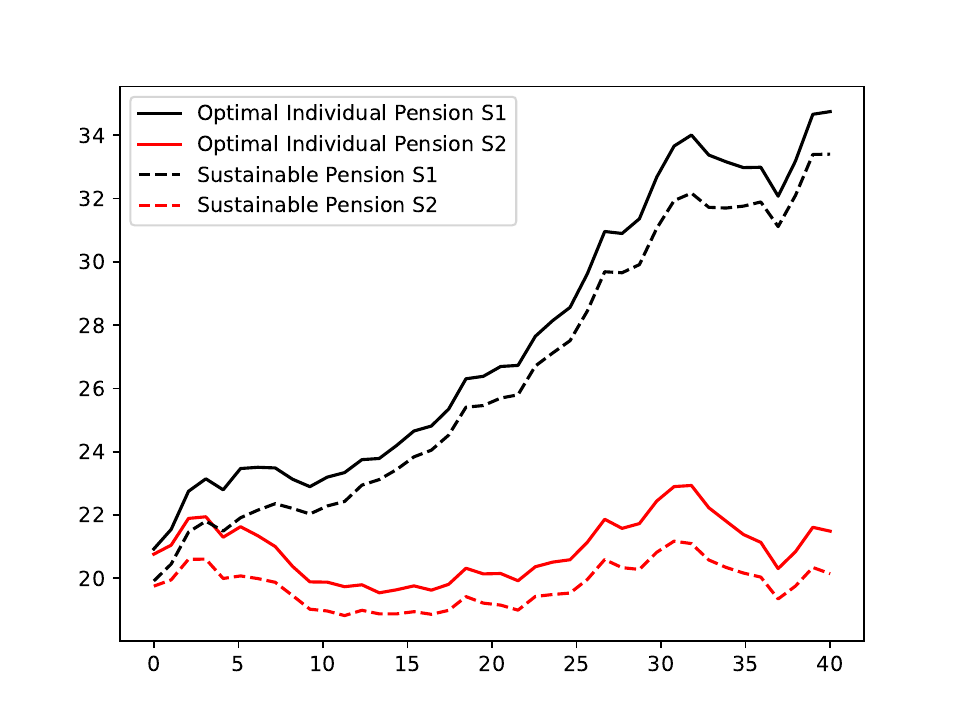}
\vspace{-0.2cm}\caption{Individual pensions $p^*_t$}\label{subfig:SS_v_BB_seed_3_p_star}
\end{subfigure}\hfill
\begin{subfigure}[t]{0.49\textwidth}
\centering
\includegraphics[width=\textwidth,height=0.65\textheight,keepaspectratio]{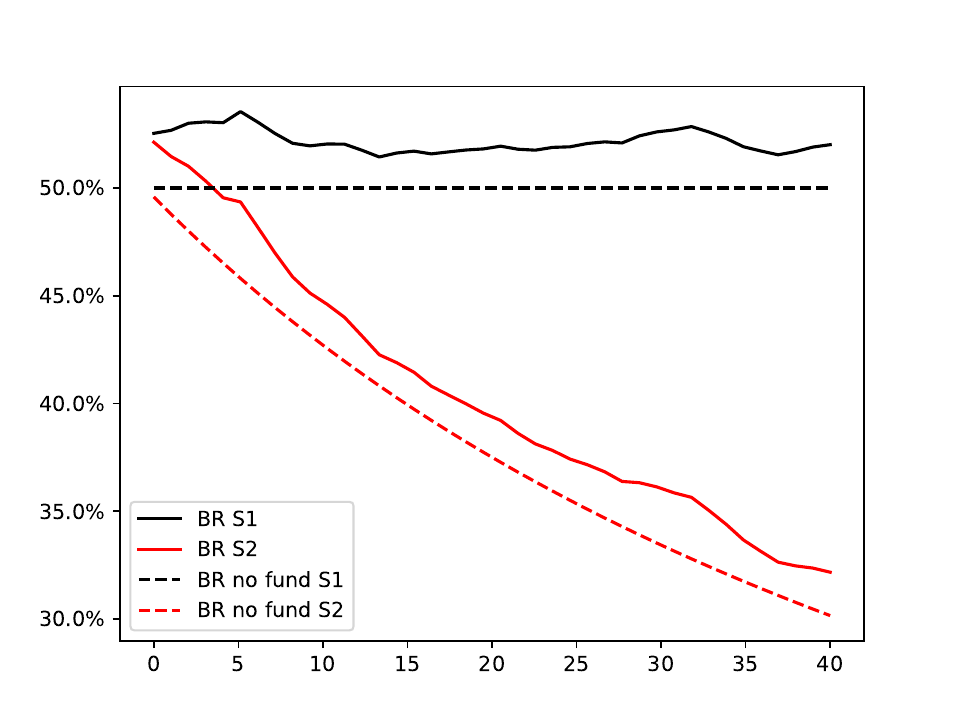}
\vspace{-0.2cm}\caption{Benefit ratio $\text{BR}_t=\frac{p^*_t}{\mathfrak{e}_t}$}\label{subfig:SS_v_BB_seed_3_BR}
\end{subfigure}
\captionsetup{font=footnotesize}
\caption{Comparison Between Steady State (black) and Baby Boom ({\color{red} red}) – Seed 3 (optimistic)\\
}\label{fig:SS_v_BB_seed_3}
\end{figure}

\begin{figure}
\centering
    \setfolder{fig_26feb_delta_base_S1_SS_v_S2_BB_MC}
\begin{subfigure}[t]{0.49\textwidth}
\centering
\includegraphics[width=\textwidth,keepaspectratio]{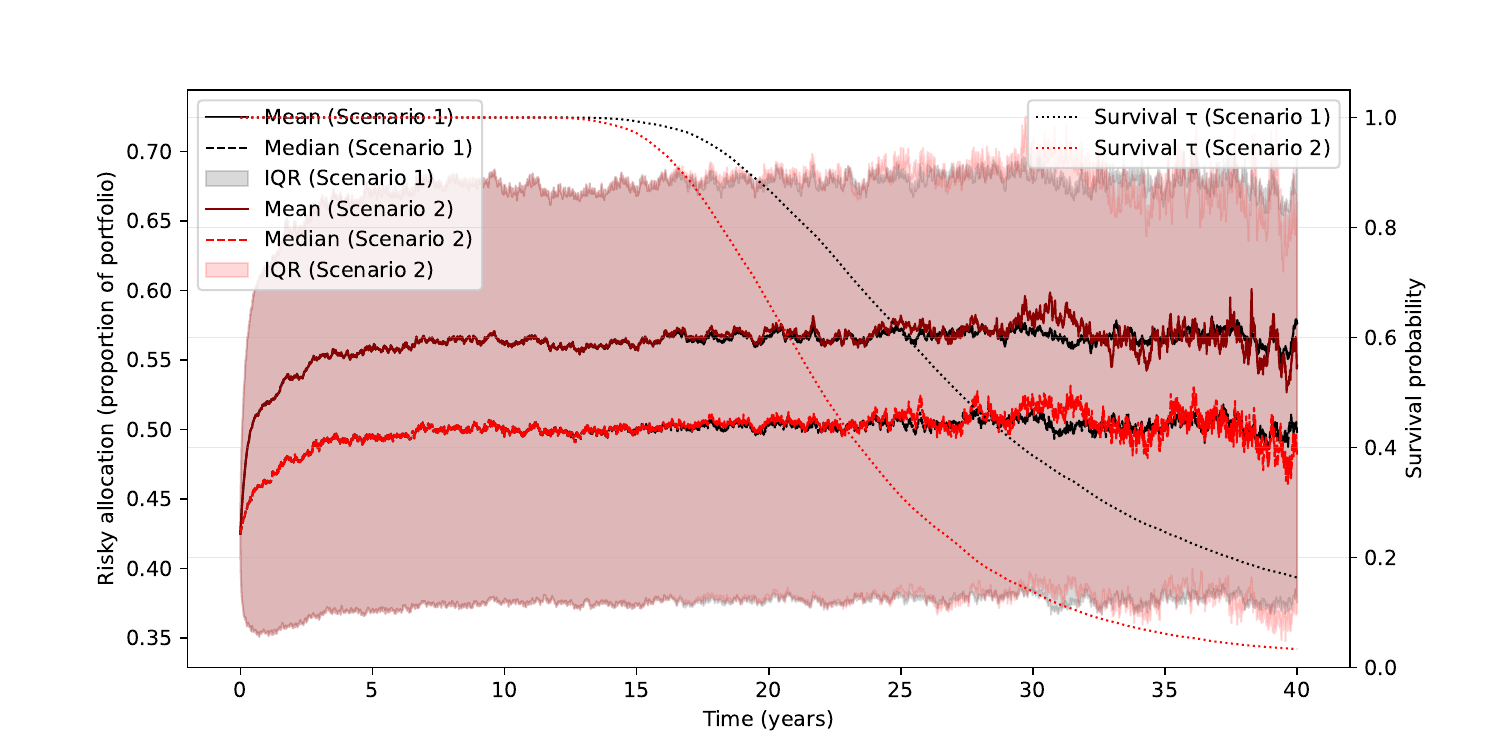}
\vspace{-0.2cm}\caption{Proportion risky investment $\left(\frac{\pi^*_t}{\sqrt{\nu_t}F^*_t}\right)$ \\conditional to $F^*_t>0$}\label{subfig:SS_v_BB_MC_proportion}
\end{subfigure}
\begin{subfigure}[t]{0.49\textwidth}
\centering
\includegraphics[width=\textwidth,keepaspectratio]{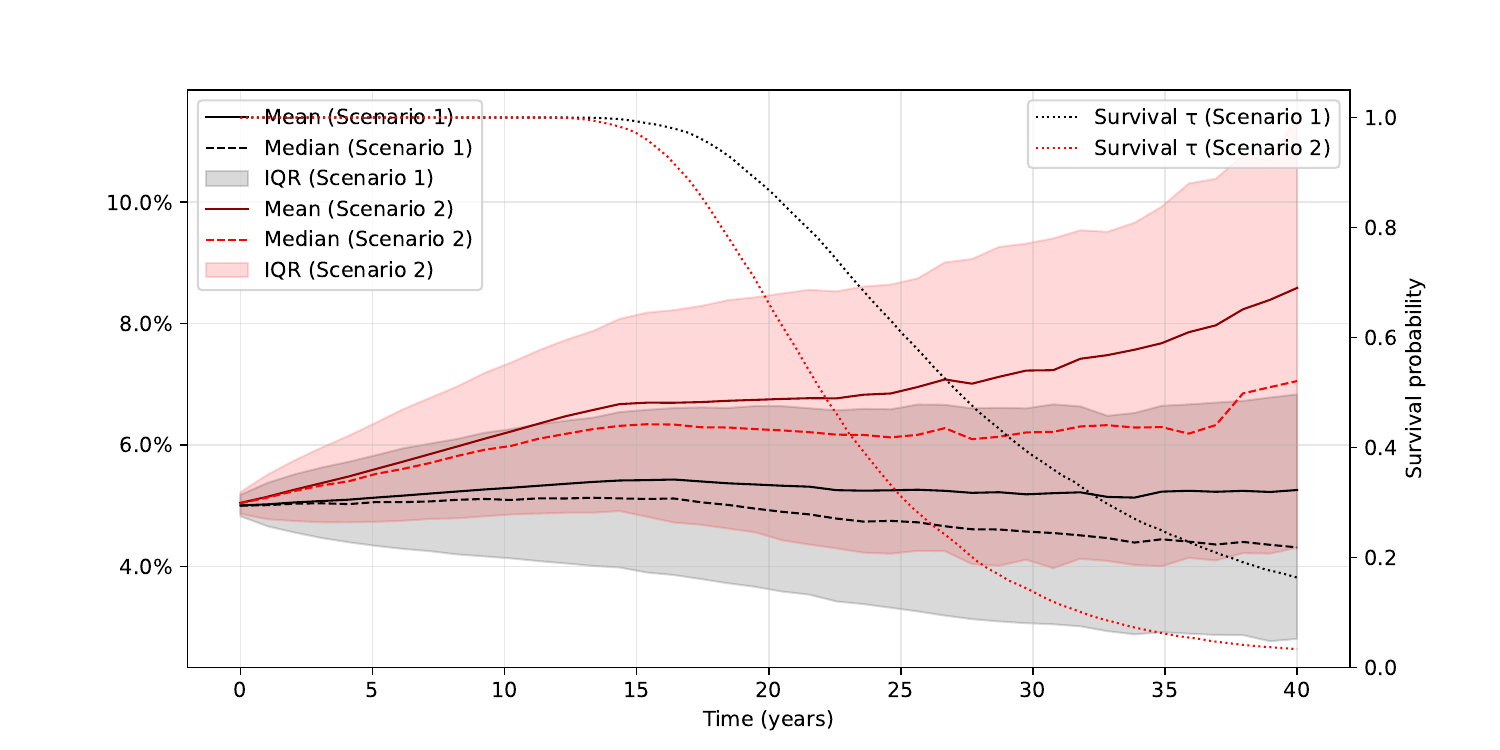}
\vspace{-0.2cm}\caption{Relative surplus $\rho_t=\frac{p^*_t-p^{\min}_t}{p^{\min}_t}$ \\conditional to $F^*_t>0$}\label{subfig:SS_v_BB_MC_relative_surplus}
\end{subfigure}\hfill 
\vspace{0.3cm}
\begin{subfigure}[t]{0.49\textwidth}
\centering
\includegraphics[width=\textwidth,keepaspectratio]{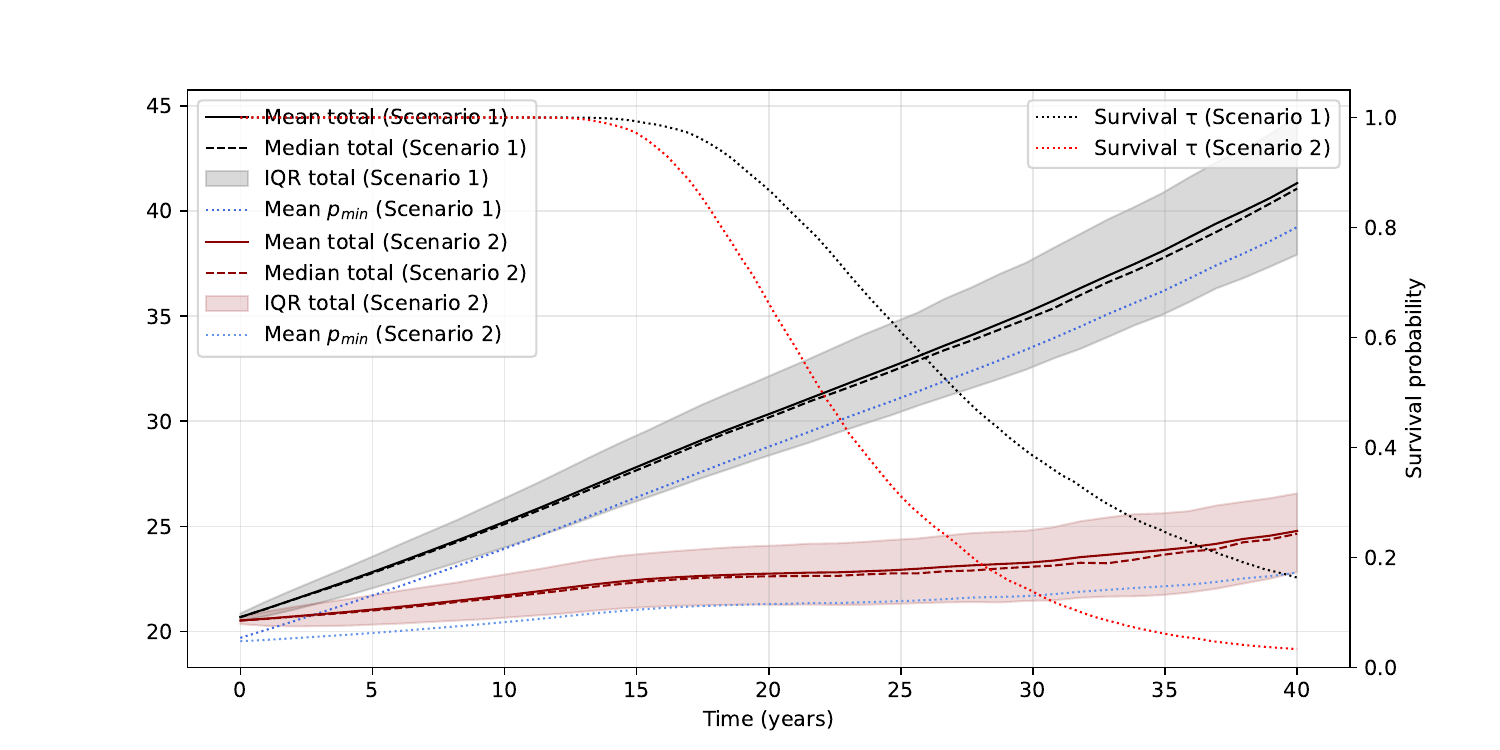}
\vspace{-0.2cm}\caption{Individual pension conditional to $F^*_t>0$}\label{subfig:SS_v_BB_MC_p_tot_withbuffer}
\end{subfigure}
\vspace{0.3cm}
\begin{subfigure}[t]{0.49\textwidth}
\centering
\includegraphics[width=\textwidth,keepaspectratio]{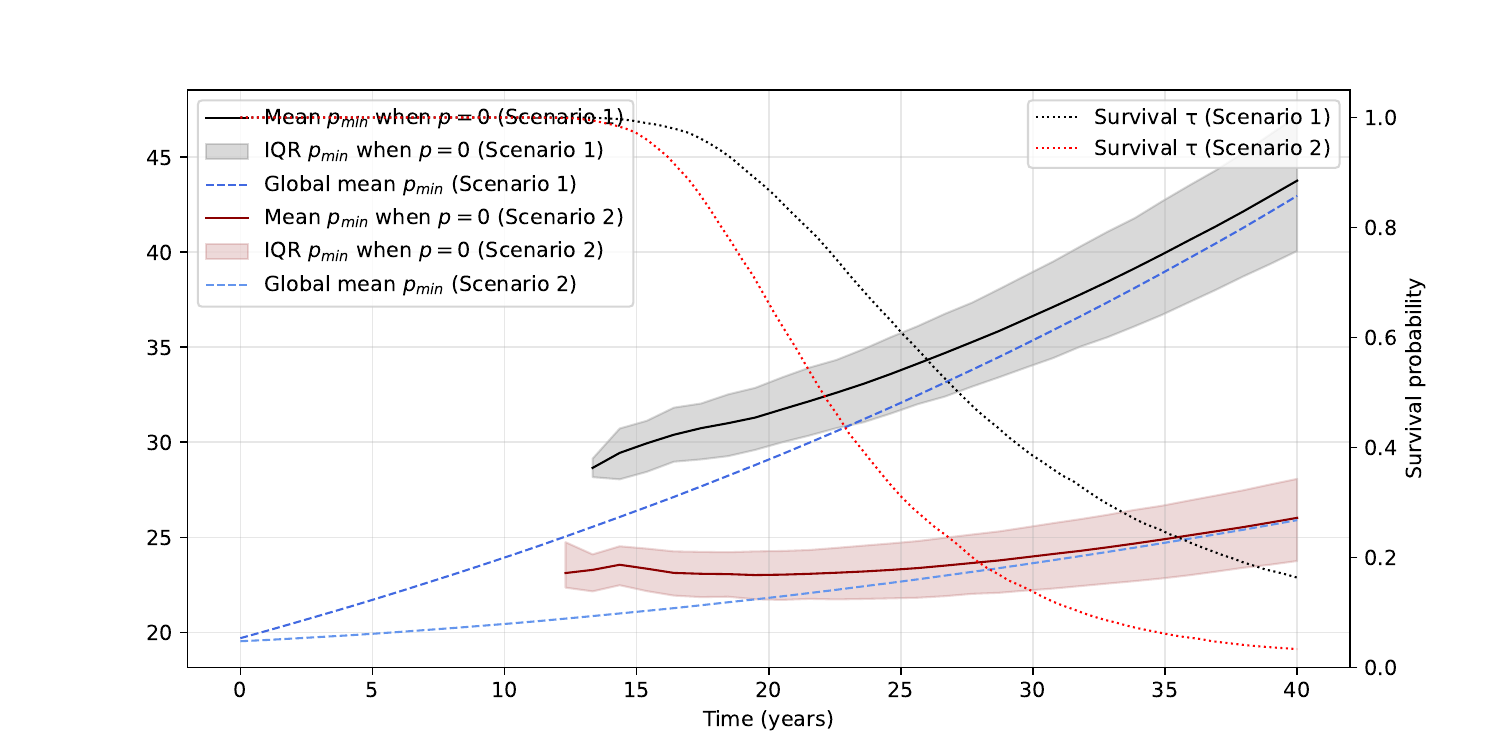}
\vspace{-0.2cm}\caption{Individual pension conditional to $F^*_t=0$}\label{subfig:SS_v_BB_MC_p_tot_withoutbuffer}
\end{subfigure}
\vspace{0.3cm}
\captionsetup{font=footnotesize}
\caption{Comparison Between Steady State (black) and Baby Boom ({\color{red} red}) – Monte Carlo }\label{fig:SS_v_BB_MC}
\end{figure}

\clearpage

\subsection{Sustainability}\label{subsec:sustainability}

This subsection shows the effect of $Z^u_0$, the initial buffer fund weight, and $\delta$, the preference parameters to hedgeable and non-hedgeable risk, on the buffer depletion time $\tau$ under the BB scenario.

\paragraph{Sensitivity to Non-Hedgeable Risk Preferences $\delta$}

\begin{figure}[h!]
\centering
\setfolder{fig_24feb_bb_S1_delta_0_v_S2_delta_base_MC}
\begin{subfigure}[t]{0.49\textwidth}
\centering
\includegraphics[width=\textwidth,height=0.65\textheight,keepaspectratio]{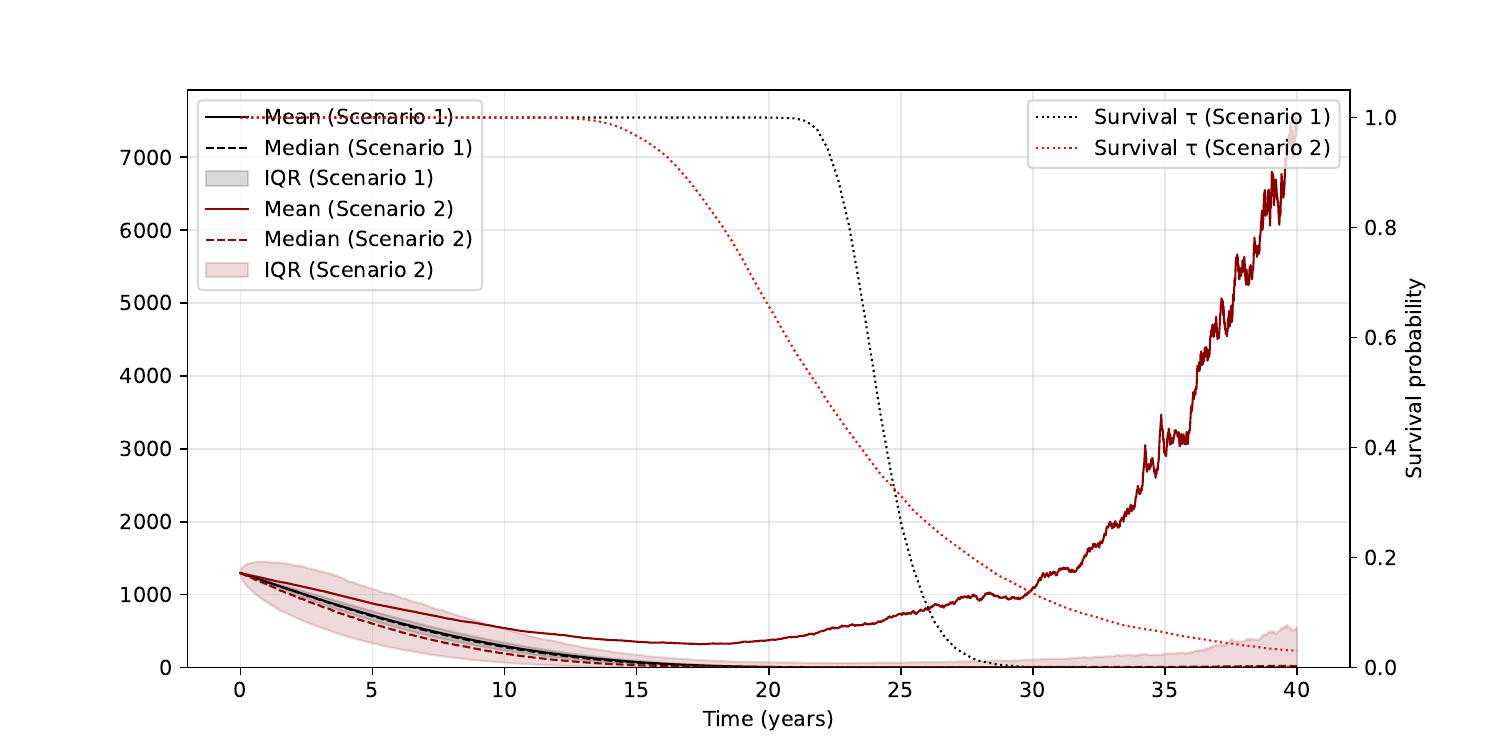}
\vspace{-0.2cm}\caption{Buffer fund utility weight $Z^u_t | Z^u_t > 0 $}\label{subfig:bb_delta_0_v_delta_base_Zut}
\end{subfigure}
\begin{subfigure}[t]{0.49\textwidth}
\centering
\includegraphics[width=\textwidth,height=0.65\textheight,keepaspectratio]{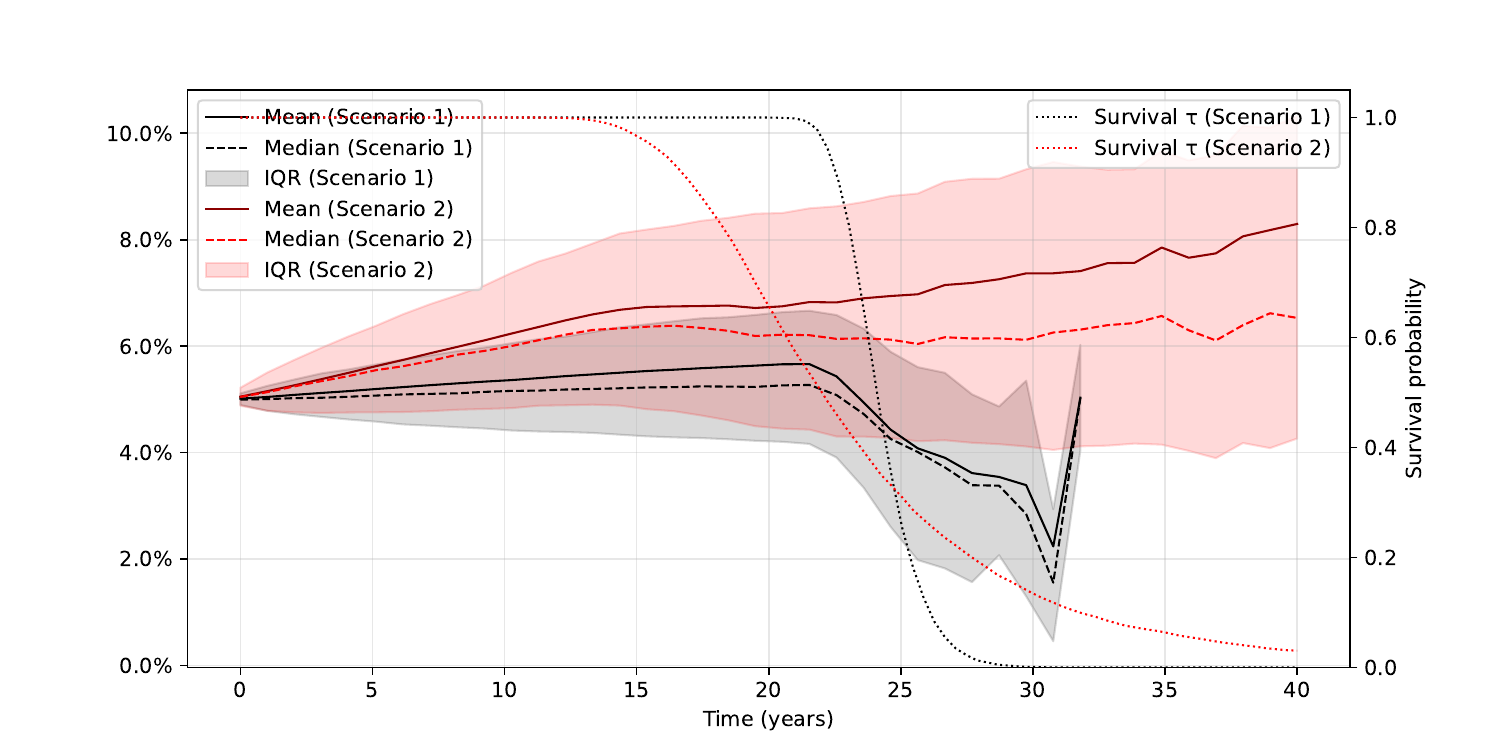}
\vspace{-0.2cm}\caption{Relative surplus $\rho_t=\frac{p^*_t-p^{\min}_t}{p^{\min}_t}$}\label{subfig:bb_delta_0_v_delta_base_rho}
\end{subfigure}
\begin{subfigure}[t]{0.49\textwidth}
\centering
\includegraphics[width=\textwidth,keepaspectratio]{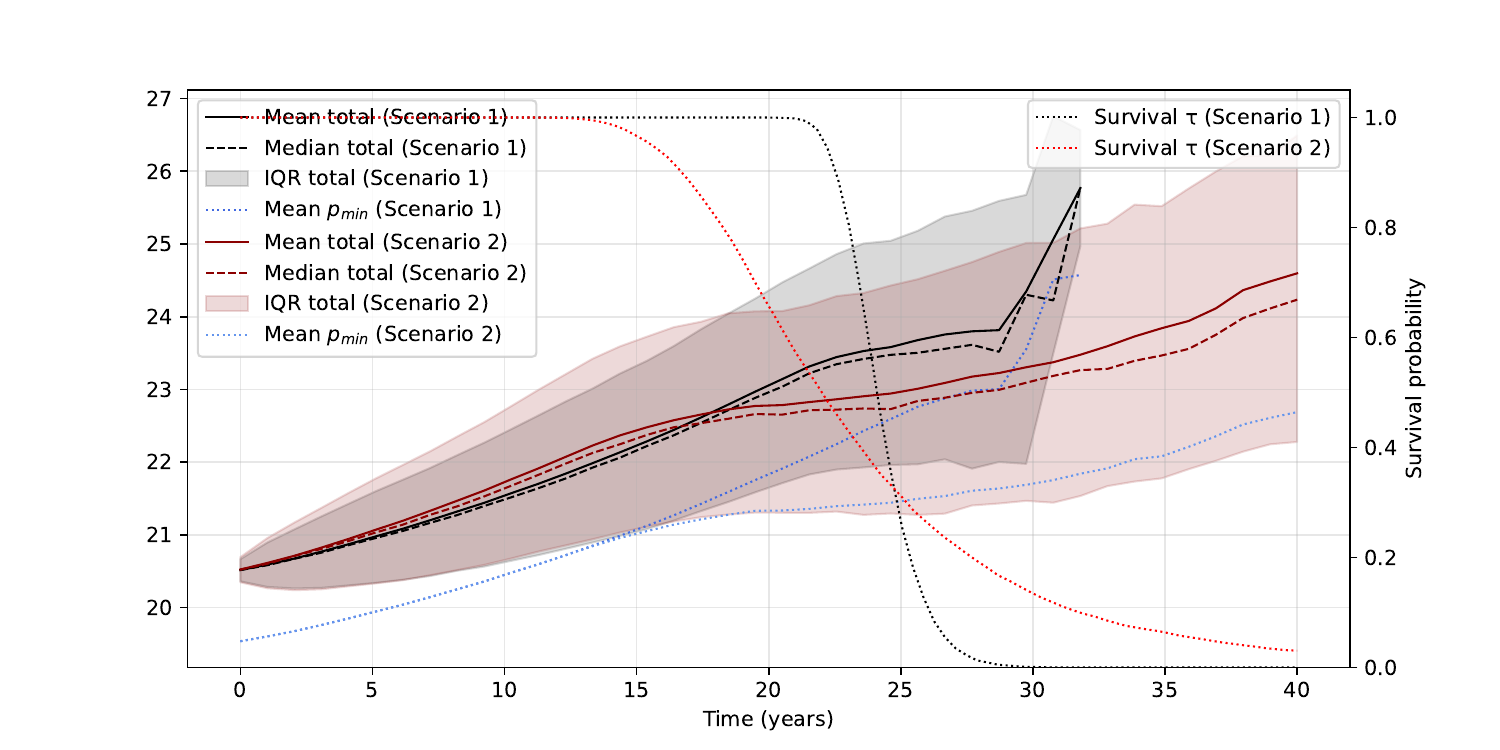}
\vspace{-0.2cm}\caption{Individual pension conditional to $F^*_t>0$}\label{subfig:bb_delta_0_v_delta_base_p_tot}
\end{subfigure}
\vspace{0.3cm}
\begin{subfigure}[t]{0.49\textwidth}
\centering
\includegraphics[width=\textwidth,keepaspectratio]{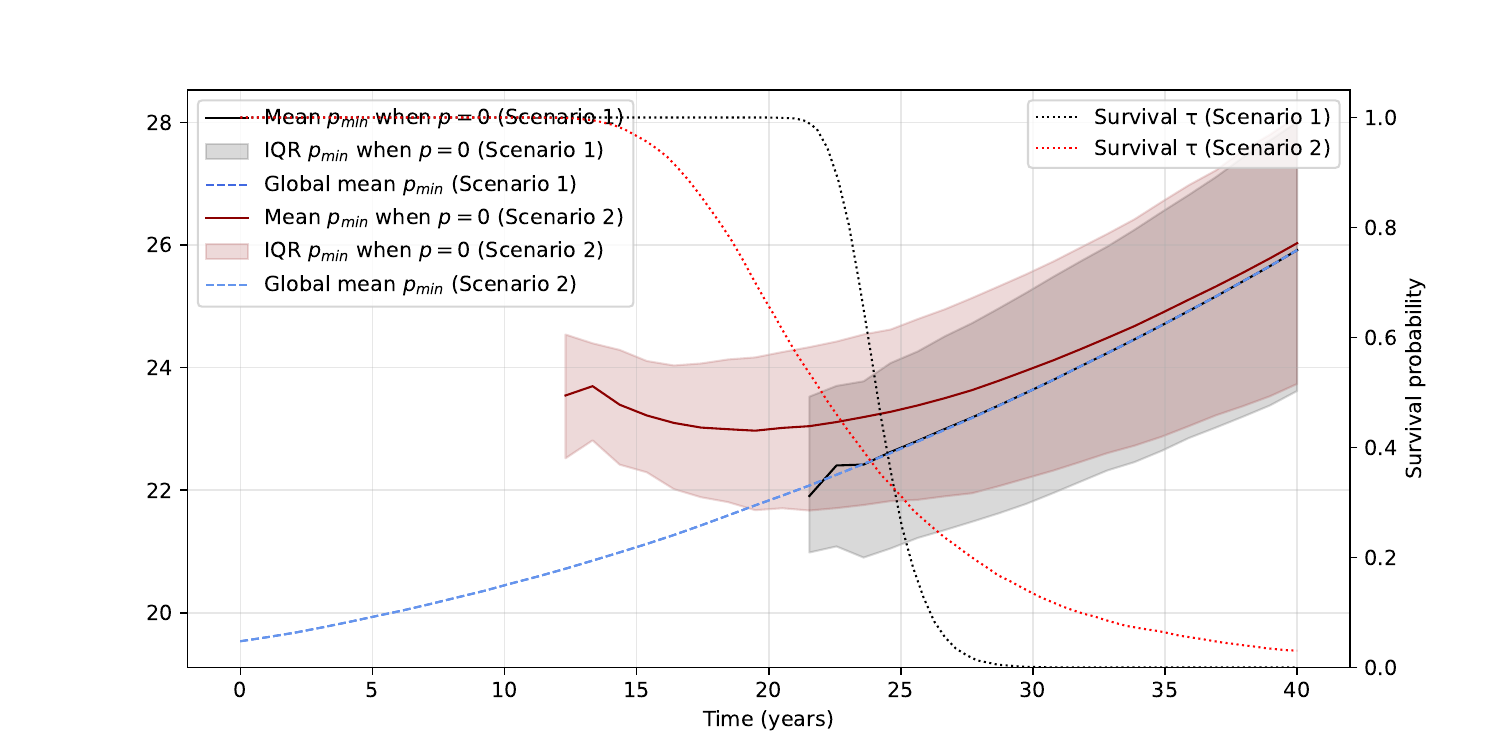}
\vspace{-0.2cm}\caption{Individual pension conditional to $F^*_t=0$}\label{subfig:bb_delta_0_v_delta_base_p_tot_zero}
\end{subfigure}
\captionsetup{font=footnotesize}
\caption{Comparison Between $\delta={}^{t} \! (0,0,0,0) $ (black) and $\delta={}^{t} \! (0, -0.2, -0.2,-0.2) $ ({\color{red} red}) – Monte Carlo}\label{fig:bb_delta_0_v_delta_base}
\end{figure}

\begin{figure}[h!]
\centering
\setfolder{fig_24feb_bb_S1_delta_0_v_S2_delta_base_MC}
\captionsetup{font=footnotesize}
\begin{subfigure}[t]{0.98\textwidth}
\centering
\includegraphics[width=\textwidth,height=0.65\textheight,keepaspectratio]{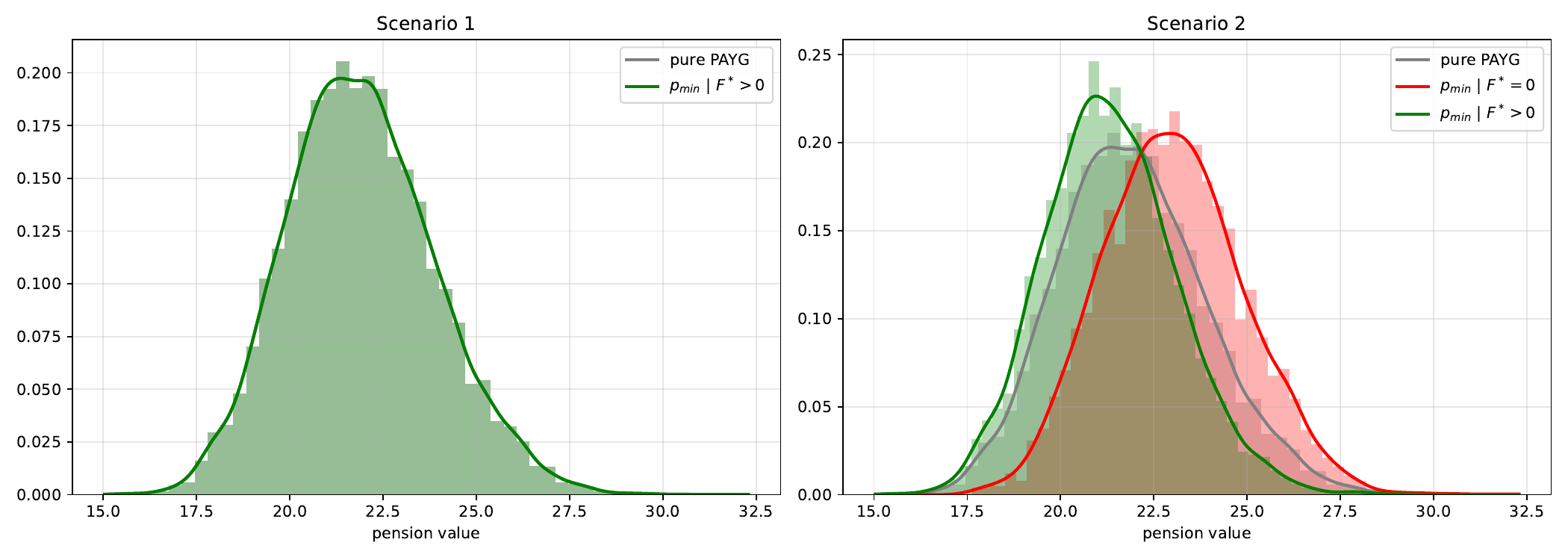}
\vspace{-0.2cm}\caption{Distribution of {minimum pension: }unconditional $p^{\min}$ (gray), and conditional $p^{\min}|F^*>0$ (green), $p^{\min}|F^*=0$ (red) }\label{subfig:bb_delta_0_v_delta_base_dist_pension_pmin}
\end{subfigure}
\begin{subfigure}[t]{0.98\textwidth}
\centering
\includegraphics[width=\textwidth,height=0.65\textheight,keepaspectratio]{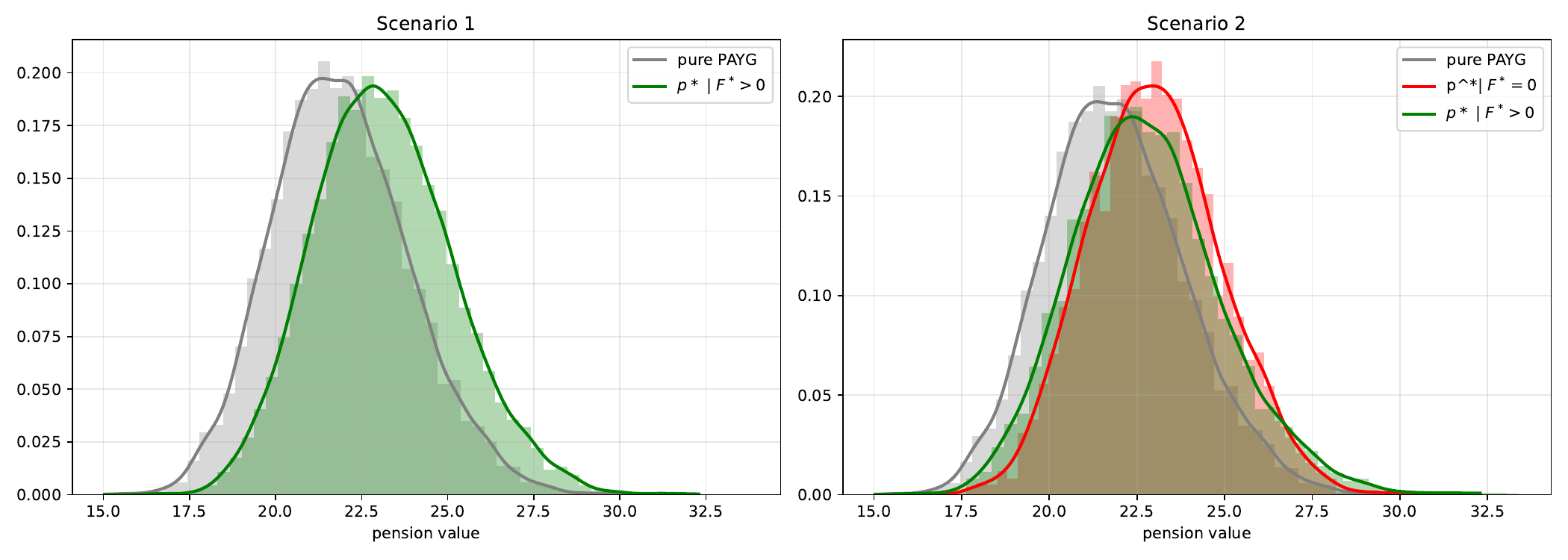}
\vspace{-0.2cm}\caption{Distribution of {total pension: }unconditional $p^{*}$ (gray), and conditional $p^{*}|F^*>0$ (green), $p^{*}|F^*=0$ (red)}\label{subfig:bb_delta_0_v_delta_base_dist_pension_pstar}
\end{subfigure}
\caption{Empirical distribution of unconditional $p^{{j}}$ (gray) and conditional $p^{j}|F^*>0$ (green) and $p^{j}|F^*=0$ ({\color{red} red}) for $j=\min$ (top) and $j=*$ (bottom) at time $t=20$ - comparison $\delta= {}^{t} \!(0,0,0,0) $ ({Scenario 1, }left) and $\delta={}^{t} \! (0, -0.2, -0.2,-0.2) $ ({Scenario 2, }right)}\label{fig:bb_delta_0_v_delta_base_dist_pension}
\end{figure}
\begin{figure}[h!]
\centering
\setfolder{fig_24feb_bb_S1_delta_4_-_v_S1_delta_4_+}
\captionsetup{font=footnotesize}
\begin{subfigure}[t]{0.49\textwidth}
\centering
\includegraphics[width=\textwidth,height=0.65\textheight,keepaspectratio]{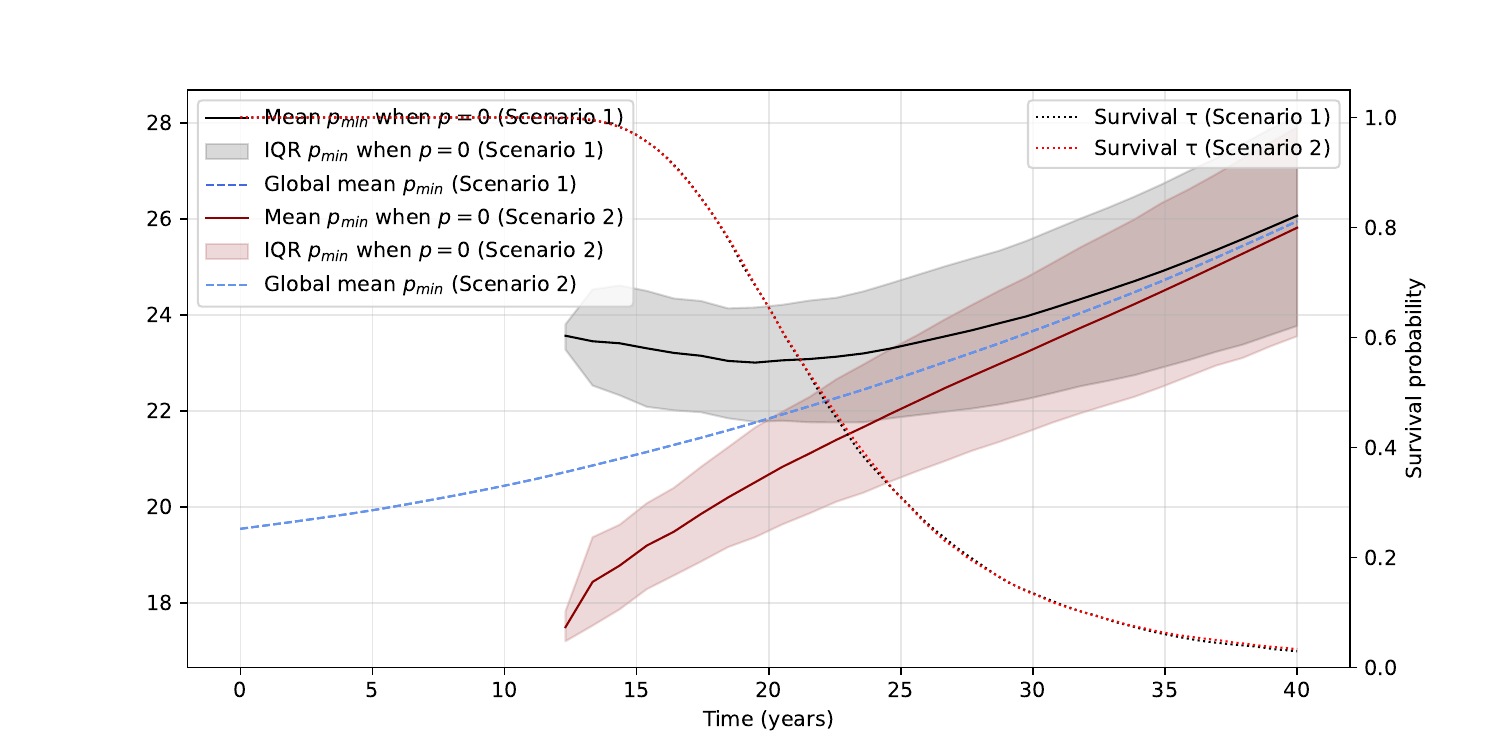}
\vspace{-0.2cm}\caption{Minimum pension conditional on $F^*_t=0$}\label{subfig:bb_delta_4_-_v_delta_4_+_dist_pension_p*_F_eq_0}
\end{subfigure}
\begin{subfigure}[t]{0.49\textwidth}
\centering
\includegraphics[width=\textwidth,height=0.65\textheight,keepaspectratio]{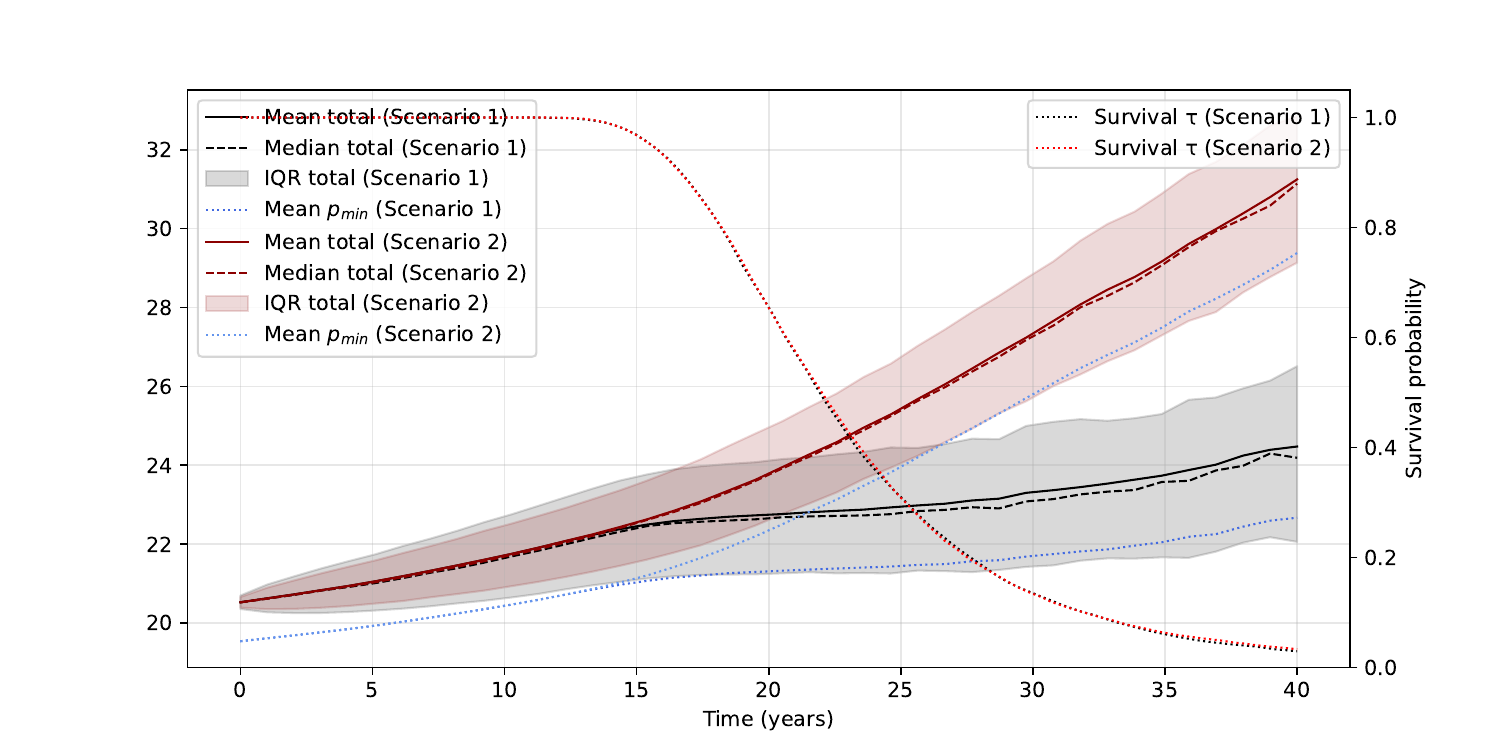}
\vspace{-0.2cm}\caption{Individual pension conditional on $F^*_t>0$}\label{subfig:bb_delta_4_-_v_delta_4_+_dist_pension_p*_F_geq_0}
\end{subfigure}\\
\begin{subfigure}[t]{0.98\textwidth}
\centering
\includegraphics[width=\textwidth,height=0.65\textheight,keepaspectratio]{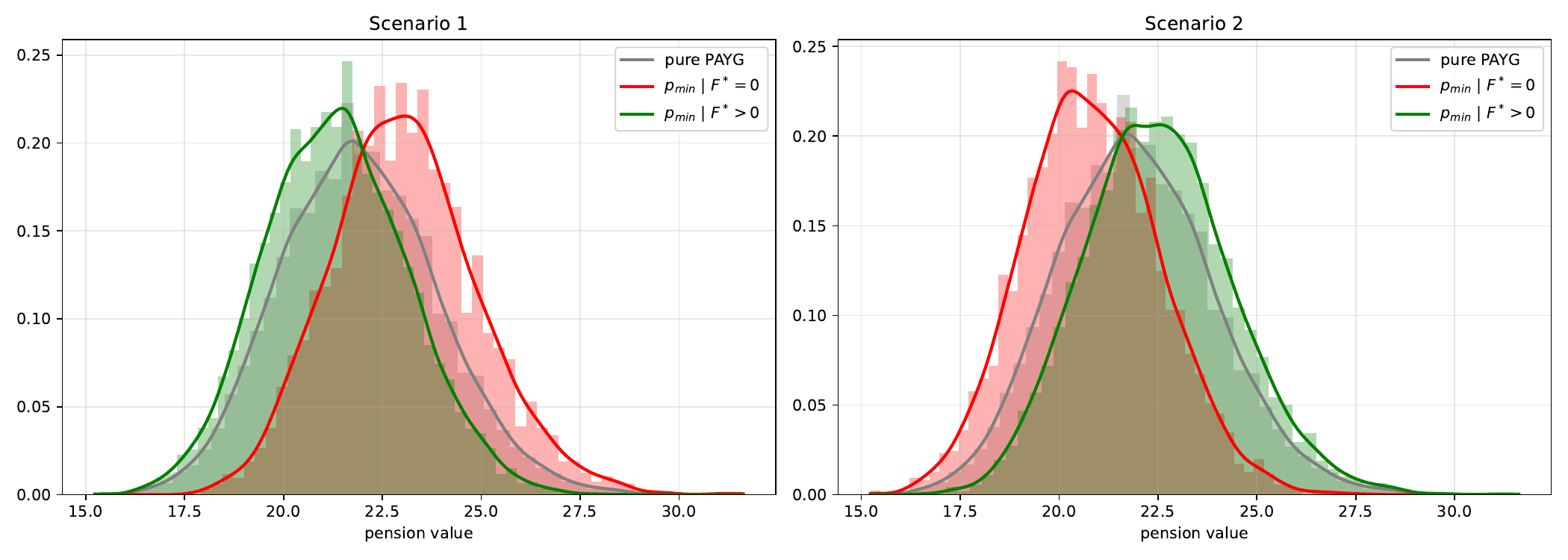}
\vspace{-0.2cm}\caption{Distribution of {minimum pension:} unconditional $p^{\min}$ (gray), and conditional $p^{\min}|F^*>0$ (green), $p^{\min}|F^*=0$ (red) }\label{subfig:bb_delta_4_-_v_delta_4_+_dist_pension_pmin}
\end{subfigure}
\begin{subfigure}[t]{0.98\textwidth}
\centering
\includegraphics[width=\textwidth,height=0.65\textheight,keepaspectratio]{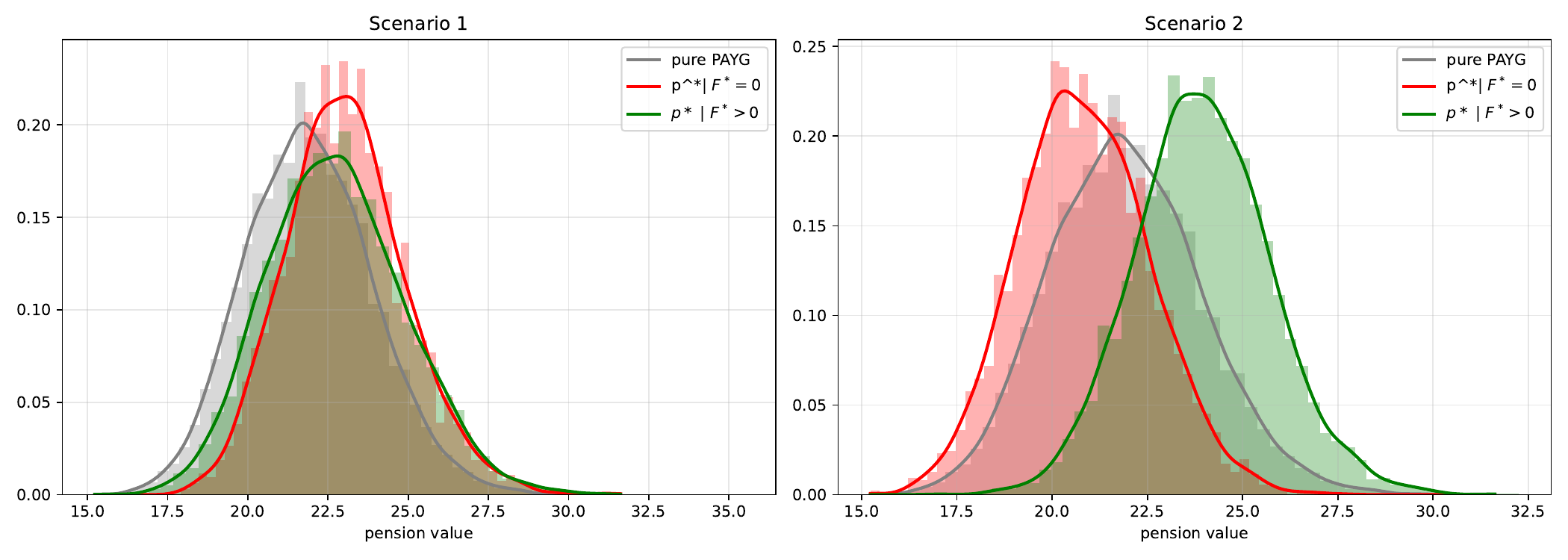}
\vspace{-0.2cm}\caption{Distribution of {total pension:} unconditional $p^{*}$ (gray), and conditional $p^{*}|F^*>0$ (green), $p^{*}|F^*=0$ (red)}\label{subfig:bb_delta_4_-_v_delta_4_+_dist_pension_pstar}
\end{subfigure}
\caption{{Values (top)} and empirical distribution of unconditional $p^{{j}}$ (gray) and conditional $p^{j}|F^*>0$ (green) and $p^{j}|F^*=0$ ({\color{red} red}) for $j=\min$ (middle) and $j=*$ (bottom) at time $t=20$ - comparison $\delta= {}^{t} \!(0, -0.2, -0.2,-0.2) $ ({Scenario 1}) and $\delta={}^{t} \! (0, -0.2, -0.2,+0.2) $ ({Scenario 2})}\label{fig:bb_delta_4_-_v_delta_4_+_dist_pension}
\end{figure}

We now investigate how the buffer fund sensitivity $\delta$ to financial and 
economic processes shapes the distribution of the buffer depletion time and  the pension surplus. 
First, Figure~\ref{fig:bb_delta_0_v_delta_base} illustrates differences in the buffer utility weight $Z^u$ and optimal processes for $\delta={}^t(0,0,0,0)$ (scenario 1, in black), and our base  $\delta={}^t(0,-0.2,-0.2,-0.2)$ (scenario 2, in red), under the BB demographic scenario.\\
When $\delta={}^t(0,0,0,0)$ everywhere, we obtain:
\begin{equation*}
Z_t^u = \xi_t^{-1}\left((Z^u_0)^{\frac{1}{\theta}} - \int_0^t N_s^r(Z_0e^{-\beta s}\xi_s)^{\frac{1}{\theta}} ds \right)^{\theta}, \quad \forall t \in [0,\tau^Z[,
\end{equation*}
with
\begin{align*}
\xi_t^{-1} = \exp\left(-\int_0^t \left((1-\theta)r_s + \frac{(1-\theta)}{2\theta}\eta_s^{2}\right)ds\right).
\end{align*}
Hence, when $\delta={}^t(0,0,0,0)$ $Z^u$ is  a finite variation (zero volatility) decreasing process with stochastic drift, 
 while  non-zero $\delta$ increases variation in $Z^u_t$, as clearly depicted in Figure~\ref{subfig:bb_delta_0_v_delta_base_Zut}.

Incorporating non-zero $\delta$ parameters allows greater responsiveness of the buffer fund weight to market and economic conditions, yielding higher pension surpluses for our chosen parametrization, as shown in Figure~\ref{subfig:bb_delta_0_v_delta_base_rho} and Table \ref{tab:eair_combined} (Panel B for $\delta=0$ versus Panel A (BB) for the base case). While average depletion times are similar (24.25 versus 
23.3 years), the variance increases dramatically from 2 ($\delta = 0$) to 37 (base case), reflecting the 
additional risk exposure via the hedgeable component and risk mitigation for 
non-hedgeable risks induced by non-zero buffer fund utility volatility.

As stated above,  allowing for a non-zero wages risk component $\delta^{\mathfrak{e}}$ introduces a dependence between $p^*$ and $p^{\min}$, so that the conditional distribution of the optimal pension differs according to whether the fund has depleted. This is clearly depicted in Figures~\ref{subfig:bb_delta_0_v_delta_base_p_tot}-\ref{subfig:bb_delta_0_v_delta_base_p_tot_zero}.

When the fund depletes, Figure~\ref{subfig:bb_delta_0_v_delta_base_p_tot_zero} shows that the average  pension in the base case exceeds that of the $\delta=0$ case, since a negative $\delta^{\mathfrak{e}}$ concentrates surplus payments in adverse  economic scenarios. Consequently, when the fund depletes, the remaining scenarios are those in which the minimum pension performs better, yielding a higher average conditional pension than in the $\delta=0$ case.
This is further illustrated in Figure~\ref{fig:bb_delta_0_v_delta_base_dist_pension}, which displays, for both the zero and non-zero volatility cases, the empirical 
distributions of $p^{\min}_t$ and $p^*_t$ under three conditioning sets: 
unconditional, conditional on the buffer fund being solvent ($F^*_t > 0$), and 
conditional on it being depleted ($F^*_t = 0$), at $t = 20$. Figure~\ref{subfig:bb_delta_0_v_delta_base_dist_pension_pmin} shows that under $\delta=0$ the conditional ($F_t^*>0$) and unconditional distributions of $p^{\min}_T$ coincide since $p^{\min}_T$  is independent of the buffer fund depletion time. Under the base case $\delta$, the distribution of $p^{\min}$ when the fund remains solvent is similar to the unconditional case, removing variability at the cost of slightly reduced upside, since this conditional case selects for lower-$p^{\min}$ scenarios. When $F^*=0$, the distribution shifts clearly to the right, confirming that the buffer fund provides surpluses precisely when $p^{\min}$ was lowest. For the total pension $p^*$ (Figure~\ref{subfig:bb_delta_0_v_delta_base_dist_pension_pstar}), introducing a buffer fund is unambiguously beneficial: a clear rightward shift is observed even in the $\delta=0$ case, and the pattern noted for $p^{\min}$ is confirmed for our base case $\delta$.

One might ask whether any non-zero value or sign of $\delta$ produces the same results. The answer is no. Figure~\ref{fig:bb_delta_4_-_v_delta_4_+_dist_pension} illustrates the effect of $\delta^{\mathfrak{e}}=-0.2$ (our base case) versus $\delta^{\mathfrak{e}}=+0.2$, with all other $\delta$ components unchanged.
The overall distribution of $p^*$ are similar.
Conditional on default,  the sign matters considerably. A negative $\delta^{\mathfrak{e}}$ continues to provide surpluses in scenarios where $p^{\min}$ is low, mitigating adverse economic outcomes. A positive $\delta^{\mathfrak{e}}$ does the opposite: it concentrates surpluses in scenarios where $p^{\min}$ is already high, exacerbating extremes rather than mitigating them. This is particularly striking in Figure~\ref{subfig:bb_delta_4_-_v_delta_4_+_dist_pension_pstar}, where a positive $\delta^{\mathfrak{e}}$ shifts the conditional distribution markedly to the left, implying that default scenarios leave retirees worse off than under no buffer fund at all. Note that this conditional dependence is specific to $\delta^{\mathfrak{e}}$, since $p^{\min}$ is directly linked to wages. The signs of $\delta^{r}$ or $\delta^{\nu}$ do not induce this type of dependence on the surplus, though $\delta^{\nu}$ naturally affects distributions through its influence on portfolio allocation.

\paragraph{Impact of initial buffer fund utility weight $Z^u_0$ on buffer fund depletion time}

\begin{figure}[ht!]
\centering
\setfolder{fig_27feb_bb_delta_base_MC_tau_vs_Z0u}
\begin{subfigure}[t]{0.46\textwidth}
\centering
\includegraphics[width=\textwidth,height=0.34\textheight,keepaspectratio]{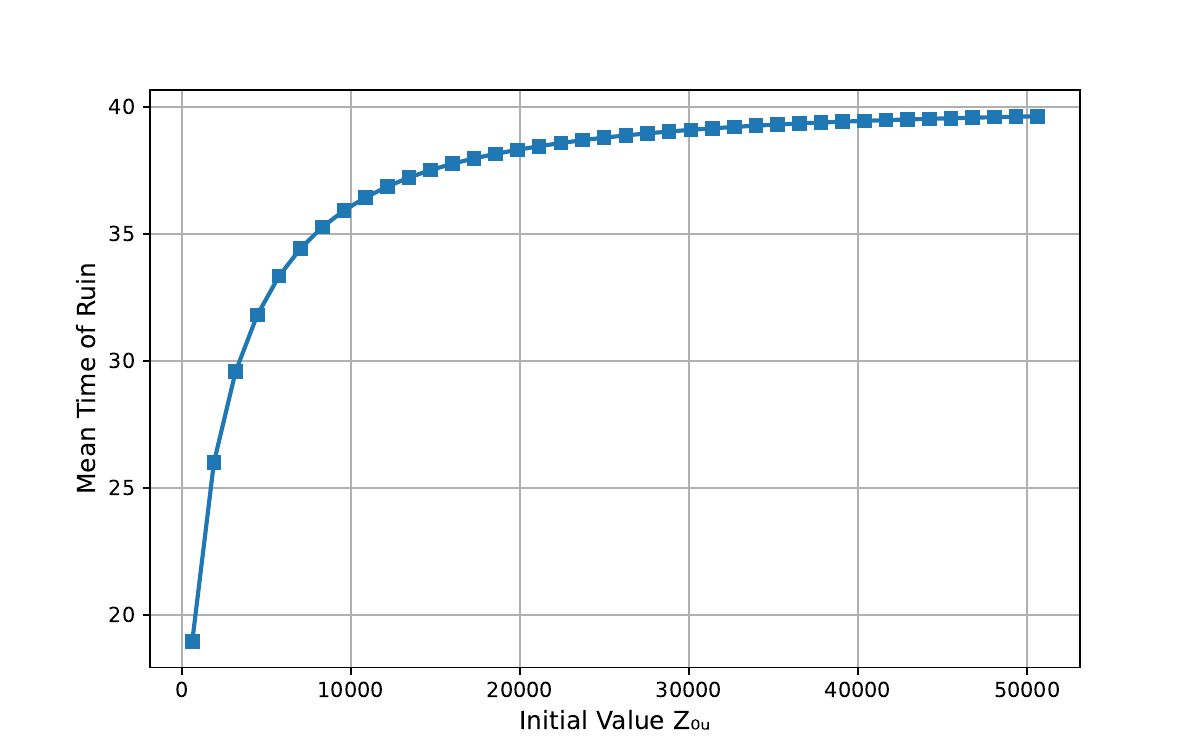}
\vspace{-0.2cm} \caption{Mean buffer depletion time $\E[\tau]$ versus initial buffer utility weight $Z^u_0$}\label{subfig:BB_tau_v_pension_sens_Z0u_mean_tau}
\end{subfigure}\hfill
\begin{subfigure}[t]{0.46\textwidth}
\centering
\includegraphics[width=\textwidth,height=0.34\textheight,keepaspectratio]{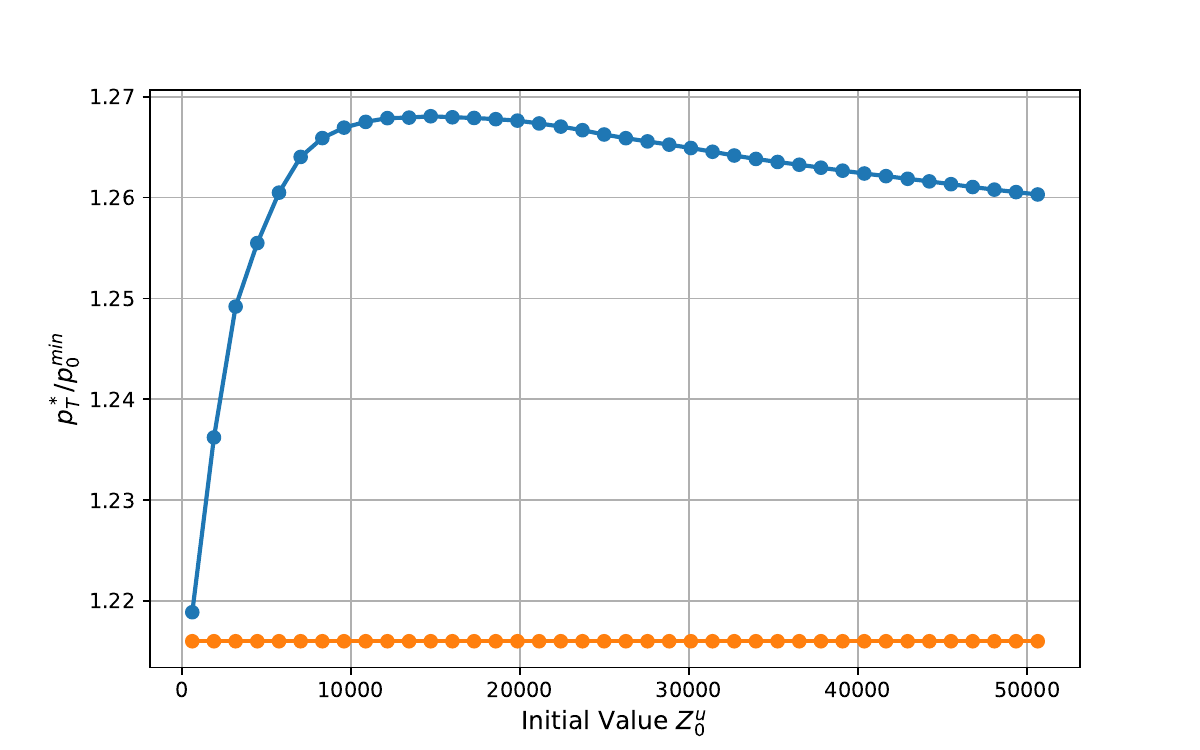}
\vspace{-0.2cm} \caption{Average  standardized individual pension $\E[p^*_T/p^{\min}_0]$ and  sustainable pension $\E[p^{\min}_T/p^{\min}_0]$  at time $T=30$}\label{subfig:BB_tau_v_pension_sens_Z0u_mean_p_tot}
\end{subfigure}
\vspace{0.3cm}
\captionsetup{font=footnotesize}
\caption{Buffer depletion time and average pension  for varying $Z^u_0$}\label{fig:BB_tau_v_pension_sens_Z0u}
\end{figure}

In the previous section, we calibrated $Z^u_0$ to yield an initial pension exceeding the sustainable pension $p^{\min}_0$ by 5\%. We now consider 
alternative calibrations where the initial surplus is set equal to 6\% and 4\%, yielding an initial buffer fund utility weight $Z^u_0=Z_0 \left(\frac{N^r_0}{0.06}\right)^{\theta}=Z_0 \left(16.67N^r_0\right)^{\theta}$ and $Z_0  \left(\frac{N^r_0}{0.04}\right)^{\theta}=Z_0  \left(25N^r_0\right)^{\theta}$, respectively, while the initial buffer fund $F_0$ remains unchanged. Pathwise results for the three initial conditions under the optimistic economic scenario are presented in Figure~\ref{fig:BB_sens_Z0u} in  Appendix \ref{app:sustainability}.
From the expression of $Z^u_t$ \eqref{eq:SDEZut_part}, we establish that a larger $Z^u_0$ implies a larger $Z^u_t$ for all $t>0$, increasing the fund's utility weight relative to retirees (Figure \ref{subfig:BB_sens_Z0u_Zut}). This translates to lower pension surpluses since the optimal pension $p^*_t$ in Equation \eqref{eq:p*} is a decreasing function of $Z_0^u$ (Figure \ref{subfig:BB_sens_Z0u_pstar}-\ref{subfig:BB_sens_Z0u_rhostar}). On the other hand, the buffer depletion time increases with $Z^u_0$, as shown by Equation~\eqref{deftau} and illustrated in 
Figure~\ref{subfig:BB_tau_v_pension_sens_Z0u_mean_tau}. This  creates a  trade-off between fund sustainability and pension adequacy: 
while $Z^u_0$ can be calibrated to satisfy a pre-specified policy objective (5\% in our base case), such a generous initial threshold may accelerate pension payments and induce early buffer depletion  in an ageing context. We recall that under the power utility framework, the buffer depletion time 
$\tau$ does not depend on the initial buffer fund value $F_0$.\\
Alternatively, policymakers could identify an initial weight $Z^u_0$ that balances the trade-off between sustainability and adequacy over a specified horizon. Figure \ref{subfig:BB_tau_v_pension_sens_Z0u_mean_p_tot} shows the average pension at $T=30$ years and for varying $Z^u_0$, that is 
\begin{equation*}
\E[p^*_T] = \E[p^{\min}_T] + \E[ p^*_T - p^{\min}_T| \tau \geq T] \,  \P(\tau \geq T). 
\end{equation*}
Lower  $Z^u_0$ allows higher pension surpluses conditional on survival 
but increases buffer depletion probability, which can substantially reduce the average terminal pension.  By examining average $p^*_T$, we implicitly weight outcomes for $Z^u_0$  with higher buffer fund survival  probabilities.

At $T = 30$, the maximum average pension is obtained for $Z^u_0 = 14{,}727$, 
corresponding to an initial pension surplus of approximately $2.75\%$ of the  initial PAYG pension $p^{\min}_0$. This is lower than the $5\%$ surplus  obtained under our baseline choice $Z^u_0 = 1{,}296$. However,  the average buffer depletion time increases from just over 23 to 37.5 years, yielding a  substantially more sustainable fund. A natural way to improve pension adequacy while maintaining the same level of 
fund sustainability is to increase the initial buffer fund value $F_0$. This is further discussed in Section \ref{subsec:adequacy}.

\subsection{Adequacy} 
\label{subsec:adequacy}

\paragraph{Sensitivity to Initial Fund Size $F_0$}

Here we analyze jointly the effect of the initial buffer fund level $F_0$ and the wages growth parameter $\lambda$. We fix $Z^u_0$ equal to the base case value $Z^u_0=1296$, which yields a 5\% surplus target when $F_0=C_0$. Fixing $Z^u_0$ ensures identical buffer depletion times $\tau$ and enables fair comparisons across varying wages growth rates and initial fund levels.

As established theoretically, the initial fund size $F_0$ does not affect the buffer depletion time $\tau$, only pension adequacy. Indeed, $F_0$ appears nowhere in Equation~\eqref{deftau}, neither directly nor indirectly, for pre-specified $Z^u_0$. This is policy-relevant: it implies that the buffer fund intervention remains feasible and delivers pension surpluses regardless of the initial level of pre-existing Public Pension Reserve Funds or, for new funds, public debt the government is willing to incur. That said, $F_0$ is not inconsequential for adequacy. Higher $F_0$ yields larger buffer funds and subsequently higher nominal pension payments, as shown directly in Equation~\eqref{eq:p*}.

Furthermore, fixing $Z^u_0$ across different values of $F_0$ naturally yields varying initial pension surplus objectives through the relationship
$$\left(\frac{Z^u_0}{Z_0}\right)^{\frac{1}{\theta}}=\frac{F_0 - \mathfrak{K}_0}{p_0 - p^{\min}_0}.$$
For fixed $Z^u_0$, $Z_0$, $\theta$, and $p^{\min}_0$, a larger (smaller) $F_0$ thus implies a higher (lower) initial pension surplus target $p_0$.

Table~\ref{tab:eair_combined} (Panel C) reports the average equivalent annual indexation rate for three wages' growth levels: $\lambda=0.03$ (top), $\lambda=0.02$ (middle), and $\lambda=0.01$ (bottom). Under $\lambda=0.02$ and $\lambda=0.03$, rising wages partially offset ageing-induced total income from contributions erosion, allowing nominal pensions to grow over time.\footnote{Figure \ref{fig:BB_sens_F0_three_wages} shows the average individual pension $p^*$ and relative pension surplus $\rho$ for varying $\lambda$ and $F_0$.} Under sufficiently low $\lambda$, wages fail to compensate for demographic decline and the minimum sustainable pension $p^{\min}_t$ falls in nominal terms. The benefit of the buffer fund is considerably larger in low-growth environments: the additional average yield it provides is greatest when $\lambda$ is low. Including a buffer fund equal to the initial contribution income $C_0$ can nearly double the yield in some cases, though unless the initial fund is sufficiently large, demographic ageing alone prevents it from reaching the level implied by wages' growth.

We are interested in further analyzing the effect of intergenerational weights $\omega_t$ and initial fund level on the adequacy of pensions as studied in Subsection \ref{subsec:sustainability_adequacy}. In particular, Table \ref{tab:eair_combined} shows the equivalent annual indexation rate at different times for alternative parametrization. 
\begin{table}[ht!]
\centering \smallerthansmall
\caption{Equivalent Annual Indexation Rate (EAIR) with buffer fund ($\overline{y^{\text{BF}}_t}$) and without ($\overline{y^{\min}_t}$) for two demographic scenarios (SS vs BB) [Panel A], difference preferences parameters in the baby boom case [Panel B] and for three wages growth levels $\lambda\in\{0.01,0.02,0.03\}$ and time horizon $t\in\{20, 30, 40\}$ [Panel C] (in \%)}\label{tab:eair_combined}
Panel A: Demographic shock: SS versus BB\\[0.5em]
\begin{tabular}{lcccc|cccc}
\hline \hline
& \multicolumn{4}{c|}{Steady state} & \multicolumn{4}{c}{Baby boom}\\ \cline{2-9}
Measure & $t=10$ & $t=20$ & $t=30$ & $t=40$ & $t=10$ & $t=20$ & $t=30$ & $t=40$\\
\hline
$\overline{y^{\text{BF}}_t}$& 3.09 & 2.51 &2.28 & 2.16 & 1.61 & 1.07& 0.86& 0.78\\
$\overline{y^{\min}_t}$& 2.00 & 2.00 & 2.00 & 2.20 & 0.41 & 0.49 & 0.55& 0.61\\
$\Delta$ & 1.09 & 0.51 & 0.28 & 0.04 & 1.20 & 0.58 & 0.31 & 0.17 \\
\hline \hline 
\end{tabular}

\bigskip

Panel B: BB: alternative preferences' sensitivities $\delta$ and weights $\omega_t$\\[0.5em]
\begin{tabular}{lcccc|cccc}
\hline \hline
& \multicolumn{4}{c|}{$\delta={}^{t} \!(0,0,0,0)$} &\multicolumn{4}{c}{$\omega=\frac{DR_t}{DR_0}$}\\ \cline{2-9}
Measure & $t=10$ & $t=20$ & $t=30$ & $t=40$ & $t=10$ & $t=20$ & $t=30$ & $t=40$\\
\hline
$\overline{y^{\text{BF}}_t}$&1.51 & 1.02 &0.83 & 0.76 & 1.64 & 1.09 & 0.86 & 0.78\\
$\overline{y^{\min}_t}$&0.41 & 0.48 & 0.55 & 0.61& 0.42 & 0.49& 0.56 & 0.62\\
$\Delta$ & 1.10 & 0.54& 0.28 & 0.15 & 1.22 & 0.60 & 0.30 & 0.16\\
\hline \hline 
\end{tabular}

\bigskip

Panel C: BB: $F_0\in\{0,\frac{1}{2},1,\frac{3}{2}\}\cdot C_0$ and $\lambda\in\{0.01,0.02,0.03\}$\\[0.5em]
\begin{tabular}{lcccccccc}
\hline \hline
& $\overline{y^{\text{BF}}_t}$ & $\Delta$ 
& $\overline{y^{\text{BF}}_t}$ & $\Delta$
& $\overline{y^{\text{BF}}_t}$ & $\Delta$
& $\overline{y^{\text{BF}}_t}$ & $\Delta$ \\ 
\hline
$\lambda=0.03$& \multicolumn{2}{c}{$t=10$} & \multicolumn{2}{c}{$t=20$} & \multicolumn{2}{c}{$t=30$} & \multicolumn{2}{c}{$t=40$}\\ \hline 
$\overline{y^{\min}_t}$ & \multicolumn{2}{c}{1.42} & \multicolumn{2}{c}{1.50} & \multicolumn{2}{c}{1.57} & \multicolumn{2}{c}{1.64} \\ \hline
$F_0=\frac{1}{2}\cdot C_0$  & 1,99 & 0,57& 1,76 & 0,26& 1,70 & 0,13& 1,70 & 0,07\\
$F_0= 1\cdot C_0$                  & 2,54 & 1,13& 2,02 & 0,52& 1,82 & 0,25& 1,77 & 0,13  \\ 
$F_0=\frac{3}{2}\cdot C_0$                  &  3,08 & 1,67& 2,27 & 0,77& 1,94 & 0,37& 1,83 & 0,19  \\ 
\hline
$\lambda=0.02$& \multicolumn{2}{c}{$t=10$} & \multicolumn{2}{c}{$t=20$} & \multicolumn{2}{c}{$t=30$} & \multicolumn{2}{c}{$t=40$}\\ \hline 
$\overline{y^{\min}_t}$ & \multicolumn{2}{c}{0.40} & \multicolumn{2}{c}{0.48} & \multicolumn{2}{c}{0.55} & \multicolumn{2}{c}{0.61} \\ \hline
$F_0= \frac{1}{2}\cdot C_0$                  & 1,01 & 0,61& 0,78 & 0,30& 0,71 & 0,16& 0,70 & 0,09  \\ 
$F_0=1\cdot C_0$                  & 1,60 & 1,20& 1,07 & 0,59& 0,86 & 0,31& 0,78 & 0,17\\ 
$F_0= \frac{3}{2}\cdot C_0$                  & 2,17 & 1,77& 1,35 & 0,87& 1,00 & 0,45& 0,86 & 0,25 \\ 
\hline
$\lambda=0.01$& \multicolumn{2}{c}{$t=10$} & \multicolumn{2}{c}{$t=20$} & \multicolumn{2}{c}{$t=30$} & \multicolumn{2}{c}{$t=40$}\\ \hline 
$\overline{y^{\min}_t}$  & \multicolumn{2}{c}{$-0.58$} & \multicolumn{2}{c}{$-0.51$} & \multicolumn{2}{c}{$-0.45$} & \multicolumn{2}{c}{$-0.39$} \\ \hline
$F_0=\frac{1}{2}\cdot C_0$                  &  0,06  & 0,64& -0,17 & 0,34& -0,26 & 0,19& -0,28 & 0,11\\
$F_0= 1\cdot C_0$                  & 0,68  & 1,26& 0,15  & 0,66& -0,08 & 0,37& -0,17 & 0,22 \\
$F_0= \frac{3}{2}\cdot C_0$                  & 1,28  & 1,86& 0,46  & 0,97& 0,09  & 0,54& -0,07 & 0,32 \\
\hline \hline 
\end{tabular}\\
\raggedright \footnotesize{\textit{Notes}: $\Delta$ represents the difference $\overline{y^{\text{BF}}_t}-\overline{y^{\min}_t}$.}
\end{table}

\paragraph{Sensitivity to Intergenerational Weights $\omega_t$}
Figure~\ref{fig:omega_1_vs_omega_DR_MC} (and Figure~\ref{fig:omega_1_vs_omega_DR_MC_appendix} in  Appendix \ref{app:adequacy}) compares two weighting schemes: the baseline where all retirees receive equal weight ($\omega_t=1$), and an alternative where the individual retiree weight is given by
$$\omega_t=\frac{N^r_t}{N^r_0},$$
which increases from 1 to 1.67 in the BB scenario over the simulation horizon. 

The alternative scheme yields greater relative surplus $\rho$ (Figure~\ref{subfig:omega_1_vs_omega_DR_MC_rho}) and benefit ratio increases (Figure~\ref{subfig:omega_1_vs_omega_DR_MC_benefit_ratio_increase}), without meaningfully affecting default: relative increases of approximately 1\% in $\rho$ and 0.5\% in the benefit ratio are achievable. Table~\ref{tab:eair_combined} (Panel B) further confirms the improvement in equivalent indexation rate, with additional compounded returns of 0.03\% by $t=10$ and 0.02\% by $t=20$. Over a longer horizon, however, these gains diminish as default becomes increasingly likely beyond $t=30$.

These results align naturally with our theoretical findings. Recall from Equation~\eqref{eq:p*} that the proportion of the buffer surplus paid out is given by $\left(\frac{Z_t\omega_t}{Z^u_t}\right)^{1/\theta}$. Under an ageing scenario, $\frac{N^r_t}{N^r_0}>1$, so the numerator is larger and a greater share of the fund is distributed under the alternative scheme. However, since $\omega$ enters through $\omega^{1/\theta}$, the increase from 1 to 1.67 in weight translates to only a modest change in $\omega^{1/\theta}$ from 1 to 1.14 by $T=40$. Combined with slightly elevated default risk and the weight effect materializing precisely when fund survival is unlikely, the overall impact of the $\omega_t$ specification remains limited. Nevertheless, the scheme effectively directs greater surpluses toward cohorts facing higher dependency ratios, providing targeted protection where demographic pressure is most acute.

\begin{figure}[ht!]
\centering
\setfolder{fig_24feb_delta_base_S1_omega_1_v_S2_omega_DR_MC}
\begin{subfigure}[t]{0.49\textwidth}
\centering
\includegraphics[width=\textwidth,height=0.65\textheight,keepaspectratio]{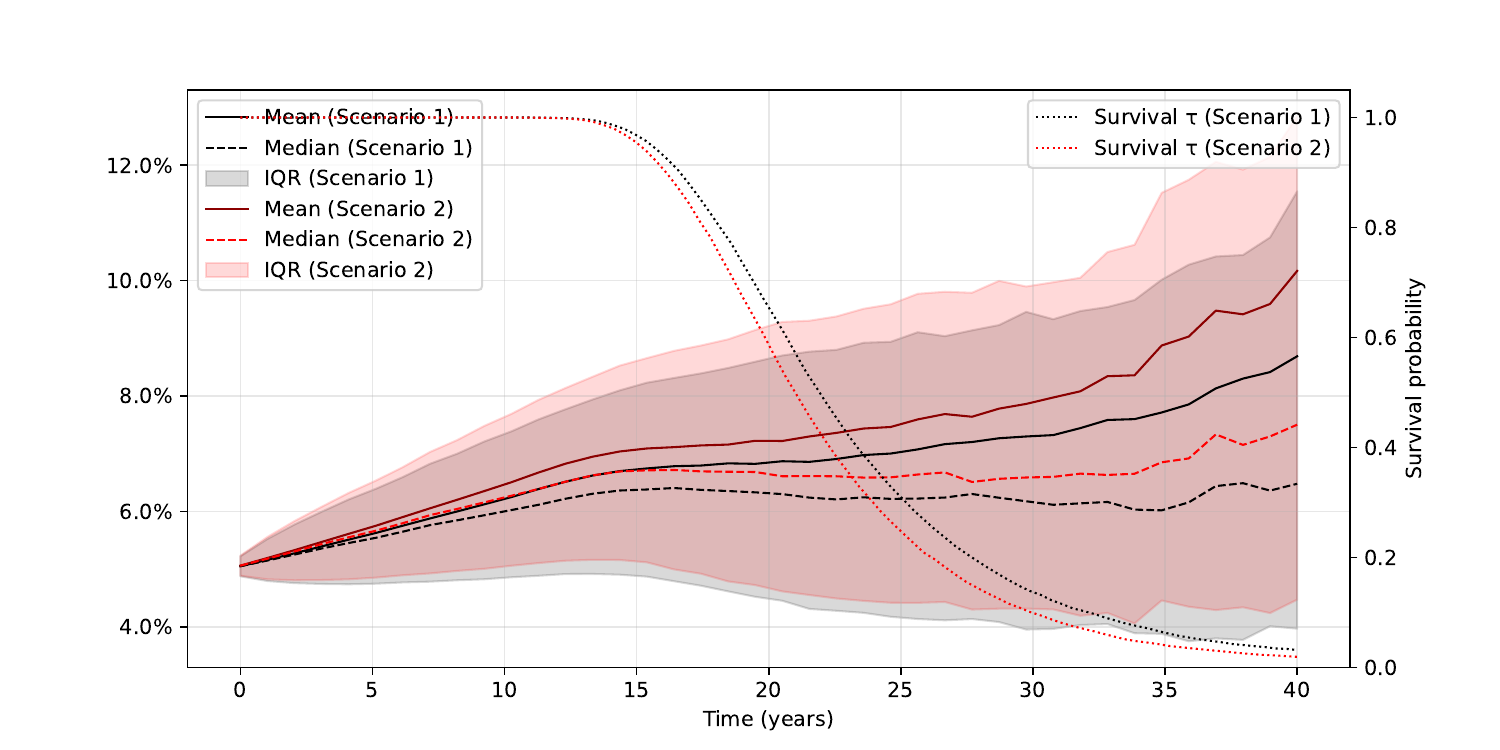}
\vspace{-0.2cm}\caption{Relative surplus $\rho_t=\frac{p^*_t-p^{\min}_t}{p^{\min}_t}$}\label{subfig:omega_1_vs_omega_DR_MC_rho}
\end{subfigure}
\begin{subfigure}[t]{0.49\textwidth}
\centering
\includegraphics[width=\textwidth,height=0.65\textheight,keepaspectratio]{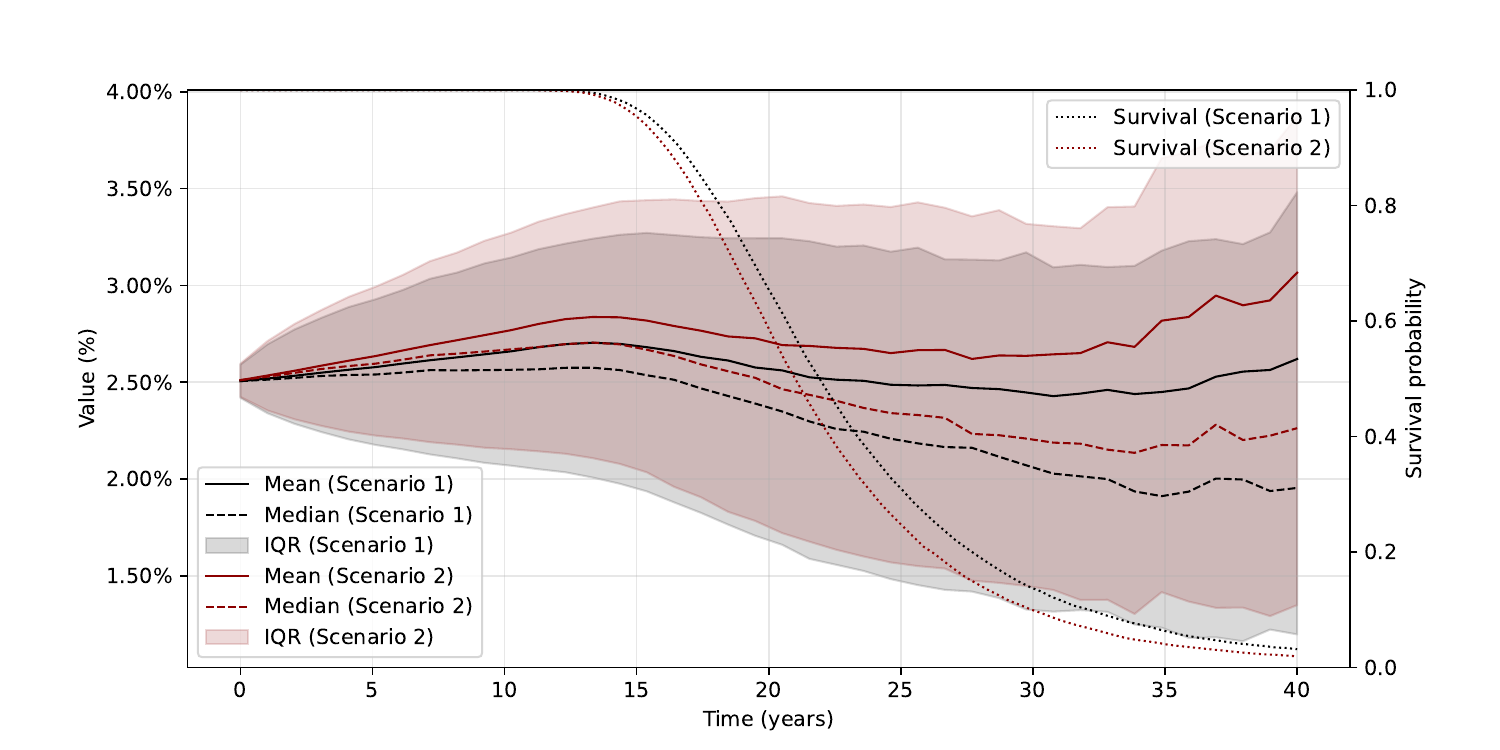}
\vspace{-0.2cm}\caption{Increase in Benefit ratio}\label{subfig:omega_1_vs_omega_DR_MC_benefit_ratio_increase}
\end{subfigure}
\captionsetup{font=footnotesize}
\caption{Comparison Between $\omega=1$ and $\omega = DR_s / DR_0$ ({\color{red} red}) – Monte Carlo}\label{fig:omega_1_vs_omega_DR_MC}
\end{figure}

\section{Conclusion}

This paper studies the design of an optimal public PAYG pension system complemented by a buffer fund in a context of persistent demographic ageing and economic uncertainty. Motivated by the projected rise in old-age dependency ratios documented by \cite{OECD25} and by the widespread adoption of Public Pension Reserve Funds, we develop a time-consistent framework in which a social planner simultaneously addresses sustainability and adequacy. Our approach relies on forward utilities, which allow preferences to adjust dynamically to evolving financial, economic, and demographic conditions, thereby avoiding the time inconsistency and horizon dependence inherent in backward formulations.

Within this framework, we characterize and study the explicit optimal pension and investment policy. The policy jointly determines the share of the buffer fund invested in risky assets and the surplus pension distributed on top of the pure PAYG benchmark. A comprehensive numerical analysis, calibrated to realistic financial dynamics and European ageing trends,  yields several
robust insights
 as well as an in-depth analysis of the impact of preference sensitivities on the pension scheme.

First, demographic pressure, captured by a baby-boom scenario, significantly accelerates buffer fund depletion relative to a steady-state population. Nevertheless, conditional on fund survival, the buffer generates a sizeable and persistent improvement in pensions. The relative surplus is particularly valuable in the baby-boom scenario and under low wage-growth conditions, where the PAYG baseline pension is reduced due to a higher dependency ratio. The buffer fund therefore plays a stabilizing role precisely when demographic and economic conditions are most adverse.

Overall, 
 the mixed PAYG system with an optimally managed buffer fund, as proposed in this paper, has the three main characteristics: (i) It can substantially enhance pension adequacy while preserving sustainability, provided that the initial calibration is carefully chosen. (ii) Demographic shocks primarily affect the distribution of the  fund depletion time rather than the structure of the optimal investment policy. (iii) Properly designed state-contingent surplus rules—implemented through forward utility sensitivities—facilitate effective risk sharing and contribute to macroeconomic stabilization.

From a policy perspective, the proposed framework provides a coherent tool for evaluating pension reforms in ageing societies. It clarifies how contribution rates, initial reserves, intergenerational weights, and risk sensitivities interact to shape long-run outcomes. More broadly, the analysis demonstrates that forward utility–based pension design offers a tractable and economically meaningful approach to reconciling sustainability with adequacy in an uncertain and ageing world.

\bibliographystyle{alpha}
\bibliography{references}

\clearpage 
\appendix

\setcounter{table}{0}
\setcounter{figure}{0}
\setcounter{equation}{0}
\renewcommand{\thetable}{\Alph{section}.\arabic{table}}
\renewcommand{\thefigure}{\Alph{section}.\arabic{figure}}

\section{Financial market scenarios}\label{app:finmarketscenarios}
\begin{figure}[H]
	\caption{Financial market, seed \# 3}\label{fig:fin_scen_seed_3}
	\setfolder{fig_24feb_delta_base_S1_SS_v_S2_BB_seed_3}
    \begin{subfigure}[b]{0.49\textwidth}
		\centering
		\includegraphics[scale=0.5]{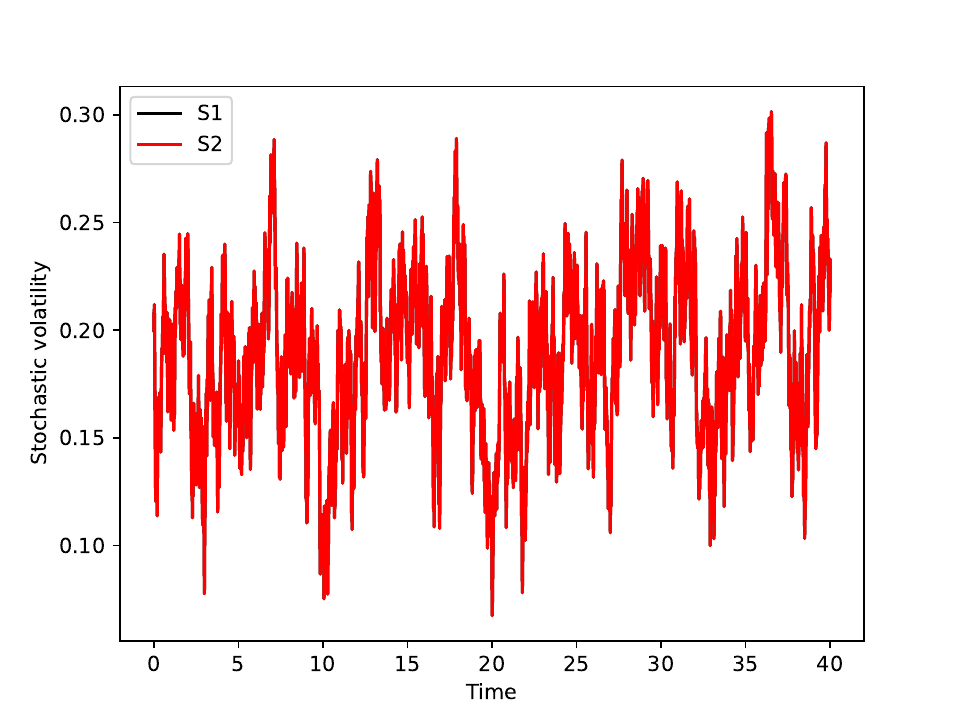}
		\caption{Stochastic volatility $\nu$} 
	\end{subfigure}
    \begin{subfigure}[b]{0.49\textwidth}
		\centering
		\includegraphics[scale=0.5]{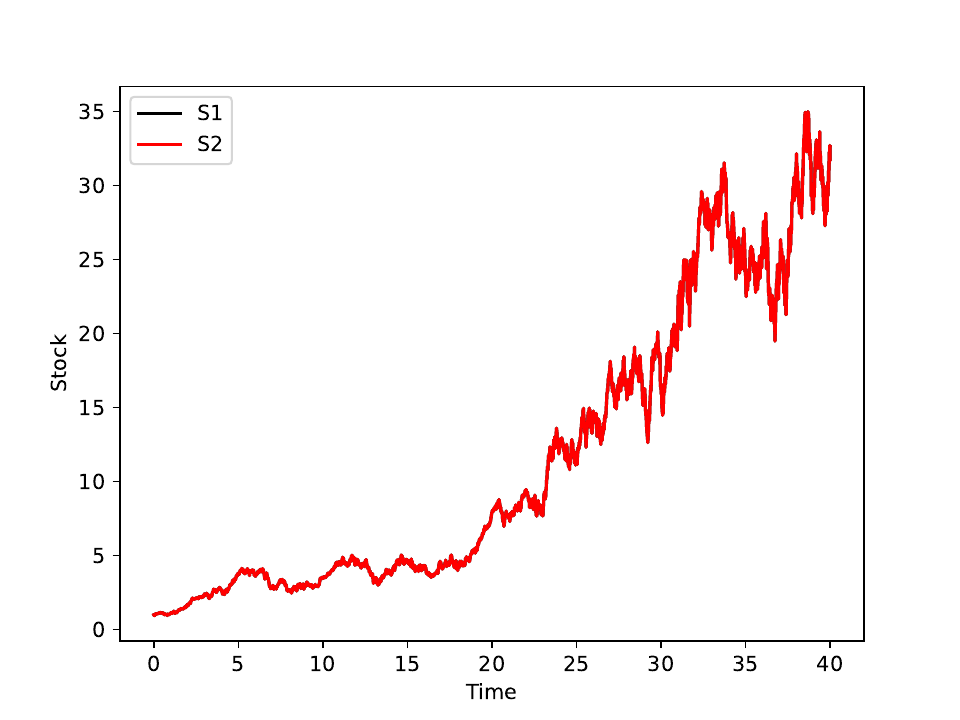}
		\caption{Stock $S$} 
	\end{subfigure}\\
	\begin{subfigure}[b]{0.49\textwidth}
		\centering
		\includegraphics[scale=0.5]{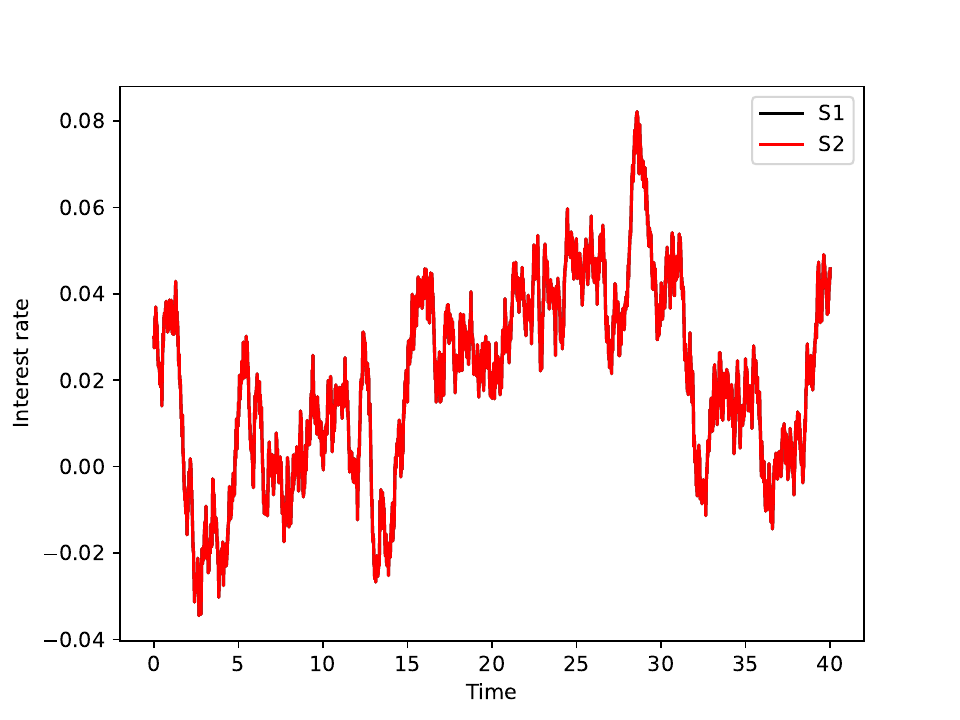}
		\caption{Interest rate $r$ } 
	\end{subfigure}
    \begin{subfigure}[b]{0.49\textwidth}
		\centering
		\includegraphics[scale=0.5]{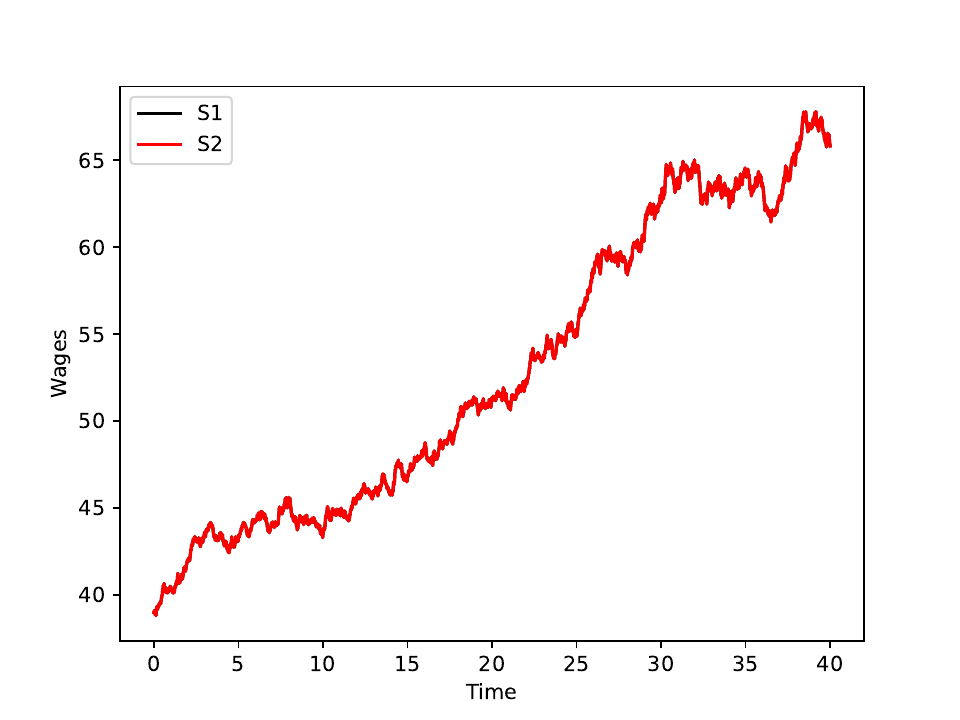}
		\caption{Wages $\mathfrak e$} 
	\end{subfigure}\\
    \begin{subfigure}[b]{0.49\textwidth}
		\centering
		\includegraphics[scale=0.5]{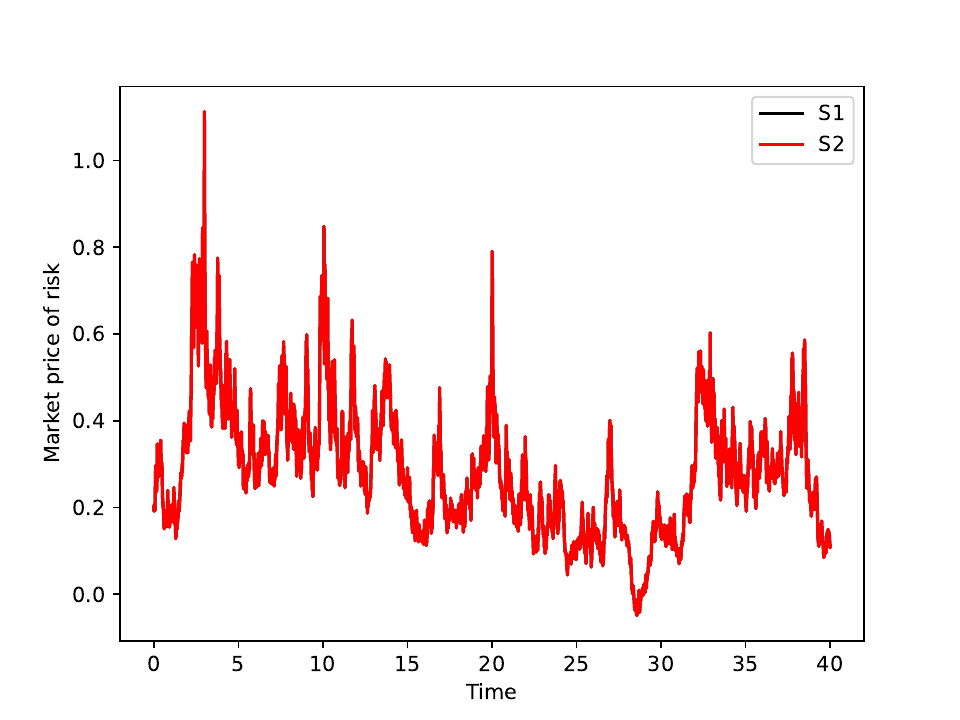}
		\caption{Market risk premium $\eta$} 
	\end{subfigure}\\
\end{figure}

\begin{figure}[H]
	\caption{Financial market, seed \# 5}\label{fin_scen_seed_5}
	\setfolder{fig_24feb_delta_base_S1_SS_v_S2_BB_seed_5}
    \begin{subfigure}[b]{0.49\textwidth}
		\centering
		\includegraphics[scale=0.5]{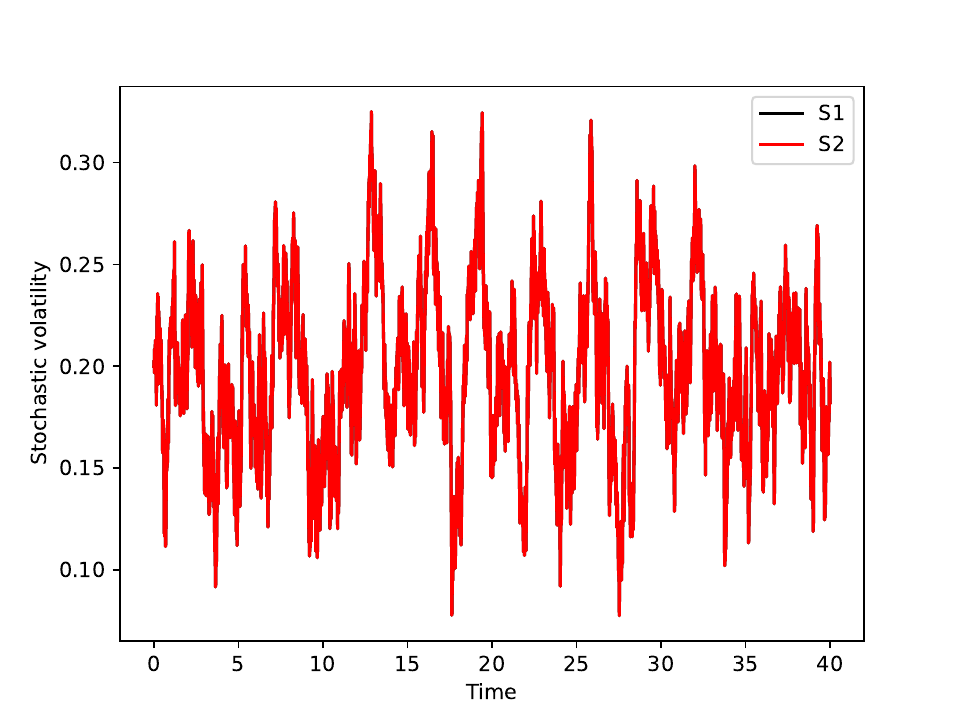}
		\caption{Stochastic volatility $\nu$} 
	\end{subfigure}
    \begin{subfigure}[b]{0.49\textwidth}
		\centering
		\includegraphics[scale=0.5]{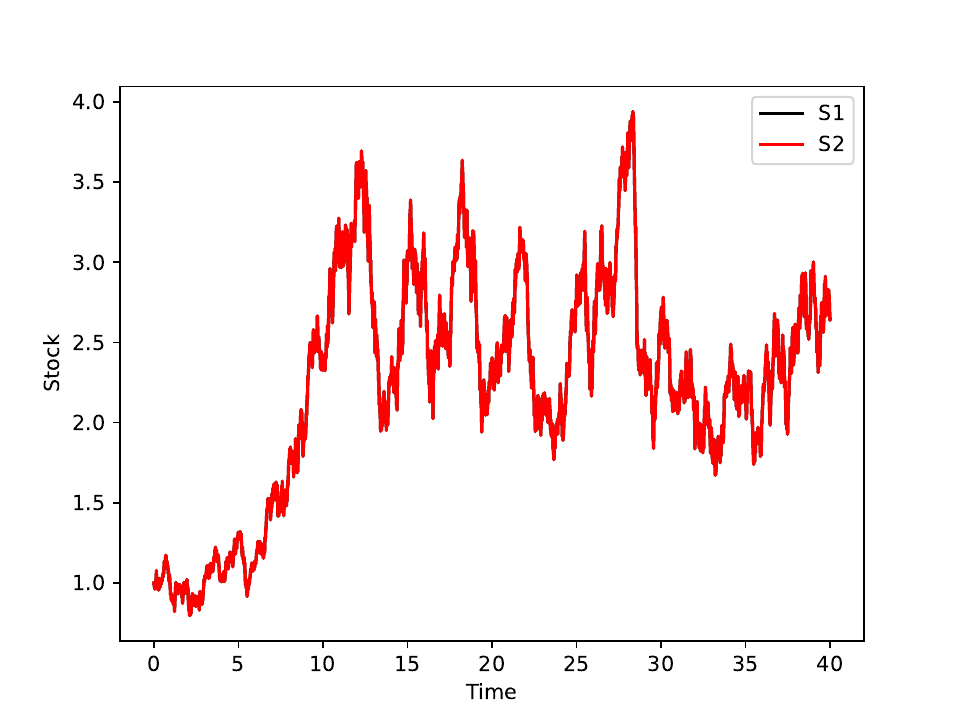}
		\caption{Stock $S$} 
	\end{subfigure}\\
	\begin{subfigure}[b]{0.49\textwidth}
		\centering
		\includegraphics[scale=0.5]{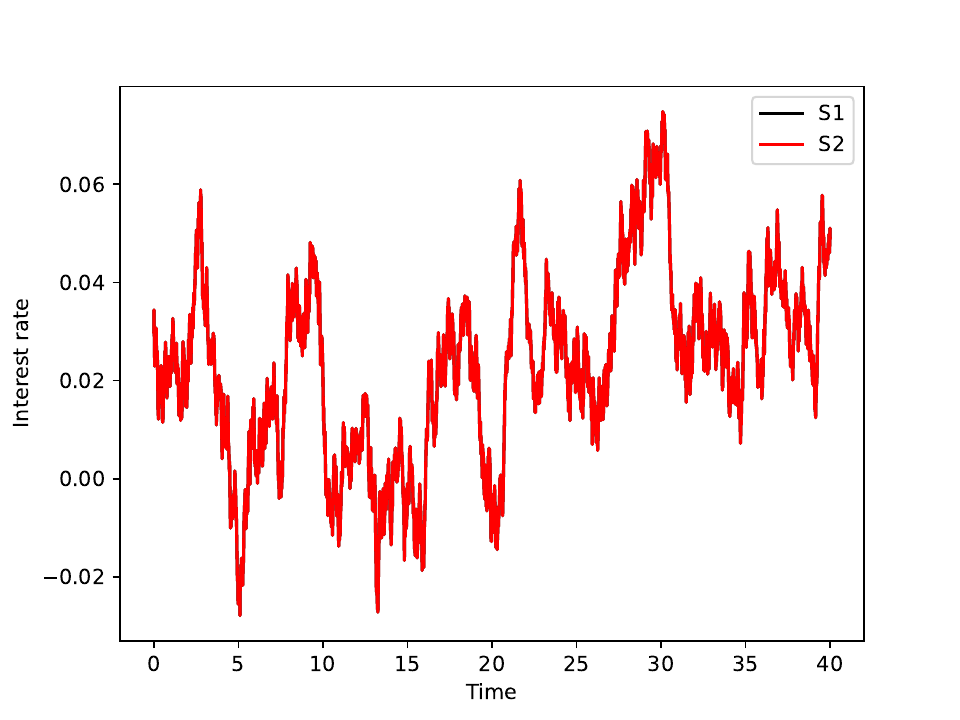}
		\caption{Interest rate $r$} 
	\end{subfigure}
    \begin{subfigure}[b]{0.49\textwidth}
		\centering
		\includegraphics[scale=0.5]{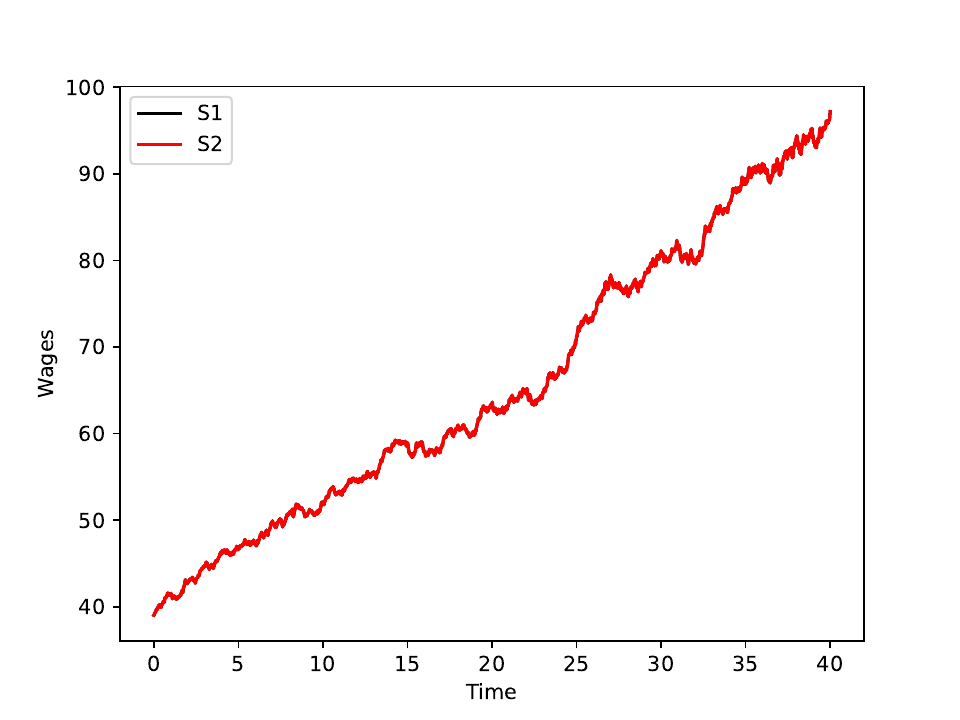}
		\caption{Wages $\mathfrak e $} 
	\end{subfigure}\\
    \begin{subfigure}[b]{0.49\textwidth}
		\centering
		\includegraphics[scale=0.5]{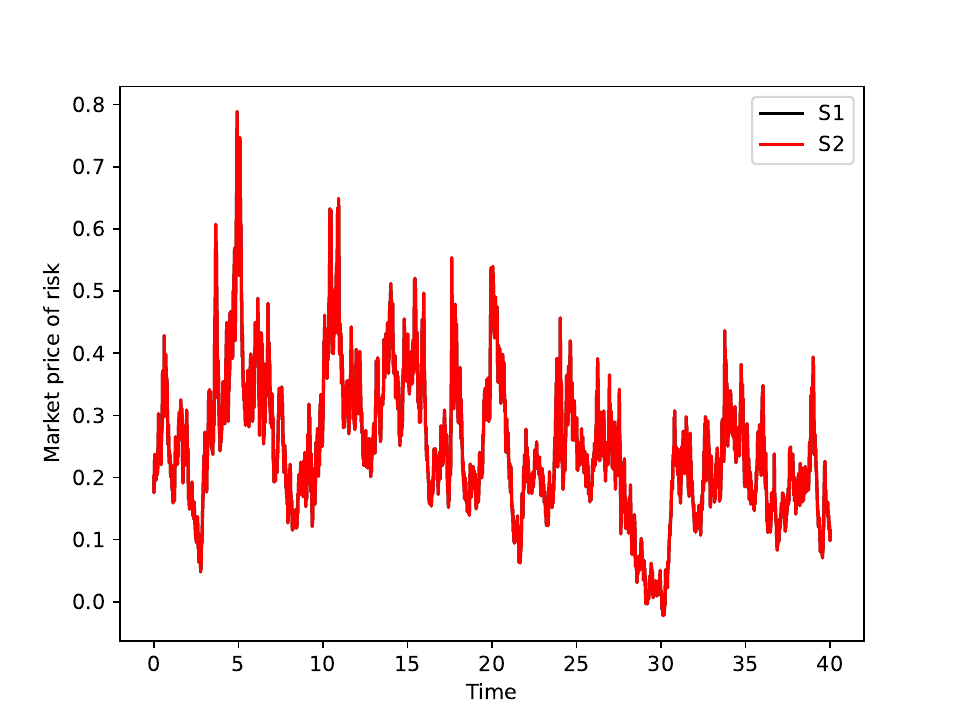}
		\caption{Market risk premium $\eta$} 
	\end{subfigure}\\
\end{figure}

\clearpage 
\newpage
\setcounter{table}{0}
\setcounter{figure}{0}
\setcounter{equation}{0}

\section{Pathwise Analysis: pessimistic scenario (seed 5)}\label{app:single_path_seed_5}
\vspace{-0.5cm}
\begin{figure}[h!]
\centering
	\setfolder{fig_24feb_delta_base_S1_SS_v_S2_BB_seed_5}
\begin{subfigure}[t]{0.49\textwidth}
\centering
\includegraphics[width=\textwidth,height=0.6\textheight,keepaspectratio]{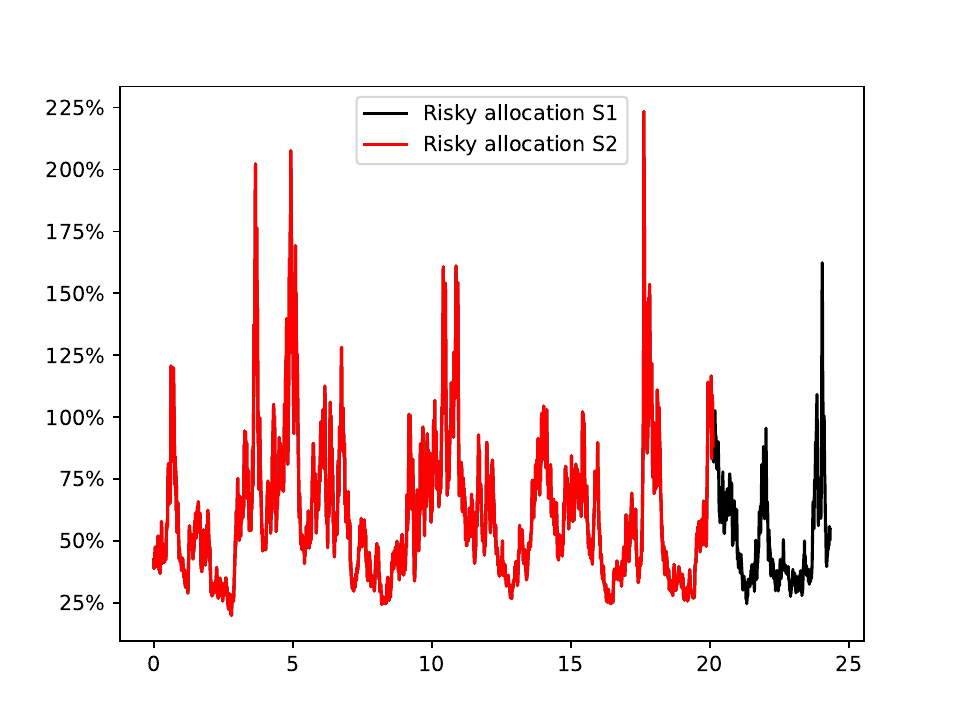}
\caption{Proportion risky investment $\left(\frac{\pi^*_t}{\sqrt{\nu_t} F^*_t}\right)$}\label{subfig:SS_v_BB_seed_5_proportion}
\end{subfigure}
\begin{subfigure}[t]{0.49\textwidth}
\centering
\includegraphics[width=\textwidth,height=0.6\textheight,keepaspectratio]{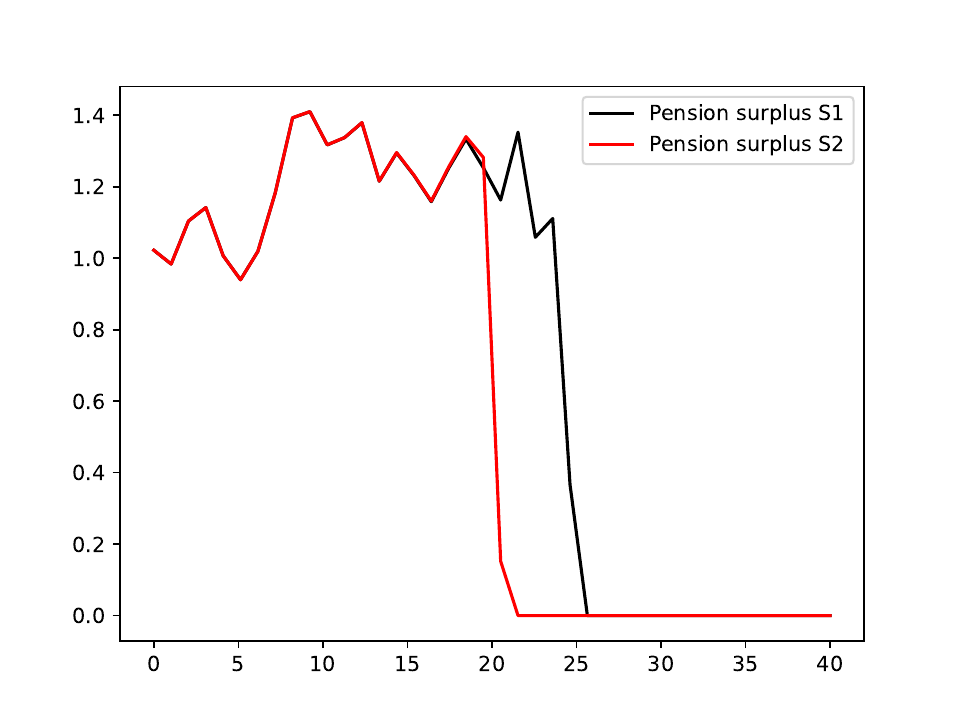}
\caption{Pension surplus $p^*_t-p^{\min}_t$}\label{subfig:SS_v_BB_seed_5_surplus}
\end{subfigure}
\begin{subfigure}[t]{0.49\textwidth}
\centering
\includegraphics[width=\textwidth,height=0.6\textheight,keepaspectratio]{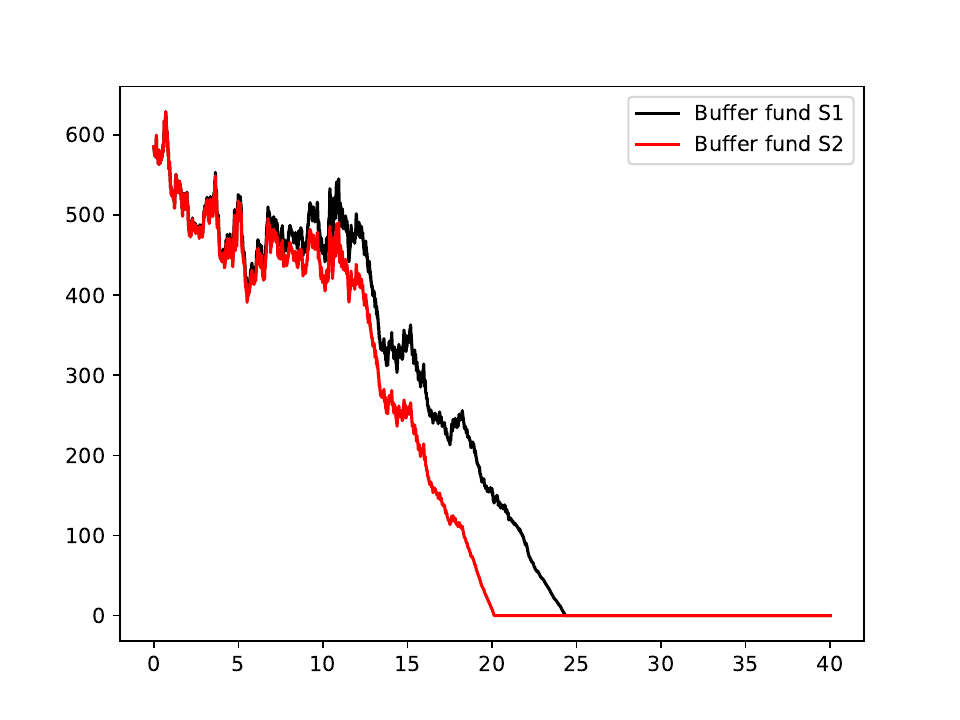}
\caption{Buffer fund $F^*_t$}\label{subfig:SS_v_BB_seed_5_buffer}
\end{subfigure}\hfill
\begin{subfigure}[t]{0.49\textwidth}
\centering
\includegraphics[width=\textwidth,height=0.6\textheight,keepaspectratio]{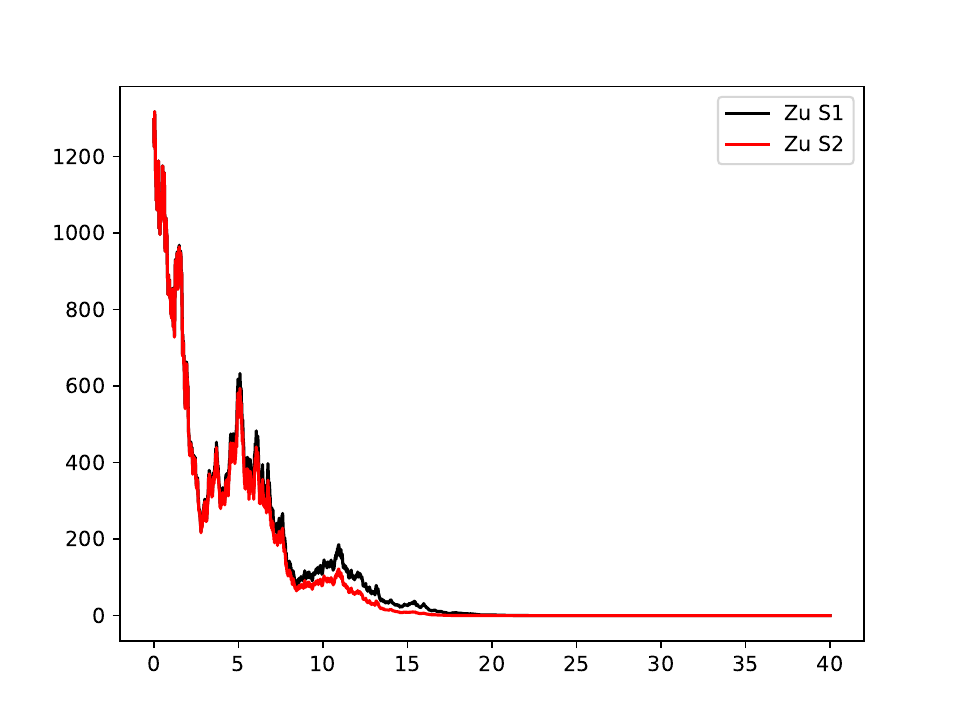}
\caption{Buffer fund utility weight $Z^u_t$}\label{subfig:SS_v_BB_seed_5_Zu}
\end{subfigure}\hfill
\begin{subfigure}[t]{0.49\textwidth}
\centering
\includegraphics[width=\textwidth,height=0.6\textheight,keepaspectratio]{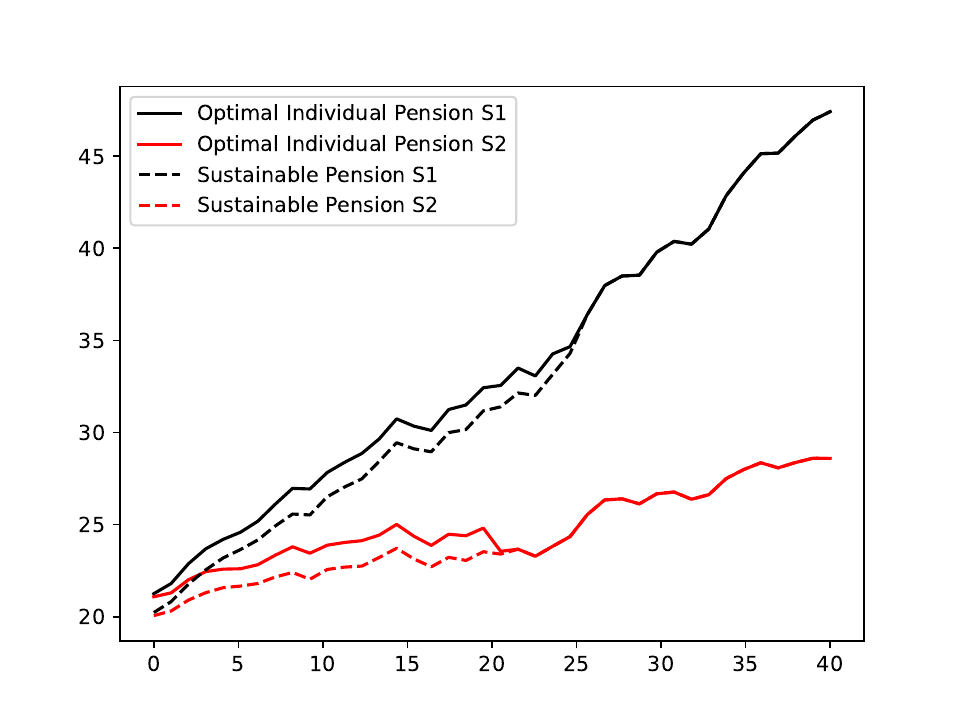}
\caption{Individual pensions $p^*_t$}\label{subfig:SS_v_BB_seed_5_pstar}
\end{subfigure}\hfill
\begin{subfigure}[t]{0.49\textwidth}
\centering
\includegraphics[width=\textwidth,height=0.6\textheight,keepaspectratio]{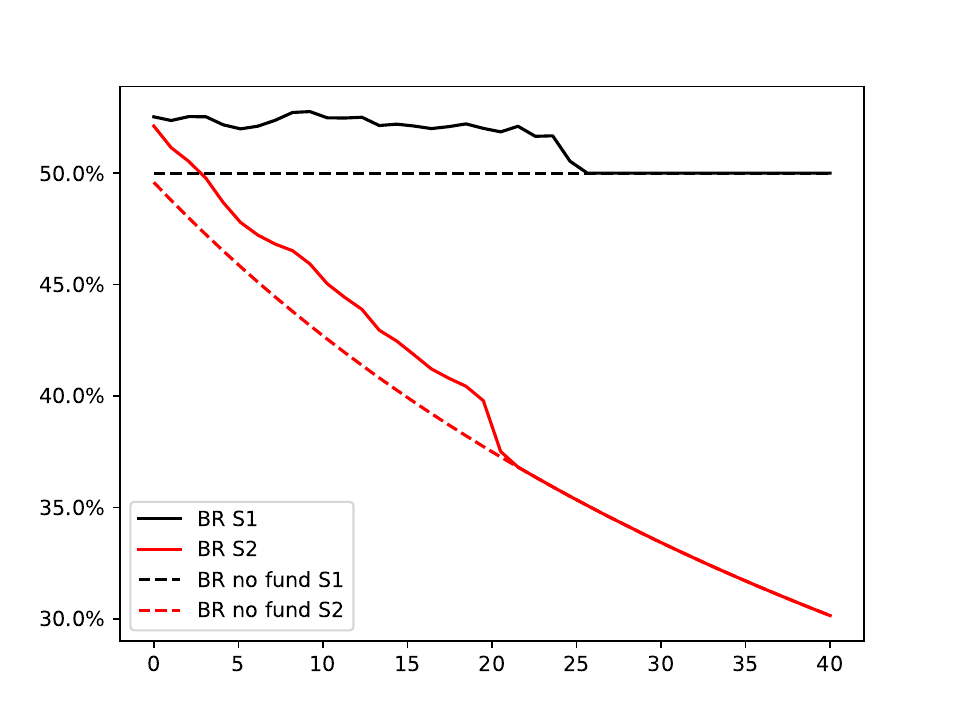}
\caption{Benefit ratio $\text{BR}_t=\frac{p^*_t}{\mathfrak{e}_t}$}\label{subfig:SS_v_BB_seed_5_benefit_ratio}
\end{subfigure}
\captionsetup{font=footnotesize}
\caption{Comparison Between Steady State (black) and Baby Boom ({\color{red} red}) – Seed 5}\label{fig:SS_v_BB_seed_5}
\end{figure}

\clearpage 
\newpage

\setcounter{table}{0}
\setcounter{figure}{0}
\setcounter{equation}{0}

\section{Distributional Analysis: Additional results}\label{app:MC_extra}

\begin{figure}[h!]
\centering
    \setfolder{fig_23feb_delta_base_S1_SS_v_S2_BB_MC}
\begin{subfigure}[t]{0.49\textwidth}
\centering
\includegraphics[width=\textwidth,keepaspectratio]{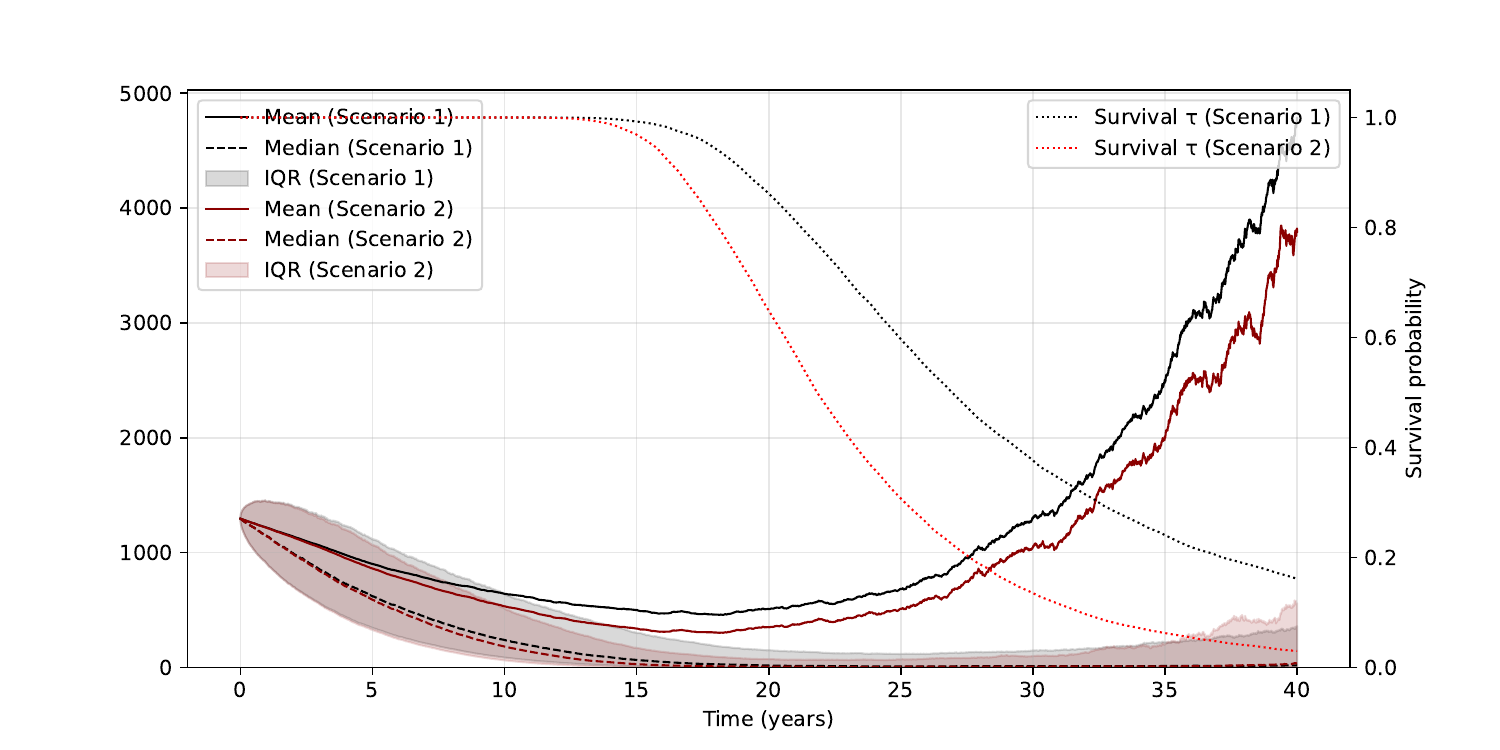}
\caption{Buffer fund utility weight $Z^u_t$}\label{subfig:SS_v_BB_MC_Zu}
\end{subfigure}\hfill
\begin{subfigure}[t]{0.49\textwidth}
\centering
\includegraphics[width=\textwidth,keepaspectratio]{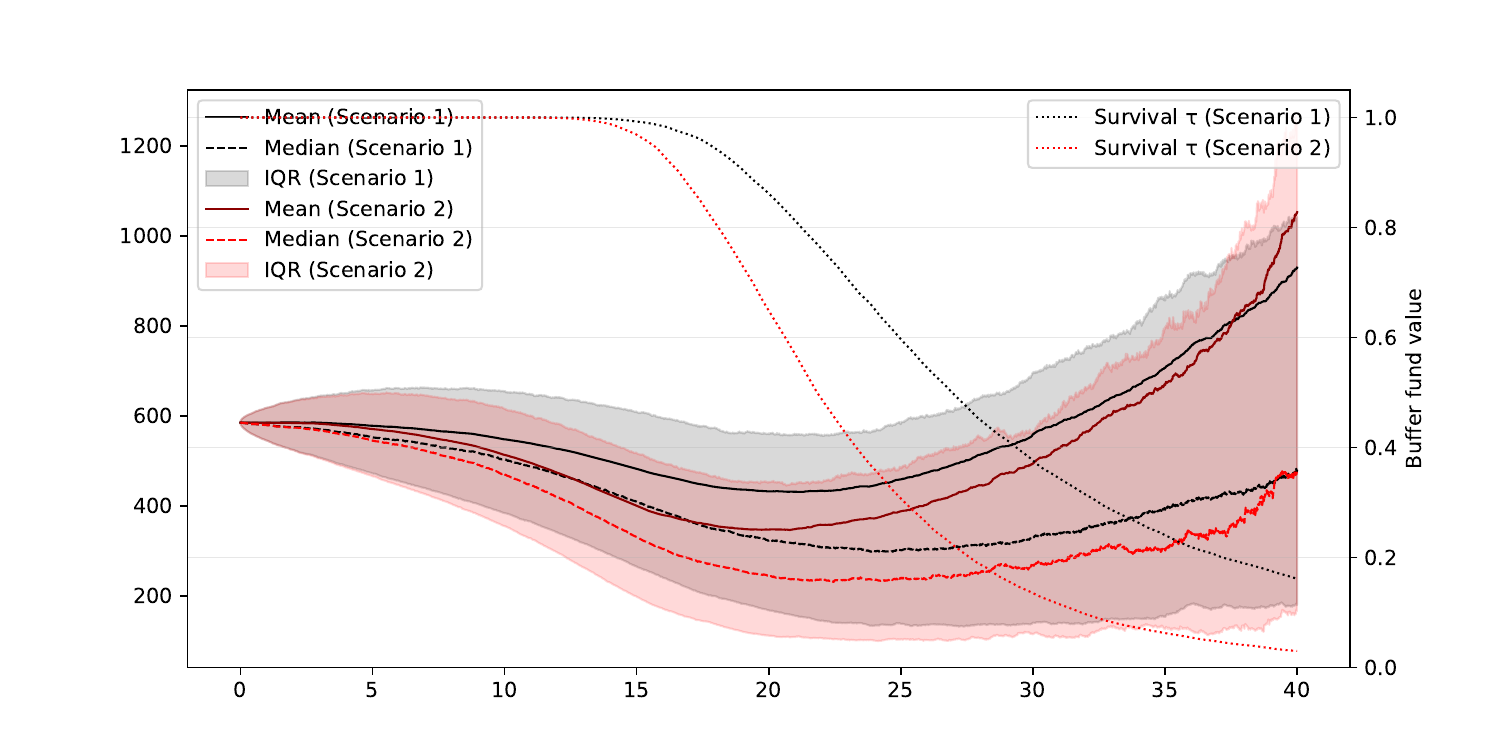}
\vspace{-0.2cm}\caption{Buffer fund $F^*_t$}\label{subfig:SS_v_BB_MC_buffer}
\end{subfigure}\hfill \\
\begin{subfigure}[t]{0.49\textwidth}
\centering
\includegraphics[width=\textwidth,keepaspectratio]{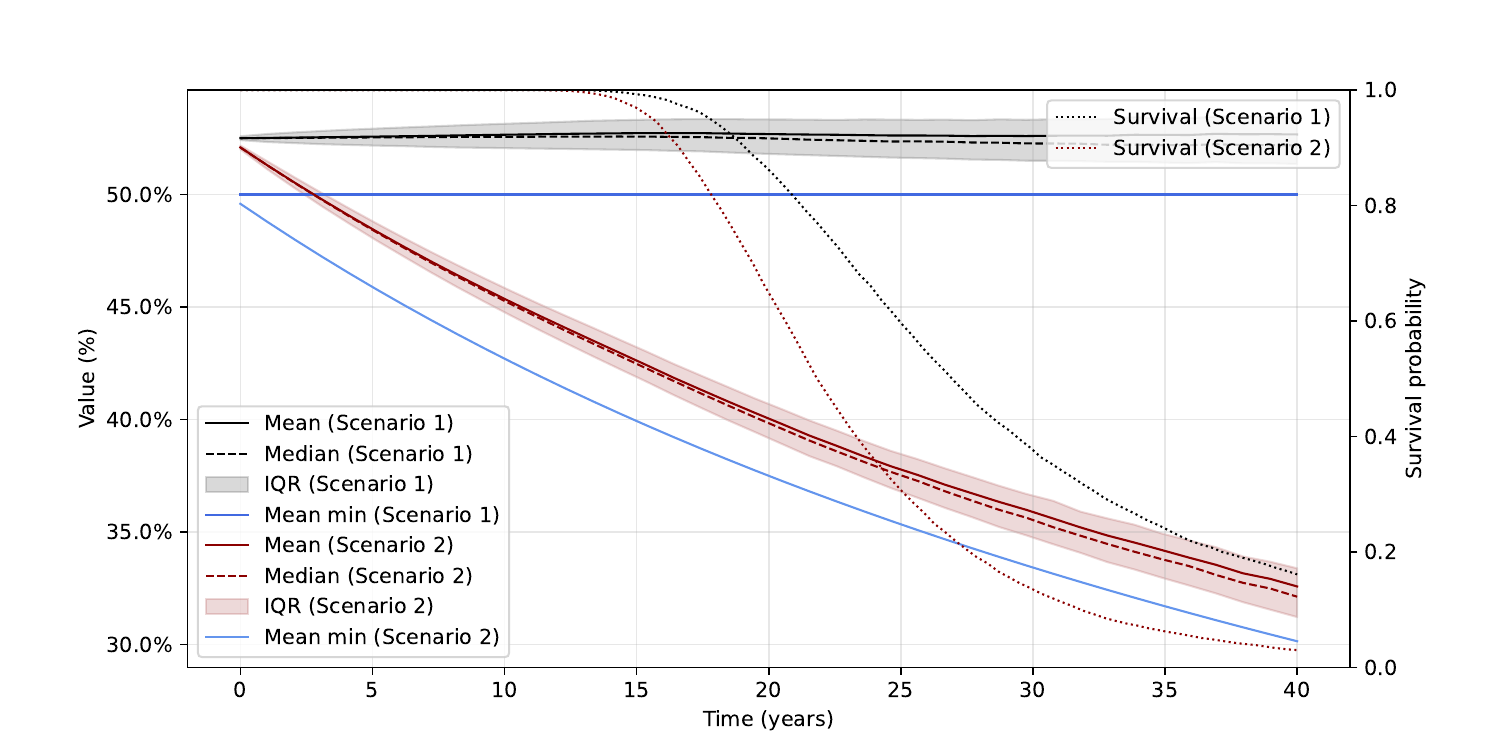}
\vspace{-0.2cm}\caption{Benefit ratio $\text{BR}_t=\frac{p^*_t}{\mathfrak{e}_t}$}\label{subfig:SS_v_BB_MC_benefit_ratio}
\end{subfigure}
\captionsetup{font=footnotesize}
\caption{Comparison between Steady State (black) and Baby Boom ({\color{red} red}) – Monte Carlo - Additional results}\label{fig:SS_v_BB_MC_appendix}
\end{figure}

\begin{figure}[h!]
\centering
    \setfolder{fig_24feb_bb_S1_delta_0_v_S2_delta_base_MC}
\begin{subfigure}[t]{0.49\textwidth}
\centering
\includegraphics[width=\textwidth,keepaspectratio]{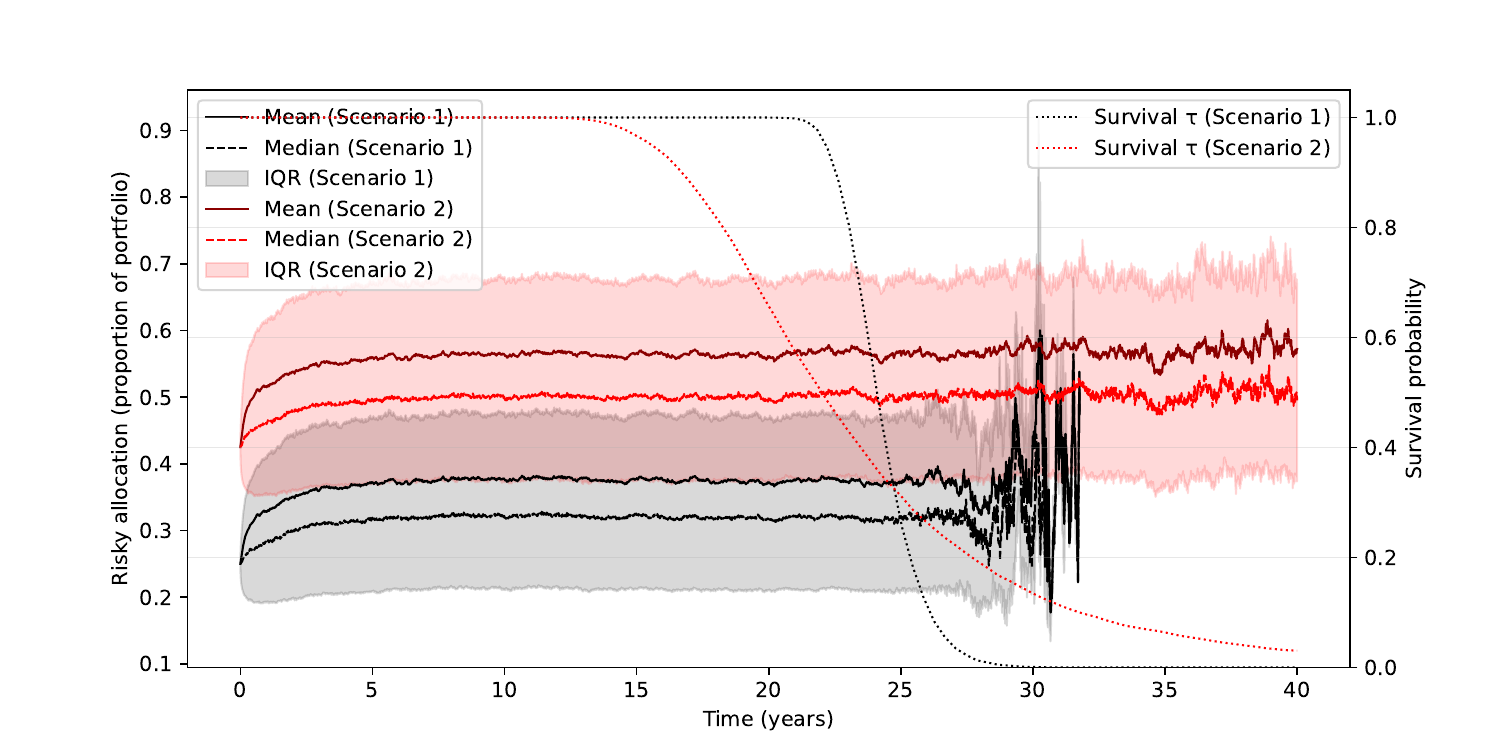}
\vspace{-0.2cm}\caption{Proportion risky investment $\frac{\left(\frac{\pi^*_t}{\sqrt{\nu_t}}\right)}{F^*_t}$}\label{subfig:bb_delta_0_v_delta_base_appendix_proportion}
\end{subfigure}
\begin{subfigure}[t]{0.49\textwidth}
\centering
\includegraphics[width=\textwidth,keepaspectratio]{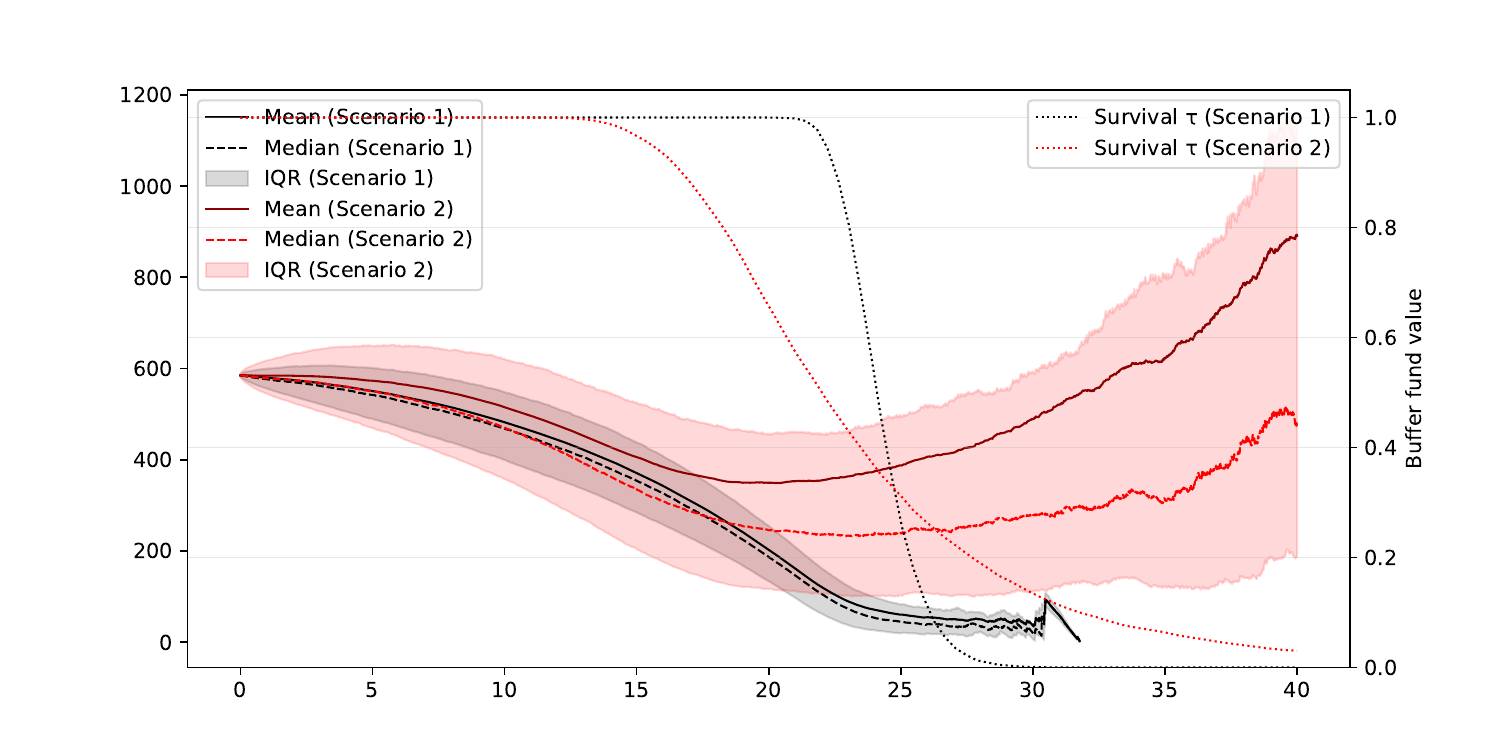}
\vspace{-0.2cm}\caption{Buffer fund $F^*_t$}\label{subfig:bb_delta_0_v_delta_base_appendix_relative_surplus}
\end{subfigure}\hfill 
\vspace{0.3cm}\\
\begin{subfigure}[t]{0.49\textwidth}
\centering
\includegraphics[width=\textwidth,keepaspectratio]{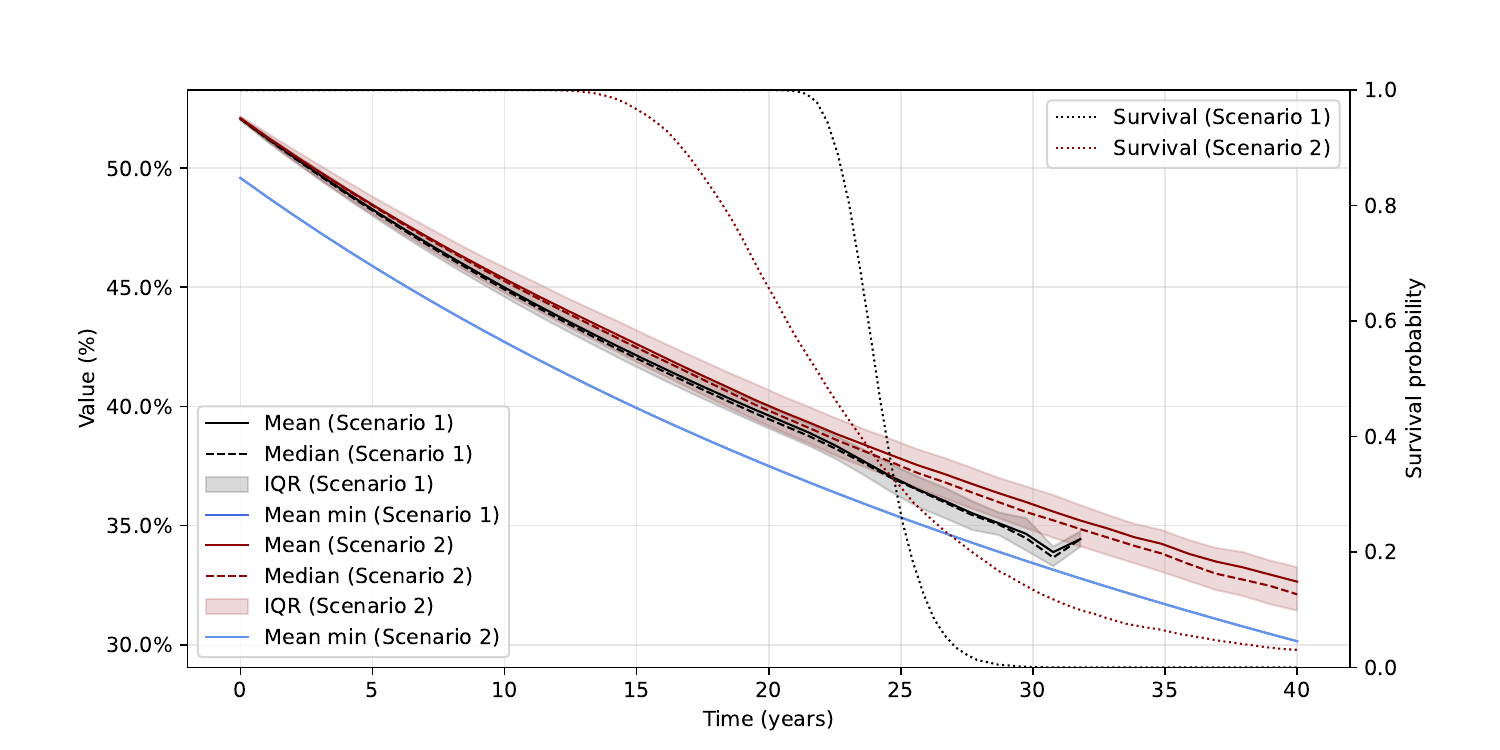}
\vspace{-0.2cm}\caption{Benefit ratio $\text{BR}_t=\frac{p^*_t}{\mathfrak{e}_t}$}\label{subfig:bb_delta_0_v_delta_base_appendix_benefit_ratio}
\end{subfigure}
\captionsetup{font=footnotesize}
\caption{Comparison between $\delta= {}^{t} \! (0,0,0,0) $ (black) and $\delta={}^{t} \!  (0, -0.2, -0.2,-0.2) $ ({\color{red} red}) – Monte Carlo}\label{fig:bb_delta_0_v_delta_base_appendix}
\end{figure}

\newpage
\mbox{   }
\newpage

\section{Sustainability}\label{app:sustainability}
\begin{figure}[ht!]
\centering
\setfolder{fig_26feb_bb_delta_base_S1_v_S2_v_S3_sens_Z0u}
\begin{subfigure}[t]{0.46\textwidth}
\centering
\includegraphics[width=\textwidth,height=0.34\textheight,keepaspectratio]{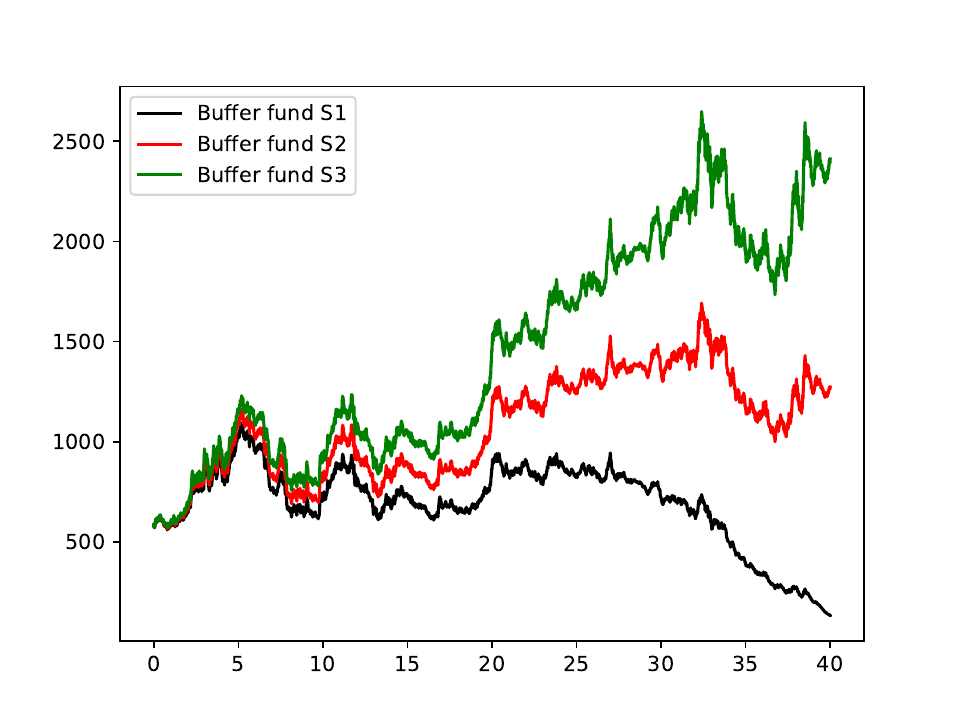}
\vspace{-0.2cm} \caption{Buffer fund $F^*$}\label{subfig:BB_sens_Z0u_Ft}
\end{subfigure}\hfill
\begin{subfigure}[t]{0.46\textwidth}
\centering
\includegraphics[width=\textwidth,height=0.34\textheight,keepaspectratio]{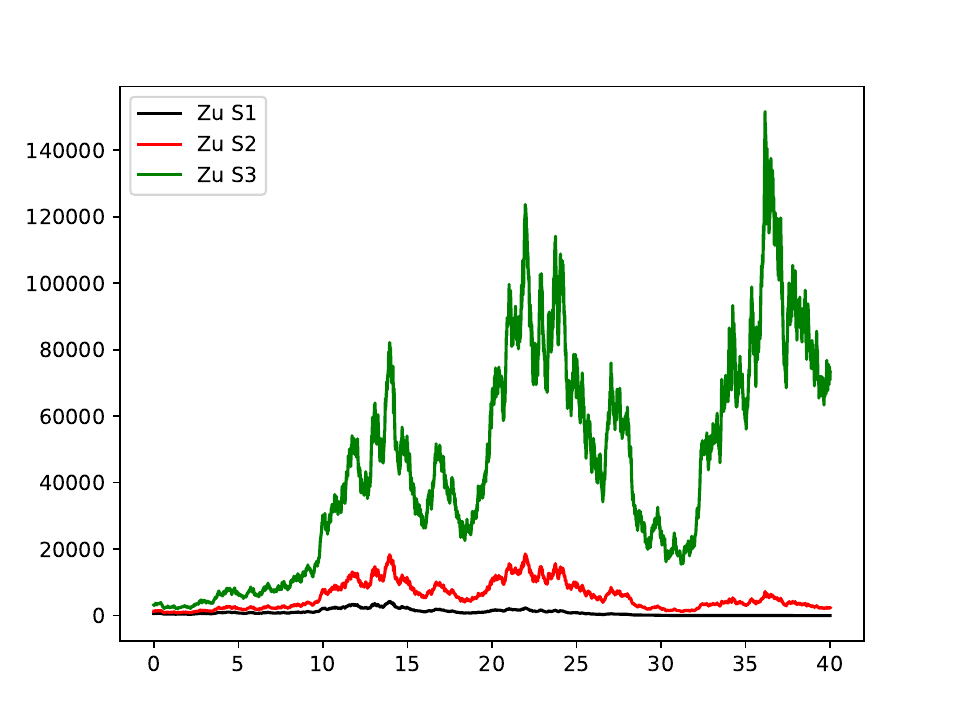}
\vspace{-0.2cm} \caption{Buffer fund weight $Z^u_t$}\label{subfig:BB_sens_Z0u_Zut}
\end{subfigure}
\begin{subfigure}[t]{0.46\textwidth}
\centering
\includegraphics[width=\textwidth,height=0.34\textheight,keepaspectratio]{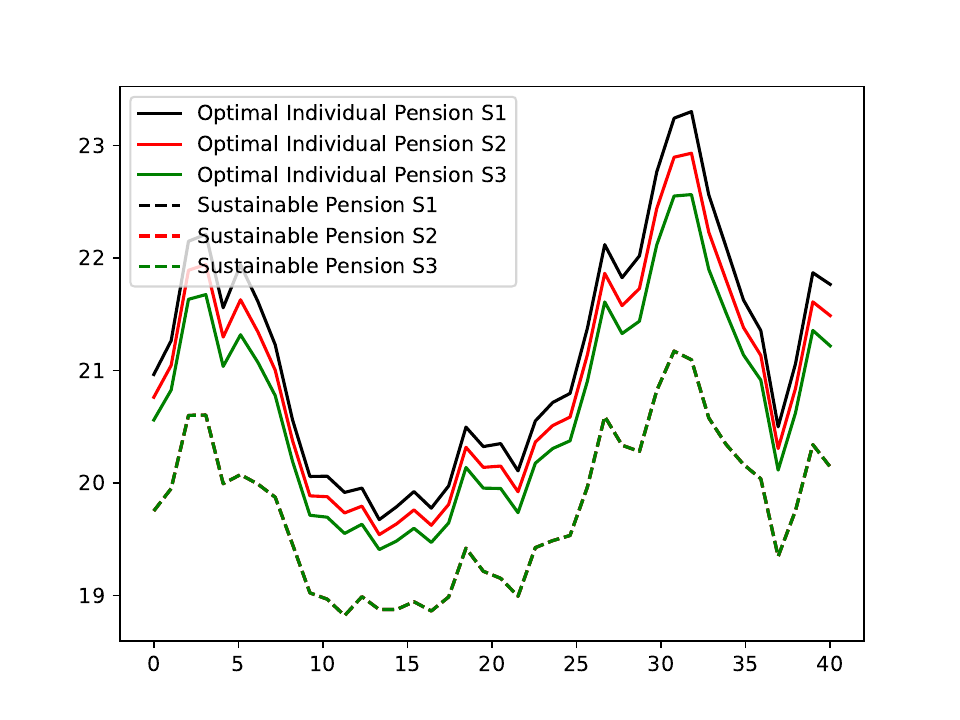}
\vspace{-0.2cm}\caption{Individual pensions $p^*_t$}\label{subfig:BB_sens_Z0u_pstar}
\end{subfigure}\hfill
\begin{subfigure}[t]{0.46\textwidth}
\centering
\includegraphics[width=\textwidth,height=0.34\textheight,keepaspectratio]{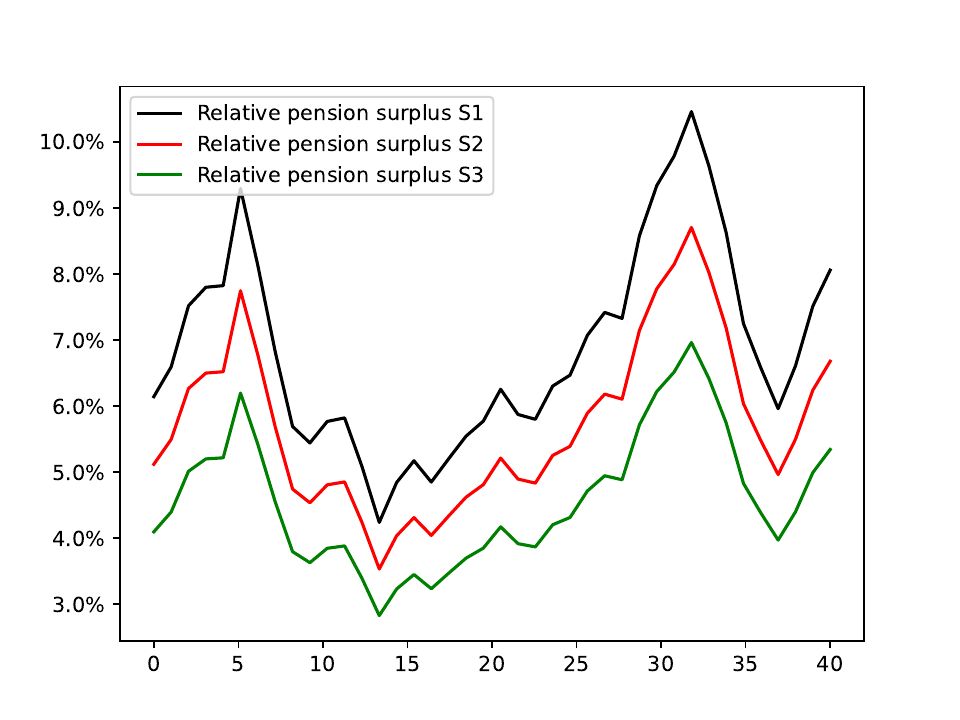}
\vspace{-0.2cm}\caption{Relative surplus $\rho_t=\frac{p^*_t-p^{\min}_t}{p^{\min}_t}$}\label{subfig:BB_sens_Z0u_rhostar}
\end{subfigure}
\captionsetup{font=footnotesize}
\caption{S1 (black): $Z^u_0 =Z_0 \left(\frac{N^r_0}{0.06}\right)^{\theta}$, S2 (red):  $Z^u_0=Z_0\left(\frac{N^r_0}{0.05}\right)^{\theta}$, S3 (green): $Z^u_0 =Z_0  \left(\frac{N^r_0}{0.04}\right)^{\theta}$ - Seed 3 (optimist)}\label{fig:BB_sens_Z0u}
\end{figure}

\clearpage 
\mbox{ }
\section{Adequacy}\label{app:adequacy}
\begin{figure}[h!]
\centering
\setfolder{fig_24feb_delta_base_S1_omega_1_v_S2_omega_DR_MC}
\begin{subfigure}[t]{0.49\textwidth}
\centering
\includegraphics[width=\textwidth,height=0.65\textheight,keepaspectratio]{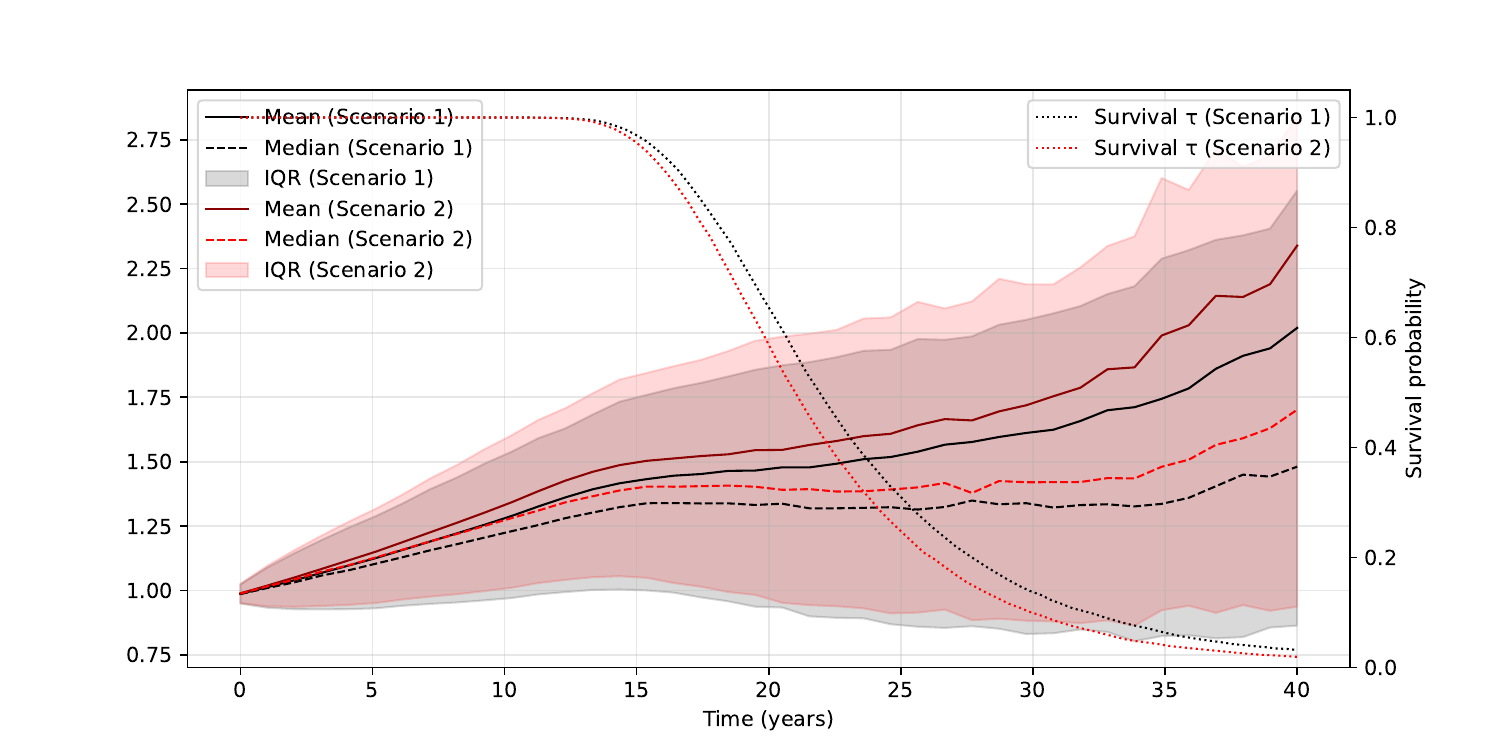}
\vspace{-0.2cm}\caption{Pension surplus $p^*_t-p^{\min}_t$}
\end{subfigure}
\begin{subfigure}[t]{0.49\textwidth}
\centering
\includegraphics[width=\textwidth,height=0.65\textheight,keepaspectratio]{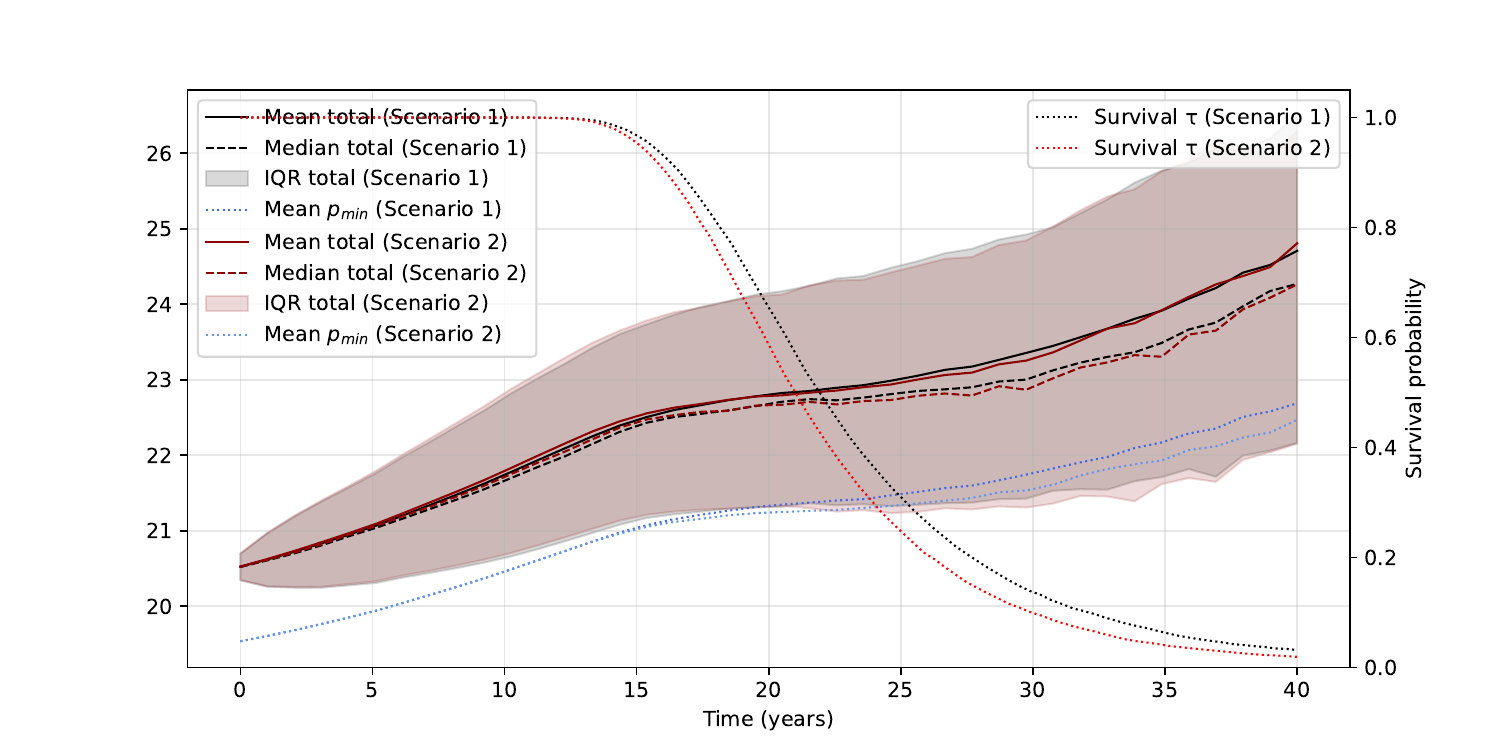}
\vspace{-0.2cm}\caption{Individual pensions $p^*_t$}
\end{subfigure}\hfill
\vspace{0.3cm}

\begin{subfigure}[t]{0.49\textwidth}
\centering
\includegraphics[width=\textwidth,height=0.65\textheight,keepaspectratio]{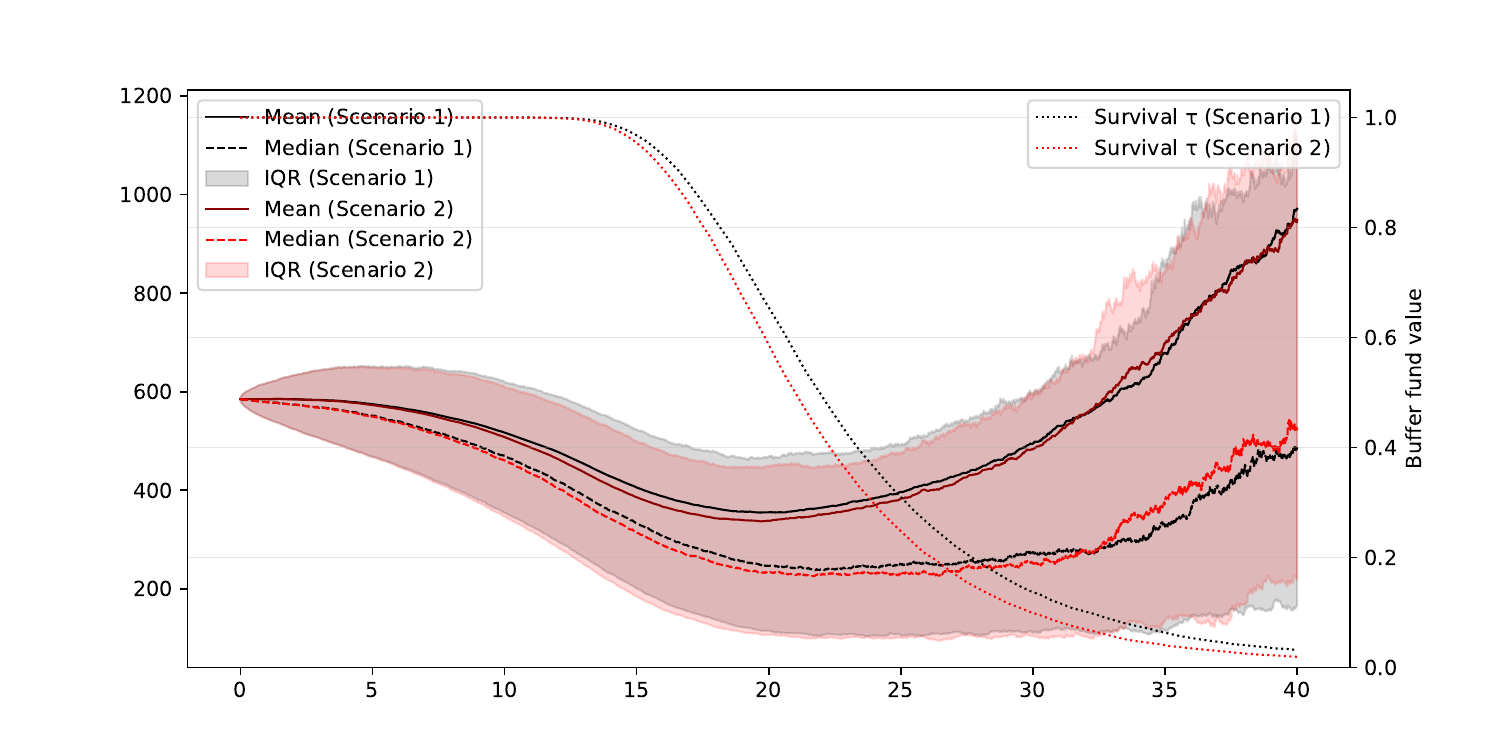}
\vspace{-0.2cm}\caption{Buffer fund $F^*_t$}
\end{subfigure}\hfill
\begin{subfigure}[t]{0.49\textwidth}
\centering
\includegraphics[width=\textwidth,height=0.65\textheight,keepaspectratio]{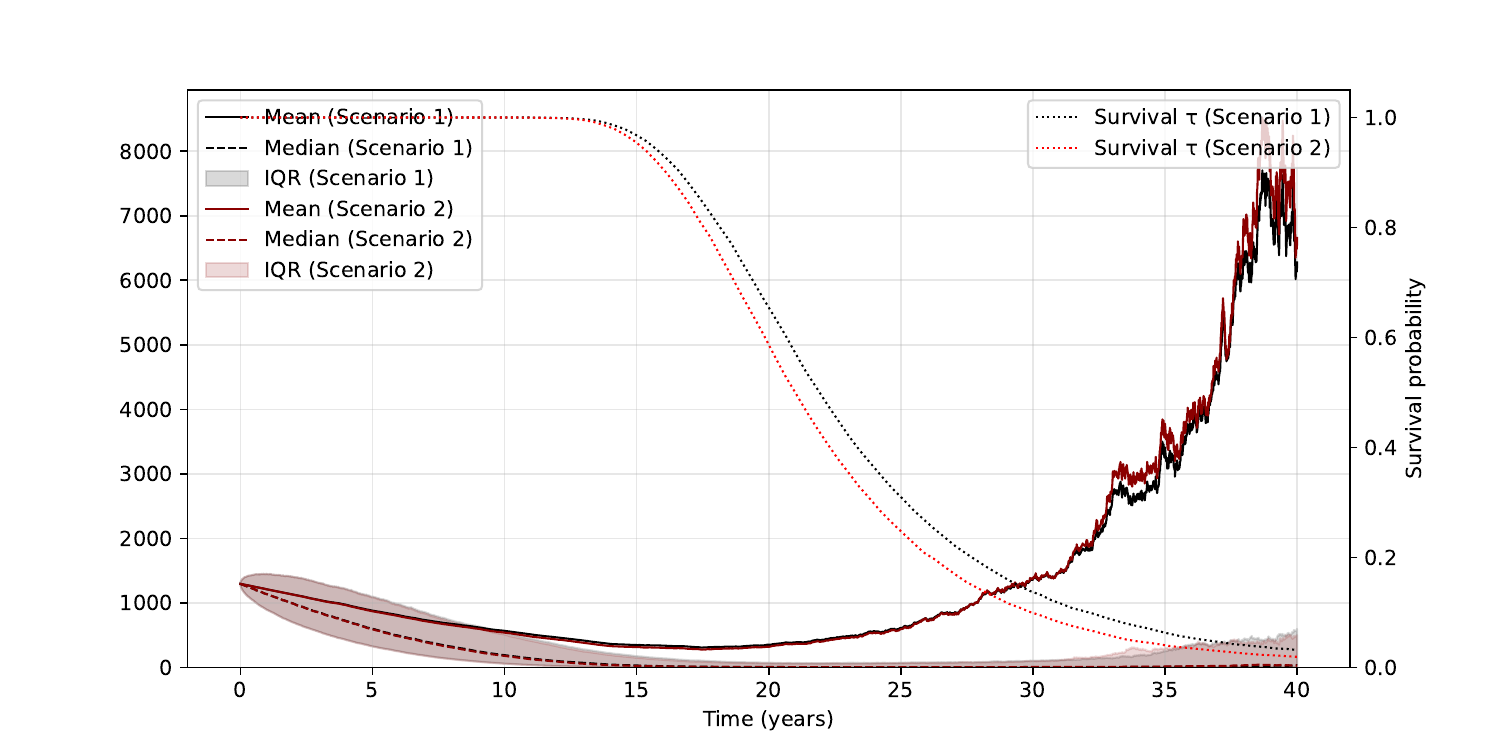}
\caption{Buffer fund utility weight $Z^u_t$}
\end{subfigure}\hfill
\vspace{0.3cm}
\begin{subfigure}[t]{0.96\textwidth}
\centering
\includegraphics[width=\textwidth,height=0.65\textheight,keepaspectratio]{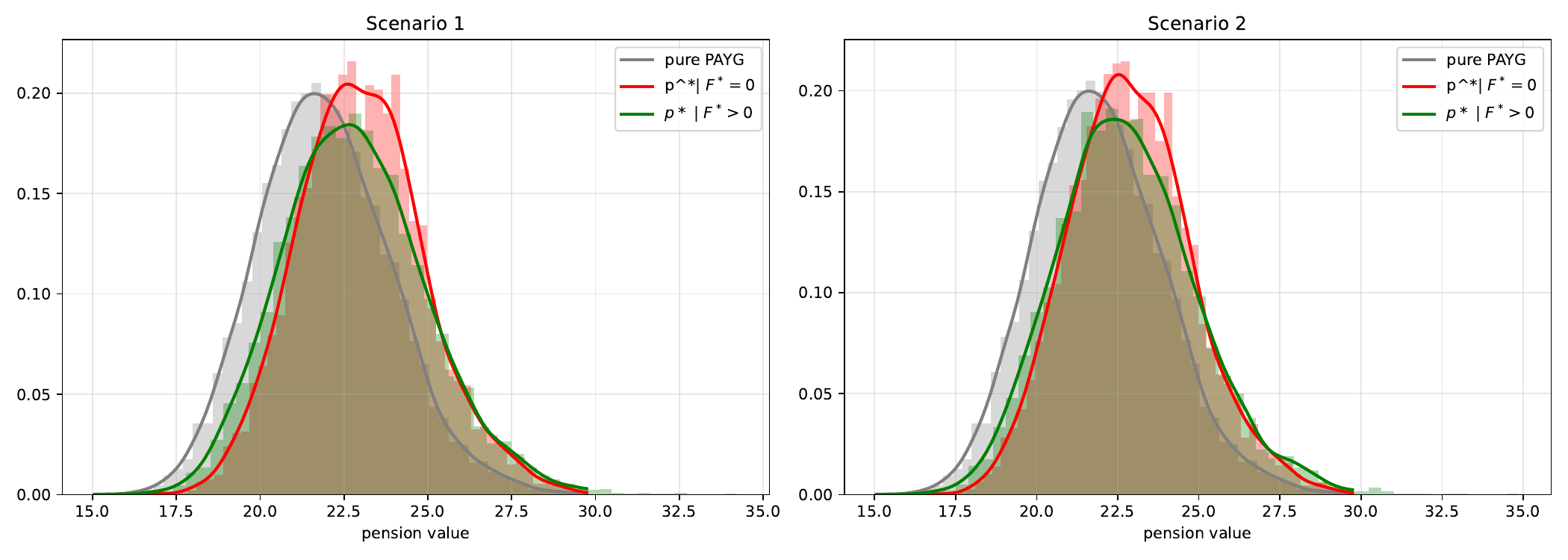}
\caption{Distribution of $p^{\min}$ (gray), $p^*|F^*>0$ (green) and $p^*|F^*=0$ ({\color{red} red}) at $t=20$}
\end{subfigure}
\captionsetup{font=footnotesize}
\caption{Comparison between $\omega=1$ and $\omega = DR_s / DR_0$ ({\color{red} red}) – Monte Carlo}\label{fig:omega_1_vs_omega_DR_MC_appendix}
\end{figure}

\begin{figure}[h!]
\centering
\begin{subfigure}[t]{0.49\textwidth}
\centering
\includegraphics[width=\textwidth,height=0.7\textheight,keepaspectratio]{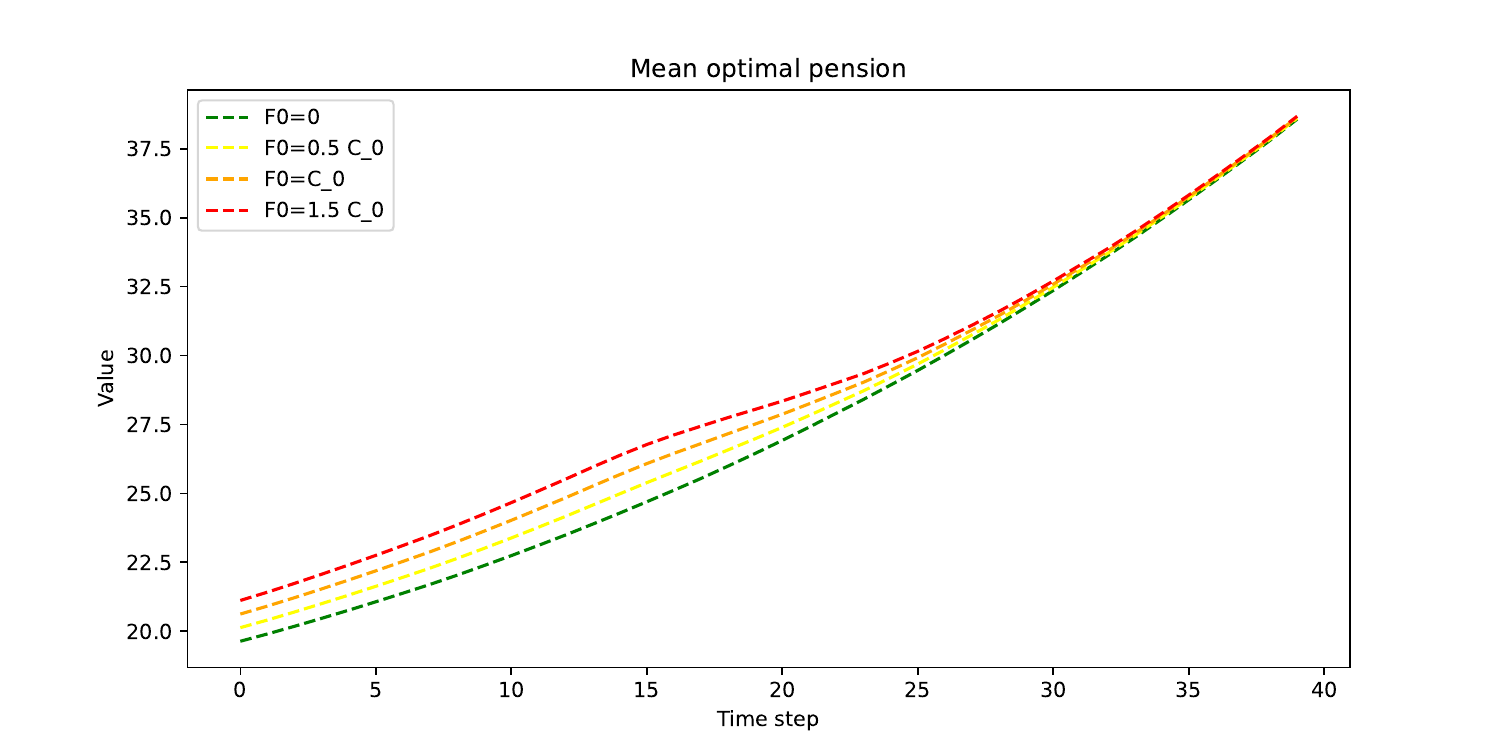}
\vspace{-0.2cm}\caption{Mean total pension $p^*$ versus $F_0$}\label{subfig:BB_sens_F0_three_wages_p*_0.03}
\end{subfigure}\hfill
\begin{subfigure}[t]{0.49\textwidth}
\centering
\includegraphics[width=\textwidth,height=0.7\textheight,keepaspectratio]{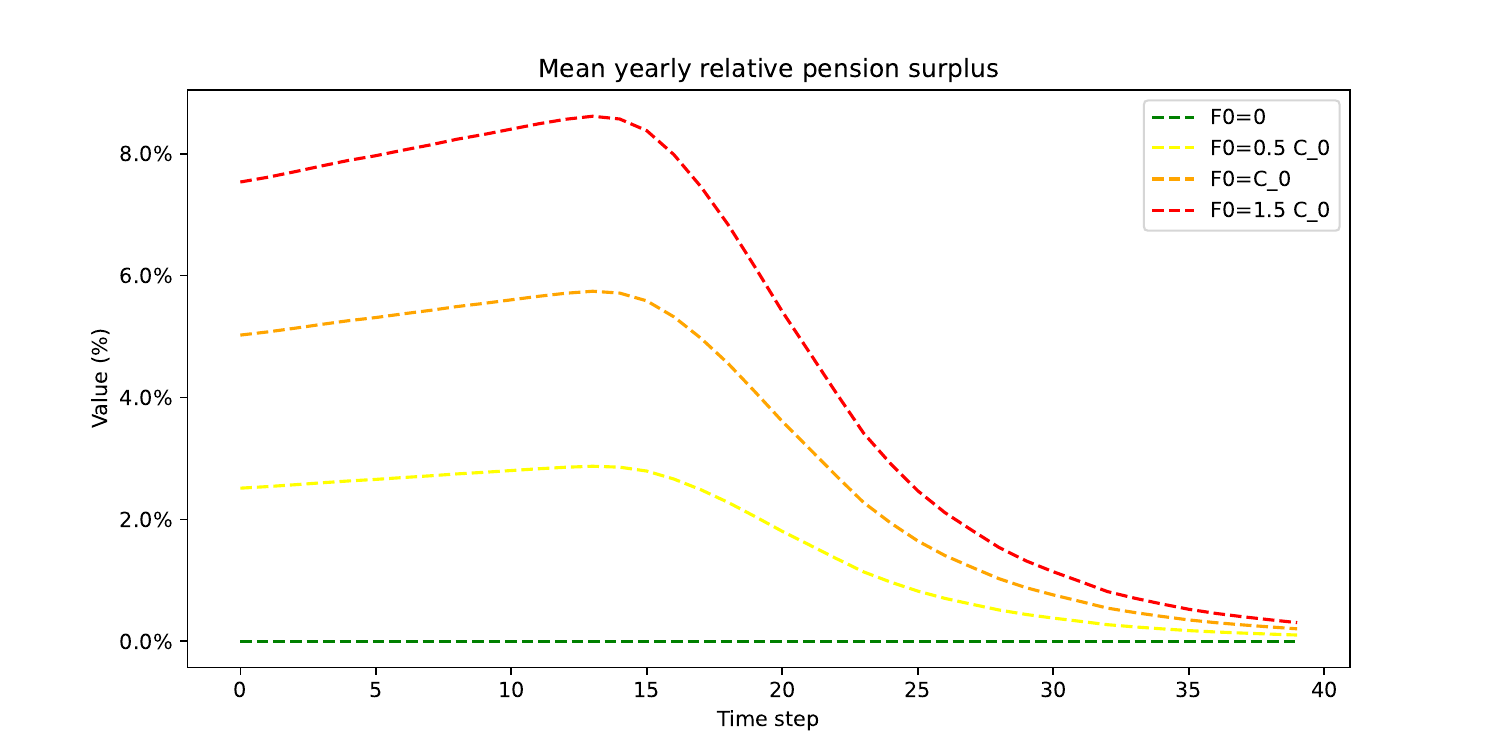}
\vspace{-0.2cm}\caption{Relative surplus $\rho_t=\frac{p^*_t-p^{\min}_t}{p^{\min}_t}$}\label{subfig:BB_sens_F0_three_wages_rho_0.03}
\end{subfigure}
\vspace{0.3cm}
\begin{subfigure}[t]{0.49\textwidth}
\centering
\includegraphics[width=\textwidth,height=0.7\textheight,keepaspectratio]{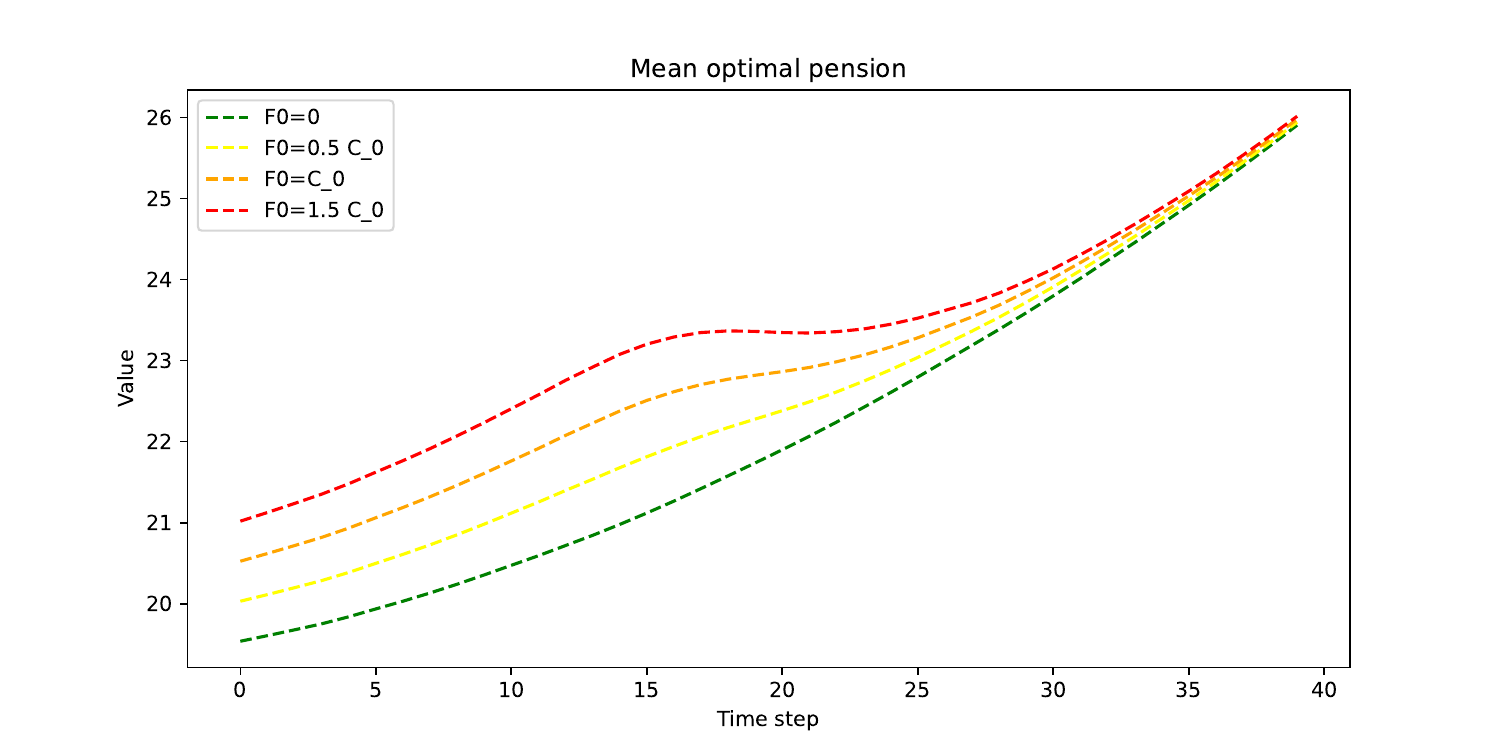}
\vspace{-0.2cm}\caption{Mean total pension $p^*$ versus $F_0$}\label{subfig:BB_sens_F0_three_wages_p*_0.02}
\end{subfigure}\hfill
\begin{subfigure}[t]{0.49\textwidth}
\centering
\includegraphics[width=\textwidth,height=0.7\textheight,keepaspectratio]{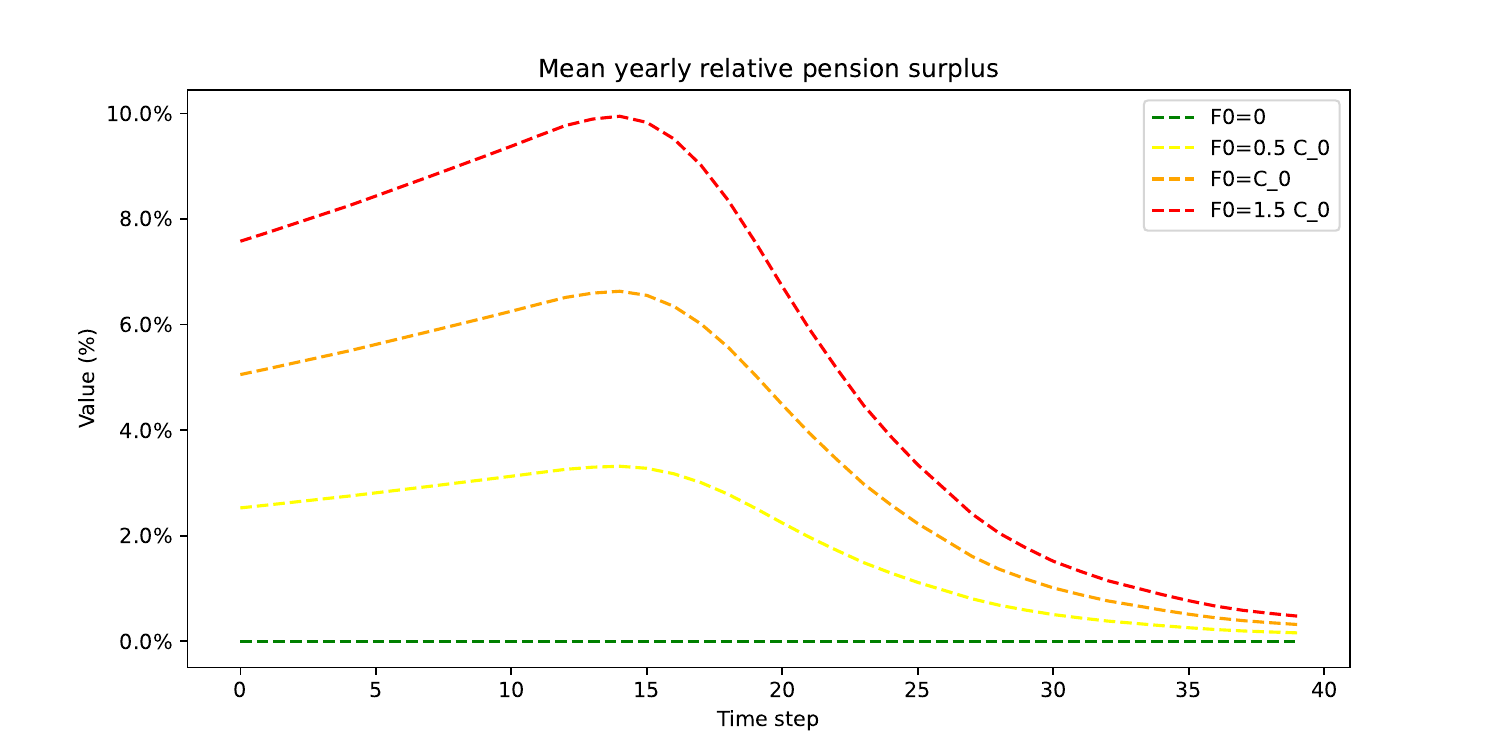}
\vspace{-0.2cm}\caption{Relative surplus $\rho_t=\frac{p^*_t-p^{\min}_t}{p^{\min}_t}$}\label{subfig:BB_sens_F0_three_wages_rho_0.02}
\end{subfigure}
\vspace{0.3cm}
\begin{subfigure}[t]{0.49\textwidth}
\centering
\includegraphics[width=\textwidth,height=0.7\textheight,keepaspectratio]{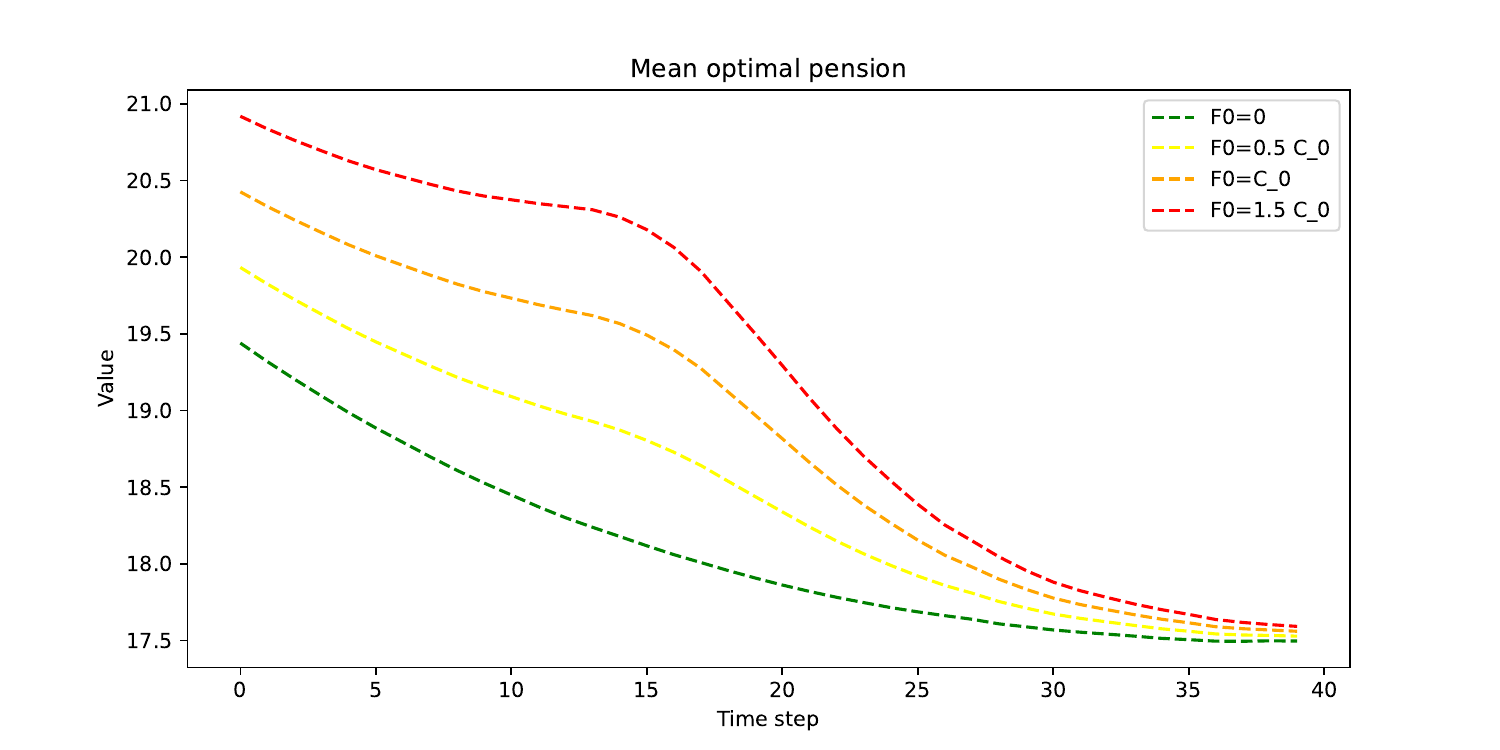}
\vspace{-0.2cm}\caption{Mean total pension $p^*$ versus $F_0$}\label{subfig:BB_sens_F0_three_wages_p*_0.01}
\end{subfigure}\hfill
\begin{subfigure}[t]{0.49\textwidth}
\centering
\includegraphics[width=\textwidth,height=0.7\textheight,keepaspectratio]{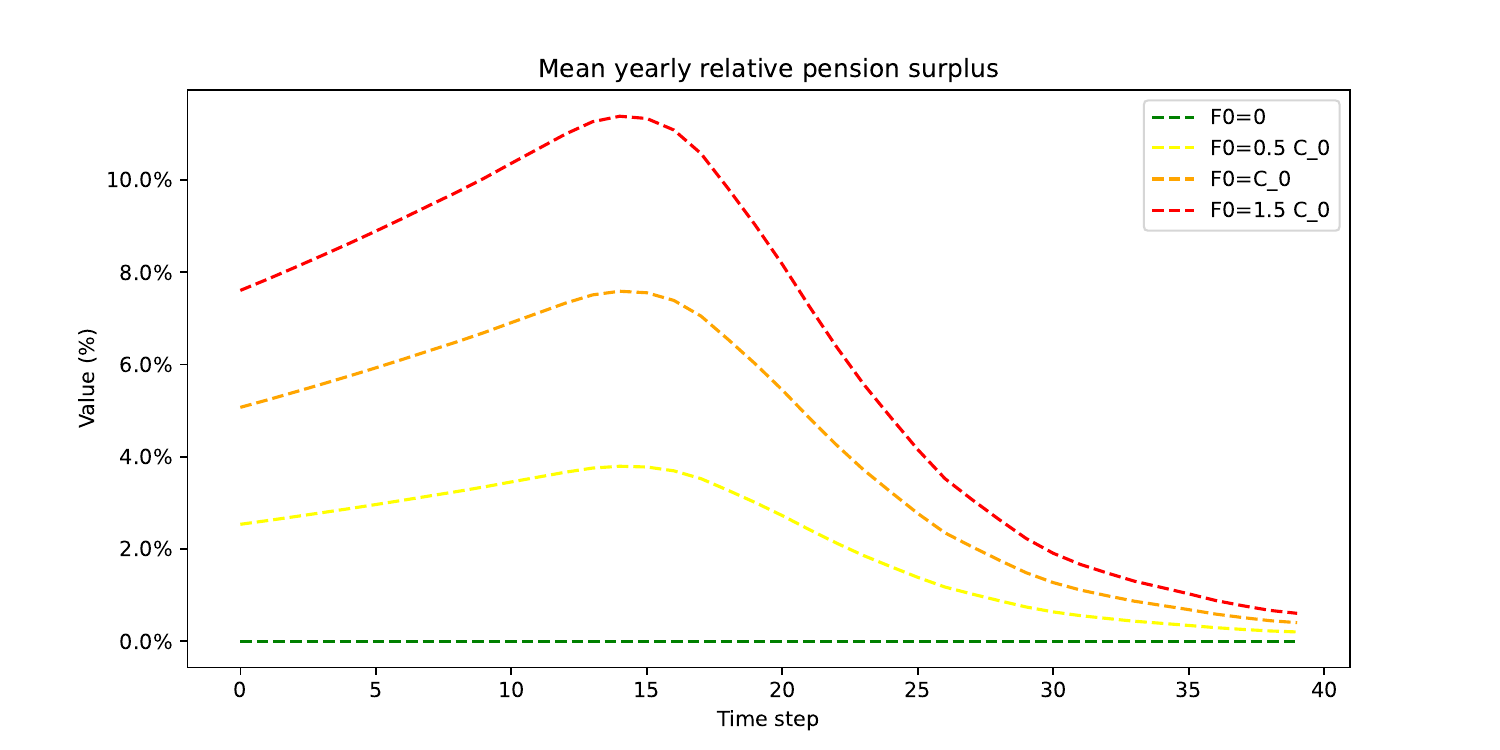}
\vspace{-0.2cm}\caption{Relative surplus $\rho_t=\frac{p^*_t-p^{\min}_t}{p^{\min}_t}$}\label{subfig:BB_sens_F0_three_wages_rho_0.01}
\end{subfigure}
\vspace{0.3cm}
\captionsetup{font=footnotesize}
\caption{Adequacy for $F_0\in \{0,1/2,1,1.5\}\cdot C_0$. $\lambda=0.03$ (top), $\lambda=0.02$ (mid, base case) and $\lambda=0.01$ (bottom)}\label{fig:BB_sens_F0_three_wages}
\end{figure}

\clearpage 
\newpage

\setcounter{table}{0}
\setcounter{figure}{0}
\setcounter{equation}{0}

\section{Sensitivity to Risk Aversion $\theta$}\label{app:theta}

\begin{figure}[h!]
\centering
	\setfolder{fig_26feb_delta_base_sens_theta}
\begin{subfigure}[t]{0.49\textwidth}
\centering
\includegraphics[width=\textwidth,height=0.65\textheight,keepaspectratio]{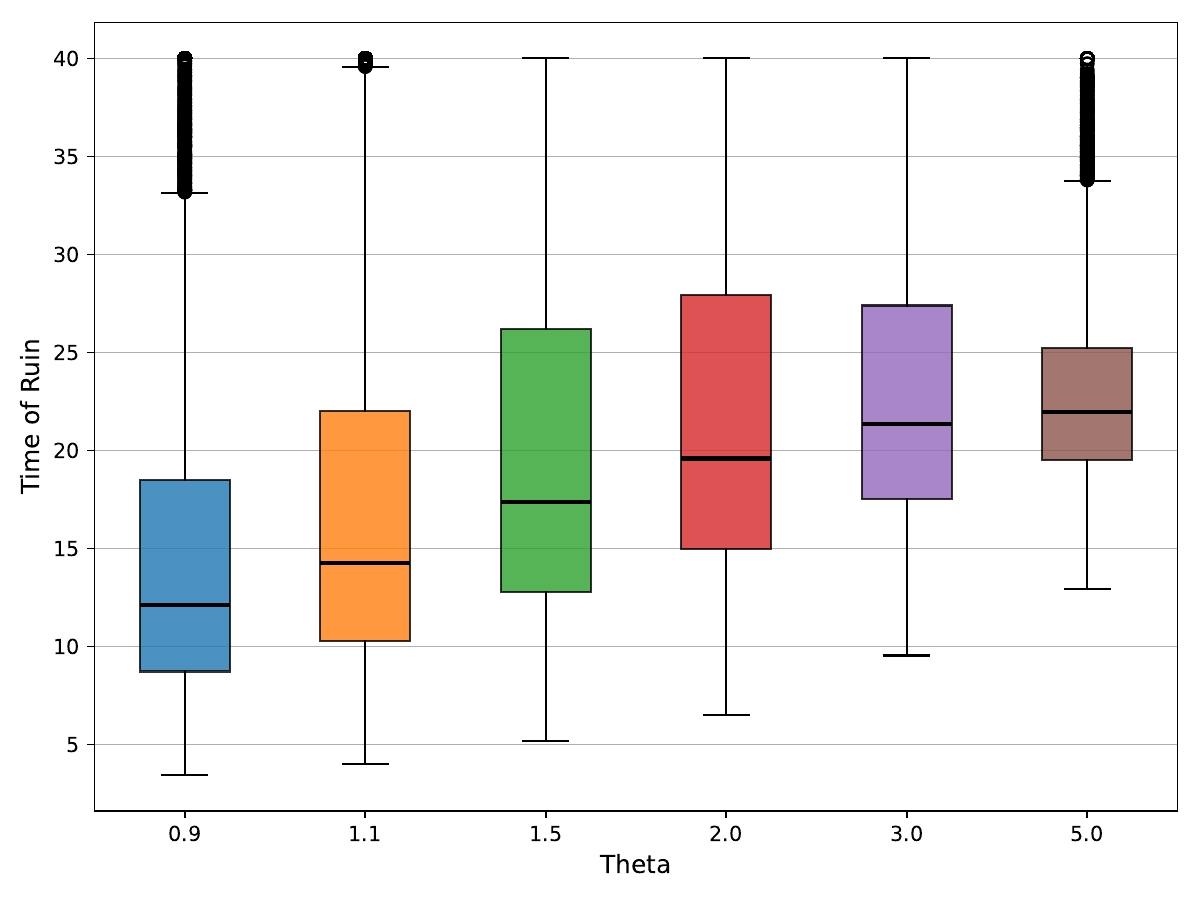}
\vspace{-0.3cm}\caption{Boxplot of depletion time of the fund $\tau$}\label{subfig:bb_delta_base_sens_theta_tau_boxplot}
\end{subfigure}
\begin{subfigure}[t]{0.49\textwidth}
\centering
\includegraphics[width=\textwidth,height=0.65\textheight,keepaspectratio]{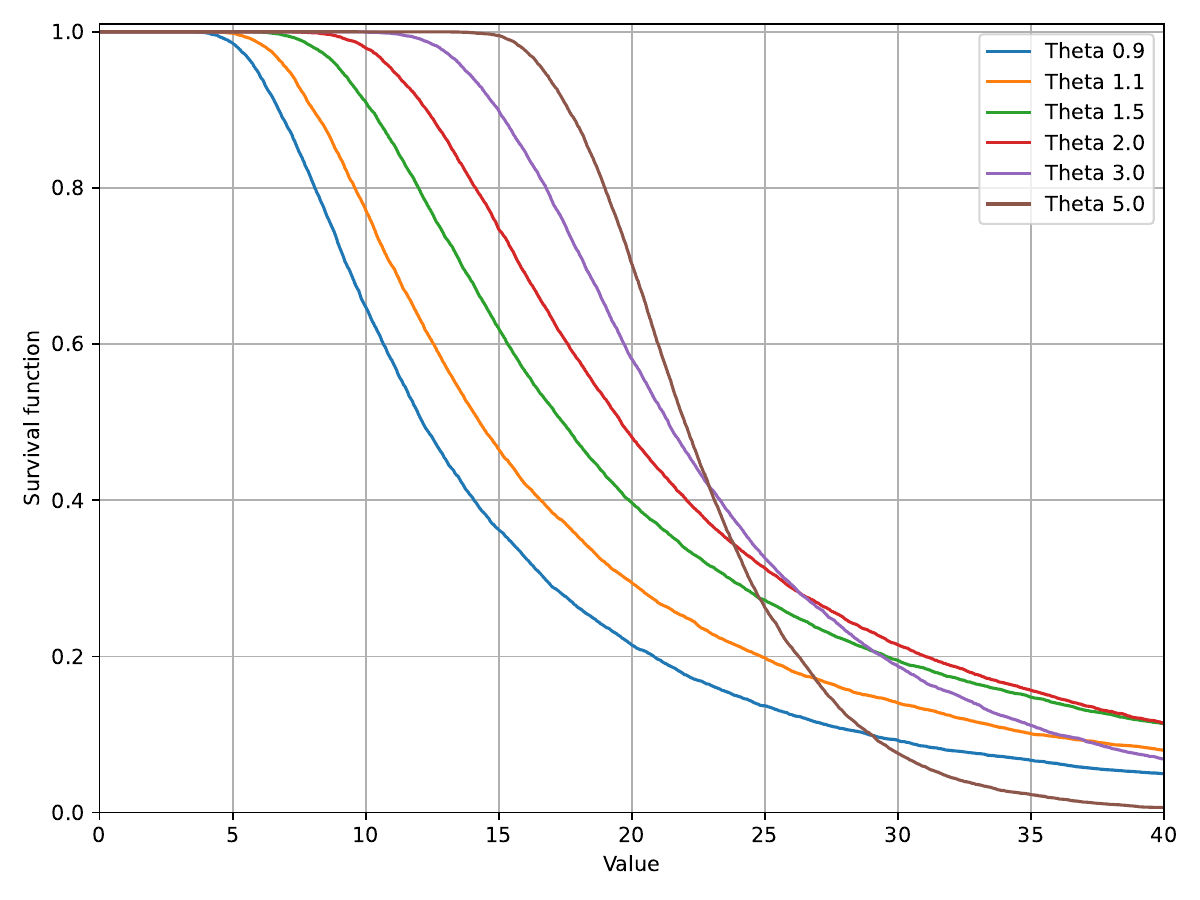}
\vspace{-0.3cm}\caption{Empirical cdf depletion time $\tau$}\label{subfig:bb_delta_base_sens_theta_tau_empirical_cdf}
\end{subfigure}
\captionsetup{font=footnotesize}
\caption{Distribution of depletion time $\tau$ for various $\theta$}\label{fig:bb_delta_base_sens_theta_tau}
\end{figure}
\begin{figure}[h!]
\centering
\setfolder{fig_26feb_delta_base_sens_theta}
\begin{subfigure}[t]{0.49\textwidth}
\centering
\includegraphics[width=\textwidth,height=0.65\textheight,keepaspectratio]{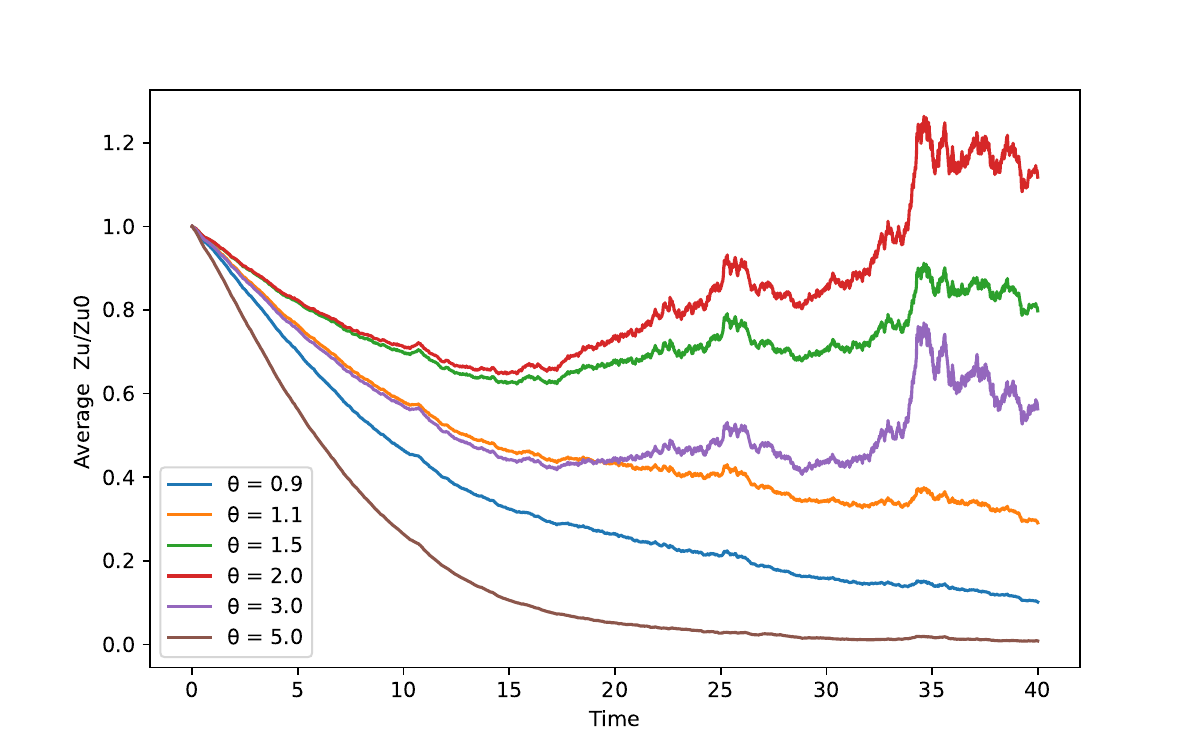}
\vspace{-0.2cm}\caption{Average buffer fund utility weight $Z^u_t$}\label{subfig:bb_delta_base_sens_theta_Zu}
\end{subfigure}
\begin{subfigure}[t]{0.49\textwidth}
\centering
\includegraphics[width=\textwidth,height=0.65\textheight,keepaspectratio]{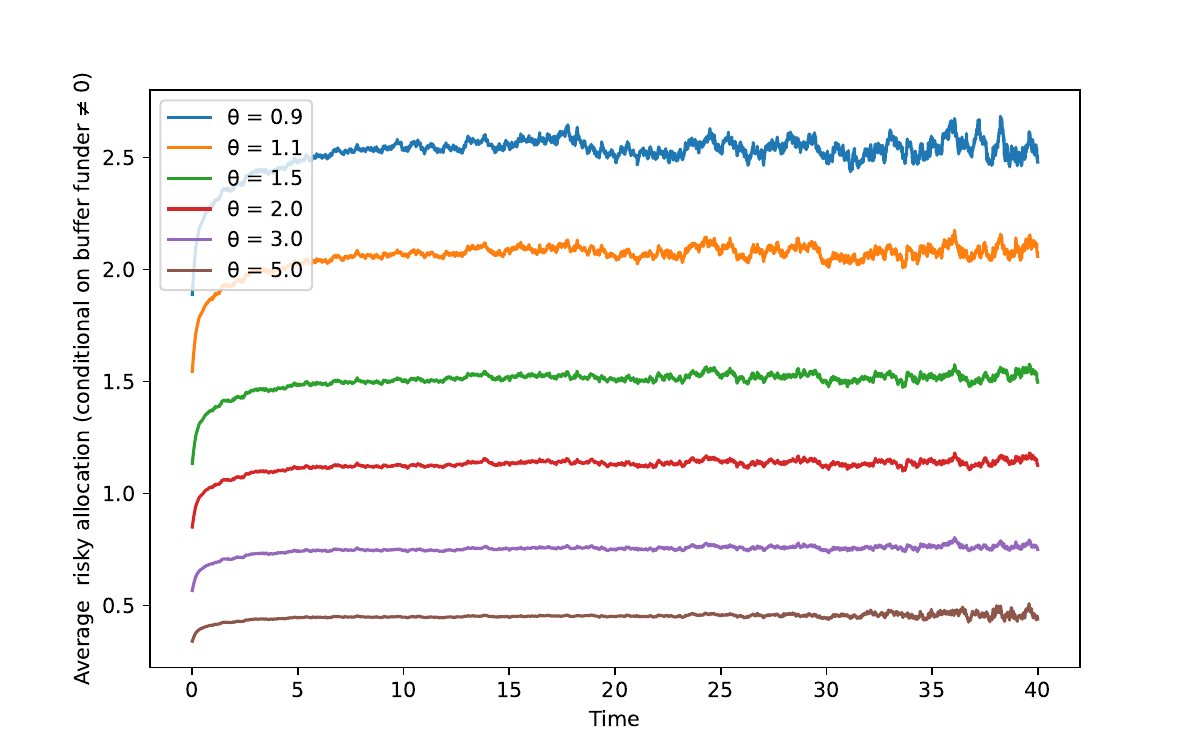}
\vspace{-0.2cm}\caption{Proportion risky investment $\frac{\left(\frac{\pi^*_t}{\sqrt{\nu_t}}\right)}{F^*_t}|F^*>0$}\label{subfig:bb_delta_base_sens_theta_strategy}
\end{subfigure}
\vspace{0.3cm}
\begin{subfigure}[t]{0.49\textwidth}
\centering
\includegraphics[width=\textwidth,height=0.65\textheight,keepaspectratio]{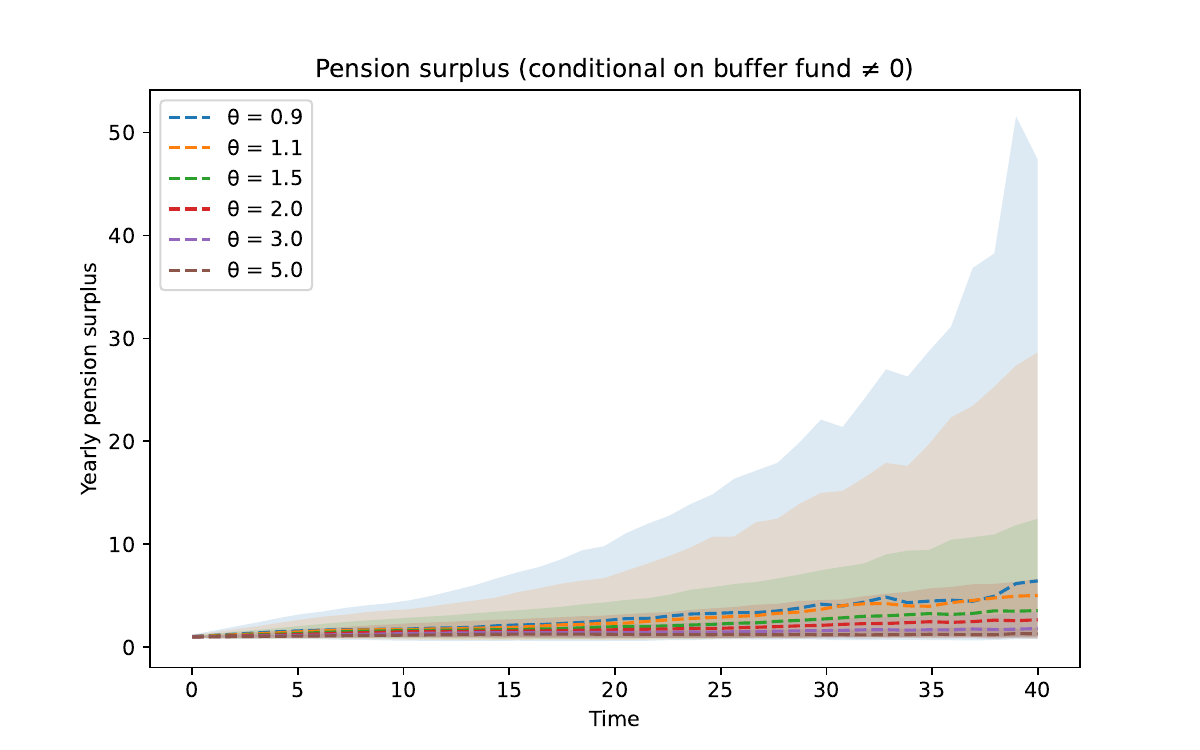}
\vspace{-0.2cm}\caption{Pension surplus $p^* - p^{\min}|F^*_t>0$}\label{subfig:bb_delta_base_sens_theta_surplus}
\end{subfigure}
\begin{subfigure}[t]{0.49\textwidth}
\centering
\includegraphics[width=\textwidth,height=0.65\textheight,keepaspectratio]{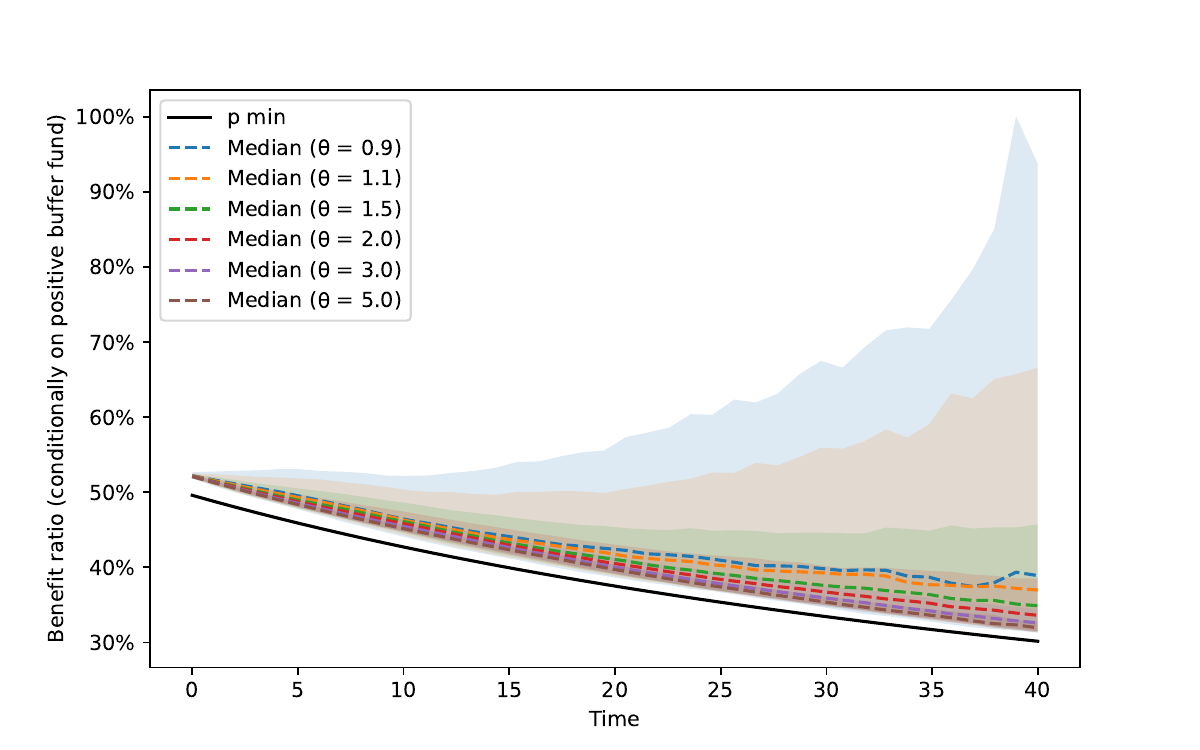}
\vspace{-0.2cm}\caption{Benefit ratio $\text{BR}_t=\frac{p^*_t}{\mathfrak{e}_t}|F^*_t>0$}\label{subfig:bb_delta_base_sens_theta_benefit_ratio}
\end{subfigure}
\captionsetup{font=footnotesize}
\caption{Comparison for various $\theta$ values – Monte Carlo}\label{fig:bb_delta_base_sens_theta}
\end{figure}

Figures~\ref{fig:bb_delta_base_sens_theta_tau} and~\ref{fig:bb_delta_base_sens_theta} present sensitivity to the risk aversion parameter $\theta$ for our base case $Z^u_0=Z_0\left(\frac{F_0 - \mathfrak{K}_0}{p_0 - p^{\min}_0}\right)^{\theta}=Z_0(20\cdot N^r_0)^\theta$. Different values of $\theta$ naturally yield different levels of $Z^u_0$ and $Z^u_t$. Rather than comparing outcomes for a fixed $Z^u_0$, which would be neither scalable nor representative of the same policy objectives, we fix an initial pension target of 5\% above $p^{\min}_0$ for all values of $\theta$. Figure~\ref{subfig:bb_delta_base_sens_theta_Zu} confirms that all cases share the same initial objective in terms of $Z^u_t/Z^u_0$; however, through adjusted portfolio management and pension payments, the subsequent evolution of $Z^u_t$, and consequently $\tau$, is heavily dependent on $\theta$.

Figure~\ref{fig:bb_delta_base_sens_theta_tau} reveals a close relationship between $\theta$ and $\tau$. The survival function shows that higher $\theta$ narrows the distribution of ruin times. Lower $\theta$ yields higher pensions conditional on fund survival, but ruin is more likely in that case, potentially reducing the higher pension income. Higher $\theta$ reduces pensions through multiple channels, including reduced allocation to risky assets and consequently lower financial returns on the buffer fund.

Figure~\ref{subfig:bb_delta_base_sens_theta_strategy} shows that for $\theta<1.5$ the risky allocations are on average systematically above 1, implying short-selling, whereas $\theta>1.5$ yields portfolio allocations that neither short-sell nor borrow on average, with $\theta=4$ providing roughly a 50\% stock and risk-free rate allocation, as shown in Figure~\ref{subfig:SS_v_BB_MC_proportion}. Despite identical initial pension surplus objectives, the actual surplus paid (Figure~\ref{subfig:bb_delta_base_sens_theta_surplus}) decreases with $\theta$ as discussed after Theorem \ref{ThOpt}, driven by lower fund variability and reduced leverage through lower risky allocations. Conditional on $F^*>0$, however, meaningful surpluses are provided that mitigate the baby boom relative decrease in pensions (Figure~\ref{subfig:bb_delta_base_sens_theta_benefit_ratio}). As expected, very low $\theta$ generates extreme variability in the benefit ratio due to excessive leverage. For our baseline range $\theta>1.5$, the upper end of the distribution suggests a benefit ratio close to 50\% is achievable in favorable scenarios, realistically raising the terminal benefit ratio from approximately 30\% to 40\%.

\end{document}